\documentclass[a4paper,11pt]{article}
\pdfoutput=1 

\usepackage{jheppub} 
\usepackage{tikz}
\usepackage{tikz-cd}

\newcommand{\cA}{\mathcal A}

\newcommand{\bea}{\begin{eqnarray}}
\newcommand{\eea}{\end{eqnarray}}

\usepackage[utf8]{inputenc}
\usepackage{hyperref}
\usepackage{amsfonts}
\usepackage{amsmath}
\usepackage{amssymb}
\usepackage{mathrsfs}  
\usepackage{color}
\usepackage{physics}
\usepackage{multicol}
\usepackage{tikz}

\usepackage{amsmath,bbm,array,amsfonts,graphicx,wrapfig,lscape,float,mathtools,multirow,longtable,amsthm}
\usepackage{stackrel}
\usepackage[all]{xy}
\usepackage{caption}
\usepackage{subcaption}
\usepackage{multirow}
\usepackage[bottom]{footmisc}
\usepackage{mathrsfs}
\usepackage{tikz}
\usepackage{bm}
\usepackage{color}
\usepackage{enumerate}
\usetikzlibrary{decorations.pathreplacing}
\usepackage{diagbox}
\numberwithin{equation}{section}
\numberwithin{figure}{section}
\numberwithin{table}{section}
\usepackage{pgfplots}
\pgfplotsset{compat=1.14}
\usepackage{makecell}
\usepackage{lscape}
\usepackage[utf8]{inputenc}
\usepackage{mathtools}
\usepackage{todonotes}
\usepackage{epigraph}

\usepackage[autostyle=true]{csquotes}
\newtheorem{definition}{Definition}[section]
\newtheorem{theorem}{Theorem}[section]
\newtheorem{corollary}{Corollary}[theorem]
\newtheorem{lemma}[theorem]{Lemma}
\newtheorem{proposition}[theorem]{Proposition}
\newtheorem{conjecture}[theorem]{Conjecture} 
\newtheorem{remark}{Remark}
\newtheorem{example}{Example}

\usepackage[toc,page]{appendix}
\usepackage[inline]{enumitem}

\newcommand{\commentout}[1]{}

\title{Mahler Measure for a Quiver Symphony}

\author[a,b]{Jiakang Bao,}
\author[b,a,c,d]{Yang-Hui He,}
\author[e,b]{Ali Zahabi}

	\affiliation[a]{
		Department of Mathematics, City, University of London, EC1V 0HB, UK}
	\affiliation[b]{
	    London Institute for Mathematical Sciences, Royal Institution, London W1S 4BS, UK}
	\affiliation[c]{
		Merton College, University of Oxford, OX1 4JD, UK}
	\affiliation[d]{
		School of Physics, NanKai University, Tianjin, 300071, P.R. China}

    \affiliation[e]{Institut de Math\'ematiques de Bourgogne, Universit\'e Bourgogne Franche-Comt\'e, France}

	\emailAdd{jiakang.bao@city.ac.uk}
	\emailAdd{hey@maths.ox.ac.uk}
    \emailAdd{zahabi.ali@gmail.com}

\preprint{
		\begin{flushright}
			LIMS-2021-010
		\end{flushright}
	}

\abstract{Adopting the Mahler measure from number theory, we introduce it to toric quiver gauge theories, and study some of its salient features and physical implications. We propose that the Mahler measure is a universal measure for the quiver, encoding its dynamics with the monotonic behaviour along a so-called Mahler flow including two special points at isoradial and tropical limits. Along the flow, the amoeba, from tropical geometry, provides geometric interpretations for the dynamics of the quiver. In the isoradial limit, the maximization of Mahler measure is shown to be equivalent to $a$-maximization. The Mahler measure and its derivative are closely related to the master space, leading to the property that the specular duals have the same functions as coefficients in their expansions, hinting the emergence of a free theory in the tropical limit. Moreover, they indicate the existence of phase transition. We also find that the Mahler measure should be invariant under Seiberg duality.
}

\begin{document} 
\maketitle

\newpage

\epigraph{Die Symphonie mu\ss~sein wie die Welt. Sie mu\ss~alles umfassen.}{Gustav Mahler}

\section{Introduction and Summary}\label{intro}
Quiver gauge theories and their string theory constructions \cite{Nakajima:1994nid,king1994moduli,Douglas:1996sw}, as an extension to the AdS/CFT correspondence, have been a centre of intensive studies over the last two decades. While their holographic realization as brane tilings \cite{Hanany:2005ve,Franco:2005rj,Franco:2005sm,Feng:2005gw} turn out to be extremely useful and productive (see reviews in \cite{Yamazaki:2008bt,Hanany:2010zza,He:2018jtw}), these combinatorial and geometrical  constructions are somewhat limited to a kinematic level, meaning that other than the famous $a$-maximization/volume minimization \cite{Intriligator:2003jj,Gubser:1998vd,Butti:2005vn,Martelli:2005tp}, such key quantities as partition functions etc., have been relatively untouched.
The purpose of the present work is to take the first step in a new direction which can be seen as an initial upgrade of the brane tilings paradigm to the dynamical level.

In order to extend brane tilings and explain the underlying structure of their quiver gauge theories such as the symmetries and dynamics, it is natural to think of the best candidates from statistical dimer models and its cousin, the crystal melting model \cite{Iqbal:2003ds,Okounkov:2003sp,Ooguri:2009ijd,Ooguri:2009ri}, as well as  their probabilistic and geometric aspects \cite{Kenyon:2003uj}. 
Bearing these in mind, the key role in our new picture is played by a universal measure, called the \emph{Mahler measure}, first introduced by Kurt Mahler \cite{mahler1962some}. 
While the original definition arose in number theory \cite{mahler1962some,boyd2002mahler,boyd2002mahlerII}, the Mahler measure can be seen as the scaled height function of the dimer model in the thermodynamic limit \cite{Zahabi:2018znq,Zahabi:2019kdm,Zahabi:2020hwu}, and is closely related to the limit shape of the crystal melting model and the Ronkin function. This measure controls the thermodynamics of the lattice models, through the computation of the growth rate and/or free energy of the model.

Our goal in this paper is to uncover the essence of the Mahler measure in quiver gauge theories, and address to what extent the dynamical aspects of the gauge theory are encoded therein and be unfolded via its geometry and analysis.
The heart of this study and our main ideas are around an observation which leads to our central theorem which we call the {\it Mahler flow} (\textbf{Equation \eqref{mfloweqn}} and \textbf{Proposition \ref{mflowprop}}): a monotonic decrease of the Mahler measure along the path from the tropical limit to the isoradial limit of the theory. This observation bears an overwhelming resemblance to renormalization group (RG) flow, as well as the $c$-theorem in the context of two-dimensional conformal field theories (CFTs). In fact, the Mahler measure is expected to show some unique and universal features along the flow, respecting dualities of the underlying gauge theories, including Seiberg/toric duality \cite{Seiberg:1994pq,Feng:2001bn} (\textbf{Conjecture \ref{Seibergconj}}) and specular duality \cite{Hanany:2012vc} (\textbf{Corollary \ref{speccor}}).

The richness of the Mahler measure comes with expected features as well as some physical and mathematical surprises: (1) the tropical geometry picture of the Mahler flow can be understood in terms of the amoeba and its bounded complements. More precisely, if we think of the Mahler measure as some sort of ``coupling'' changing along the flow, then the scale is the K\"ahler parameter controlling the size of the hole of the amoeba;
(2) the equivalence of the maximization of Mahler measure in the isoradial limit and $a$-maximization (\textbf{Proposition \ref{propmax}});
(3) its r\^ole in the exploration of the phase structure of the theory, as a parameter that points out the phase transition via the analysis of its derivative (\textbf{Propsition \ref{transitionprop}} and \textbf{Corollary \ref{transitioncor}}).

The novelty of the Mahler measure in the context of quiver gauge theories signals possible new directions of studies, and we are only at the beginning. The most fascinating one would be the universality of this measure as it connects a wide range of topics from the superconformal index to the entropy of black holes.

The paper is organized as follows. In \S\ref{background}, we review some background materials on Mahle measure, dimer models, crystal melting models, and tropical geometry. In \S\ref{mahlerquiver}, which is the heart of this article, we present the main new ideas and results about the applications and implications of Mahler measure in quiver gauge theory. In \S\ref{outlook}, possible future research directions are discussed.

\section{Prelude}\label{background}
Since we are attempting to connect a multitude of concepts from mathematics and physics, it is expedient to present an introductory summary here, as much to motivate the reader as to set notation. Let us start with the Mahler measure. Then we will see how this is connected to dimer models/brane tilings.

\subsection{The Mahler Measure}\label{mahler}
Originating in algebraic number theory, the Mahler measure \cite{mahler1962some} is a seemingly innocuous object. Given a Laurent polynomial in $n$ complex variables, the Mahler measure can be considered as an average on the $n$-torus:
\begin{definition}
For a (non-zero) Laurent polynomial $P(\bm{z})=P(z_1, \ldots, z_n) \in \mathbb{C}[z_1^{\pm1}, \ldots, z_n^{\pm1}]$, the {\bf Mahler measure} is
\begin{equation}
    m(P):=\int_0^1\dots\int_0^1\log|P(\exp(2\pi i\theta_1),\dots,\exp(2\pi i\theta_n))|\textup{d}\theta_1\dots\textup{d}\theta_n.\label{mahlerdef}
\end{equation}
By convention, we set $m(0)=\infty$. 
\end{definition}
We emphasize that the name Mahler measure often means the exponential $\exp(m(P))$ in the literature. However, since we will exclusively work with $m(P)$ in this paper, we will always refer to $m(P)$ in
\eqref{mahlerdef}
as the Mahler measure.

The Mahler measure enjoys many salient features, such as additivity, meaning that $m(PQ)=m(P)+m(Q)$ for any two Laurent polynomials $P,Q$.
Furthermore, for a univariate polynomial $P(z)=a \prod\limits_{i=1}^n(z-\alpha_i)$, we have Jensen's formula:
\begin{equation}
    m(P)=\log|a|+\sum_{i=1}^n\max\{0,\log|\alpha_i|\} \ .
\end{equation}
However, for more than one variable, the integral in \eqref{mahlerdef} is already highly non-trivial. Since in this paper we are mainly considering bivariate Laurent polynomial $P(z,w)$, and already no such simple formula as Jensen's is known\footnote{By writing $P(z,w)$ as $a(w)\prod\limits_{i=1}^n(z-\alpha_i(w))$, we can still use Jensen's formula to compute $m(P(z,w))$, but the expression is much more involved \cite{vandervelde2008mahler} and no analytic results are known explicitly.}.

What we do know about the $n$ variable case is that the Mahler measure is GL$(n,\mathbb{Z})$ invariant \cite{schinzel_2000,boyd2005small}:
\begin{theorem}
    Let $0\neq (M_{ij})_{n\times n}\in\textup{GL}(n,\mathbb{Z})$. Then
    \begin{equation}
        m(P(\bm{z}))
        =m(P(\bm{z}^M))
        =m(P(-\bm{z}^M)),
    \end{equation}
    where $\bm{z}^M=\left(\prod\limits_{i=1}^nz_i^{M_{i1}},\dots,\prod\limits_{i=1}^nz_i^{M_{in}}\right)$.
\end{theorem}
Other than the few nice properties mentioned above, perhaps the most extraordinary about the Mahler measure is that for certain polynomials, it evaluates to special values of zeta and $L$-functions \cite{villegas1999modular,boyd2002mahler,boyd2002mahlerII}.

In our paper, we will mainly apply the expansion of the integrand and residue theorem to calculate the integral. Writing $\bm{z} := (z_1, \ldots, z_n)$, and extract the constant term of the polynomial $P$ as $k$, i.e.,
\begin{equation}\label{Pkp}
    P(\bm{z}):= k - p(\bm{z}) \ , 
\end{equation}
we have the series expansion (formally in $p$)
\begin{equation}
    \log(k-p(\bm{z}))=\log k-\sum_{n=1}^\infty\frac{p^n(\bm{z})}{n}k^{-n}.
    \label{logexpansion}
\end{equation}
Since we are integrating on the $n$-torus, we need the restriction $|k|\geq\max\limits_{\bm{z}\in T^n}|p(\bm{z})|$. This then ensures the series expansion converges uniformly on the support of the integration path and hence we are also allowed to exchange the sum and integral in our calculation. Therefore, we may write the Mahler measure as
\begin{equation}
    m(P)=\textup{Re}\left(\frac{1}{(2\pi i)^n}\int_{|z_i|=1}\log(P(z_1,\dots,z_n))\frac{\textup{d}z_1}{z_1}\dots\frac{\textup{d}z_n}{z_n}\right).\label{mahlerre}
\end{equation}
By the residue theorem and integrating over $|k|>\max\limits_{\bm{z}\in T^n}|p(\bm{z})|$, only the constant term in \eqref{logexpansion} contributes. 
Therefore, we have that 
\begin{equation}
    m(P)=\textup{Re}\left(\log k-\int_0^{k^{-1}}(u_0(t)-1)\frac{\text{d}t}{t}\right) \ ,
\end{equation}
where
\begin{equation}
    u_0(k)=\frac{1}{(2\pi i)^n}\int_{|z_i|=1}\frac{1}{1-k^{-1}p(z_1,\dots,z_n)}\frac{\textup{d}z_1}{z_1}\dots\frac{\textup{d}z_n}{z_n} \ .
    \label{u0integral}
\end{equation}
For a concrete example of the details, see Appendix \ref{exF0}.

In fact, \eqref{u0integral} means that $u_0(k)$ is a period of a holomorphic 1-form $\omega_Y$ on the curve $Y$ defined by $1-k^{-1}P=0$ \cite{griffiths1969periods,villegas1999modular}, that is, $\int_{\gamma}\omega_Y$ where $\gamma$ is a 1-cycle on $Y$. We shall also refer to this as the period of the curve $Y$. Therefore, $u_0$ also satisfies the Picard-Fuchs equation:
\begin{equation}
    A(k)\frac{\text{d}^2u_0}{\text{d}k^2}+B(k)\frac{\text{d}u_0}{\text{d}k}+C(k)u_0=0,
\end{equation}
where $A(k),B(k),C(k)$ are polynomials in $k$.

We may also extend the definition of Mahler measure.
\begin{definition}
The generalized Mahler measure extends definition \eqref{mahlerdef} to
an arbitrary torus with variable sizes $a_i$:
\bea
\label{genMahler}
m(P;a_i)= \frac{1}{(2\pi i)^n}\int_{|z_i|=a_i}\log P(z_1,\dots,z_n)\frac{\textup{d}z_1}{z_1}\dots\frac{\textup{d}z_n}{z_n}.
\eea
\end{definition}

\subsection{Dimer Models}\label{dimer}
To discuss how $P(z,w)$ and Mahler measures are related to dimer models \cite{Kenyon:2003uj}\footnote{See also \cite{Stienstra:2005wz,Stienstra:2005ns} for relevant discussions on Mahler measures and dimers, including corresponding densities on the torus as well as counting paths in dimers.}, we shall first give a quick review on dimers. Let $\mathcal{G}$ be a bipartite graph, i.e., one whose vertices can be partitioned into two separate sets, which we can colour as black and white, so that each edge connects one black and one white node. We emphasize that our graphs are simple (no multiple edges between nodes) and undirected.

A {\bf perfect matching} of $\mathcal{G}$ is a collection of edges such that each vertex is incident to exactly one edge. The {\bf dimer model} is then the study of (random) perfect matchings of $\mathcal{G}$ \cite{Kenyon:2003uj,kenyon2003introduction}. Physically, the dimer models have a nice interpretation in terms of brane systems. Hence, they are also known as {\bf brane tilings} \cite{Hanany:2005ve,Franco:2005rj,Franco:2005sm,Feng:2005gw}.
Moreover, we shall always take $\mathcal{G}$ to be $\mathbb{Z}^2$-periodic, so that it constitutes a doubly-periodic tiling of the plane.
In other words, $\mathcal{G}$ is embedded in the $\mathbb{Z}^2$ lattice\footnote{More generally, one may also consider any 2-dimensional lattice instead of $\mathbb{Z}^2$. Indeed, one can consider high-genus tilings.}. Now, the plane quotiented by $\mathbb{Z}^2$ is a torus, of genus 1, and we will use $\mathcal{G}_1 := \mathcal{G}/(\mathbb{Z}^2)$ to denote the fundamental domain of the bipartite graph. More generally, we use $\mathcal{G}_n$ to denote the quotient $\mathcal{G}/(n\mathbb{Z}^2)$, where $n\mathbb{Z}^2$ is the $n$-times enlarged fundamental domain.

Given a perfect matching $M$, we can define a unit flow $\omega$ that flows by one along each edge in $M$ from white node to black node. 
Consider a reference perfect matching $M_0$ with flow $\omega_0$, and let $\gamma$ be a path from face $f_0$ to $f_1$ in the graph. Then for any matching $q$ with flow $\omega$, the total flux of $\omega-\omega_0$ across $\gamma$ is independent of $\gamma$ and defines a height function of $M$. The difference of height functions of any two perfect matchings is independent of the choice of $M_0$. A perfect matching $M_1$ on the fundamental domain $\mathcal{G}_1$ gives a periodic perfect matching $M$ on $\mathcal{G}$. The {\bf height change} of $M_1$ is defined to be $(h_x,h_y)$ if the horizontal and vertical height changes of $M$ for one period are $h_x$ and $h_y$ respectively.

\begin{example}
A square dimer model is given in Figure \ref{dimerexample}(a). The red square indicates a fundamental domain. In fact, this is one of the toric phases for the zeroth Hirzebruch surface $\mathbb{F}_0$. In Figure \ref{dimerexample}(b), we give two example perfect matchings in the fundamental region. If we choose the green one to be our reference matching, then the blue one would have height change $(h_x,h_y)=(0,-1)$ (following the horizontal and vertical arrows).
\begin{figure}[h]
    \centering
    \begin{subfigure}{4cm}
		\centering
		\includegraphics[width=3cm]{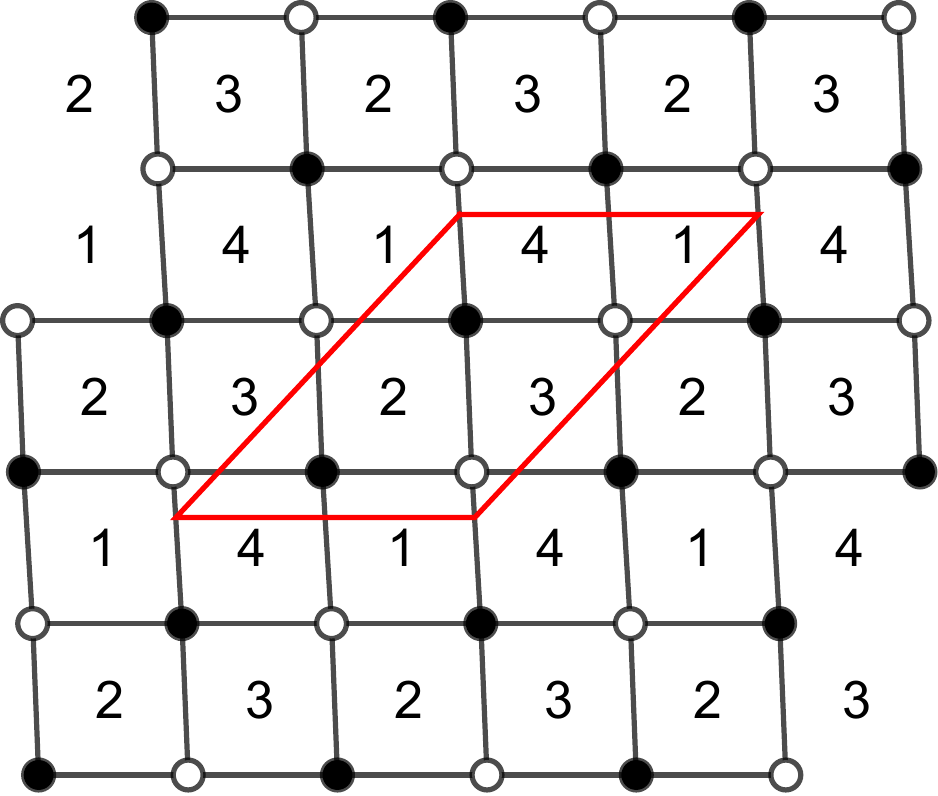}
		\caption{}
	\end{subfigure}
	\begin{subfigure}{4cm}
	\centering
\tikzset{every picture/.style={line width=0.75pt}}
\begin{tikzpicture}[x=0.75pt,y=0.75pt,yscale=-1,xscale=1]
\draw [color={rgb, 255:red, 199; green, 199; blue, 199 }  ,draw opacity=1 ]   (249.07,158.89) .. controls (286.91,146.98) and (300.12,151.75) .. (318.76,158.55) ;
\draw [color={rgb, 255:red, 199; green, 199; blue, 199 }  ,draw opacity=1 ]   (247.68,117.37) .. controls (282.73,104.78) and (295.95,108.18) .. (317.37,117.03) ;
\draw [color={rgb, 255:red, 199; green, 199; blue, 199 }  ,draw opacity=1 ]   (303.32,103.76) .. controls (314.03,102.74) and (314.03,183.06) .. (304.15,172.17) ;
\draw [color={rgb, 255:red, 199; green, 199; blue, 199 }  ,draw opacity=1 ]   (261.59,103.76) .. controls (273.69,101.38) and (271.61,191.22) .. (262.43,172.17) ;
\draw [line width=1.5]    (247.68,117.37) -- (317.37,117.03) ;
\draw [line width=1.5]    (249.07,158.89) -- (318.76,158.55) ;
\draw [line width=1.5]    (261.59,103.76) -- (262.43,172.17) ;
\draw [line width=1.5]    (303.32,103.76) -- (304.15,172.17) ;
\draw [color={rgb, 255:red, 0; green, 255; blue, 0 }  ,draw opacity=1 ][line width=1.5]    (261.83,116.97) -- (262.18,158.96) ;
\draw  [fill={rgb, 255:red, 255; green, 255; blue, 255 }  ,fill opacity=1 ] (265.53,116.65) .. controls (265.53,114.94) and (264.11,113.55) .. (262.36,113.55) .. controls (260.61,113.55) and (259.19,114.94) .. (259.19,116.65) .. controls (259.19,118.37) and (260.61,119.75) .. (262.36,119.75) .. controls (264.11,119.75) and (265.53,118.37) .. (265.53,116.65) -- cycle ;
\draw [color={rgb, 255:red, 0; green, 0; blue, 255 }  ,draw opacity=1 ][line width=1.5]    (262.36,158.86) -- (304.09,158.86) ;
\draw  [fill={rgb, 255:red, 0; green, 0; blue, 0 }  ,fill opacity=1 ] (265.53,158.86) .. controls (265.53,157.14) and (264.11,155.76) .. (262.36,155.76) .. controls (260.61,155.76) and (259.19,157.14) .. (259.19,158.86) .. controls (259.19,160.57) and (260.61,161.96) .. (262.36,161.96) .. controls (264.11,161.96) and (265.53,160.57) .. (265.53,158.86) -- cycle ;
\draw [color={rgb, 255:red, 0; green, 0; blue, 255 }  ,draw opacity=1 ][line width=1.5]    (306.56,117.33) -- (317.37,117.03) ;
\draw [color={rgb, 255:red, 0; green, 0; blue, 255 }  ,draw opacity=1 ][line width=1.5]    (247.68,117.37) -- (258.49,117.07) ;
\draw [color={rgb, 255:red, 0; green, 0; blue, 0 }  ,draw opacity=1 ][fill={rgb, 255:red, 155; green, 155; blue, 155 }  ,fill opacity=1 ]   (244.7,168.9) -- (335.2,168.24) ;
\draw [shift={(337.2,168.22)}, rotate = 539.5799999999999] [color={rgb, 255:red, 0; green, 0; blue, 0 }  ,draw opacity=1 ][line width=0.75]    (6.56,-1.97) .. controls (4.17,-0.84) and (1.99,-0.18) .. (0,0) .. controls (1.99,0.18) and (4.17,0.84) .. (6.56,1.97)   ;
\draw [color={rgb, 255:red, 0; green, 0; blue, 0 }  ,draw opacity=1 ][fill={rgb, 255:red, 155; green, 155; blue, 155 }  ,fill opacity=1 ]   (254.52,176.75) -- (253.84,97.75) ;
\draw [shift={(253.83,95.75)}, rotate = 449.51] [color={rgb, 255:red, 0; green, 0; blue, 0 }  ,draw opacity=1 ][line width=0.75]    (6.56,-1.97) .. controls (4.17,-0.84) and (1.99,-0.18) .. (0,0) .. controls (1.99,0.18) and (4.17,0.84) .. (6.56,1.97)   ;
\draw [color={rgb, 255:red, 0; green, 255; blue, 0 }  ,draw opacity=1 ][line width=1.5]    (303.74,116.86) -- (304.09,158.86) ;
\draw  [fill={rgb, 255:red, 0; green, 0; blue, 0 }  ,fill opacity=1 ] (306.56,117.33) .. controls (306.56,115.62) and (305.14,114.23) .. (303.39,114.23) .. controls (301.64,114.23) and (300.22,115.62) .. (300.22,117.33) .. controls (300.22,119.05) and (301.64,120.43) .. (303.39,120.43) .. controls (305.14,120.43) and (306.56,119.05) .. (306.56,117.33) -- cycle ;
\draw  [fill={rgb, 255:red, 255; green, 255; blue, 255 }  ,fill opacity=1 ] (307.25,158.86) .. controls (307.25,157.14) and (305.84,155.76) .. (304.09,155.76) .. controls (302.34,155.76) and (300.92,157.14) .. (300.92,158.86) .. controls (300.92,160.57) and (302.34,161.96) .. (304.09,161.96) .. controls (305.84,161.96) and (307.25,160.57) .. (307.25,158.86) -- cycle ;
\end{tikzpicture}
\caption{}
	\end{subfigure}
	\caption{An example dimer (a) and two of its perfect matchings (b).}\label{dimerexample}
\end{figure}
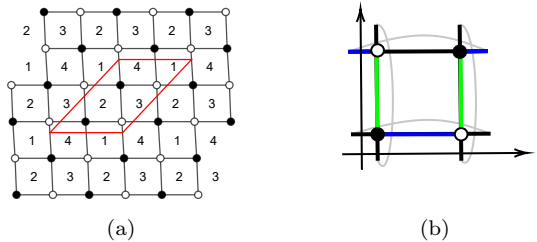
\end{example}

We can define a real function $\mathcal{E}(e)$ on the edges $e$ of $\mathcal{G}$. This is known as the {\bf energy} of the edges.
\begin{definition}
Given a perfect matching (or more generally, any set of edges) $M$, its energy is $\mathcal{E}(M):=\sum\limits_{e\in M}\mathcal{E}(e)$. For any edge $e$ in the graph, its {\bf edge weight} is defined to be $\textup{e}^{-\mathcal{E}(e)}$. Let $\mathcal{M}(\mathcal{G})$ be the set of perfect matchings on $\mathcal{G}$, then the {\bf partition function} of $\mathcal{M}$ is $Z(\mathcal{G}):=\sum\limits_{M\in\mathcal{M}(\mathcal{G})}\textup{e}^{-\mathcal{E}(M)}$.
\end{definition}
Given the edge weights, one can define the {\bf Kasteleyn matrix} $K$.
\begin{definition}
A Kasteleyn matrix has rows (columns) representing the white (black) nodes in $\mathcal{G}$. Its entries are the corresponding edge weights multiplied by $\pm1$ as follows. Around each face there are an odd number of edge weights multiplied by $-1$ if the face has $0~(\textup{mod }4)$ edges and an even number if it has $2~(\textup{mod }4)$ edges.
\end{definition}
It was shown in \cite{Kasteleyn1967} that this construction is always possible, and
\begin{theorem}
The absolute value of the determinant is the partition function, that is,
\begin{equation}
    |\det(K)|=Z(\mathcal{G})=\sum\limits_{M\in\mathcal{M}(\mathcal{G})}\textup{e}^{-\mathcal{E}(M)}.
\end{equation}
\end{theorem}

\paragraph{Newton polygons} As $\mathcal{G}_1$ is embedded on a torus, let $\gamma_x$ and $\gamma_y$ be paths winding horizontally and vertically around the torus. Then we can multiply an edge weight by $z$ (or $z^{-1}$) if the $\gamma_x$ crosses the edge with the black node on its left (or right). Likewise, we multiply an edge weight by $w^{\pm1}$ if $\gamma_y$ crosses the edge. This leads to the ``magnetically altered'' Kasteleyn matrix $K(z,w)$ \cite{Kenyon:2003uj}. We may then construct a Laurent polynomial from this.
\begin{definition}
The {\bf Newton} or {\bf charateristic polynomial} of $\mathcal{G}$ is $P(z,w):=\det(K(z,w))$ in formal complex variables $z,w$. It defines a so-called {\bf spectral curve} $P(z,w)=0$, as a Riemann surface.
\end{definition}
For each monomial $c_{(m,n)}z^mw^n$ in $P(z,w)$ with coefficient $c_{(m,n)}$, we can associate a point $(m,n)$ on the lattice. These points form a lattice polygon known as the {\bf Newton polygon}. In toric geometry, every (convex) lattice polygon gives rise to a non-compact toric variety which is a Gorenstein singularity of (complex) dimension 3 \cite{hartshorne1977algebraic,cox2011toric}. Hence, the Newton polygon is also known as the toric diagram. In particular, each vertex/corner point in the polygon is associated with a toric divisor of the Gorenstein singularity.

\paragraph{Quivers} For each consistent brane tiling \cite{Gulotta:2008ef}, we can construct a quiver for the gauge theory from its dual graph (see quick review in \cite{He:2016fnb}). Each face in the tiling corresponds to a unitary gauge node in the quiver diagram. The gauge nodes are connected by arrows representing supermultiplets. These arrows $X_I$ are graph dual to the edges $e_I$ in the tiling. Each white or black node yields a superpotential term $\prod\limits_IX_I$. This term has a positive (or negative) sign if the arrows $X_I$ surround a white (or black) node. In Type IIB brane system, this is the worldvolume theory of a stack of D3-branes probing a Gorenstein singularity $\mathcal{X}$. The Gorenstein singularity is exactly the one encoded by the corresponding toric diagram. There is also a Type IIA picture of the tiling, and we will review it in \S\ref{crystal}.

\begin{example}
Let us again consider the dimer in Figure \ref{dimerexample}, which is reproduced in Figure \ref{dimerexample2}(a).
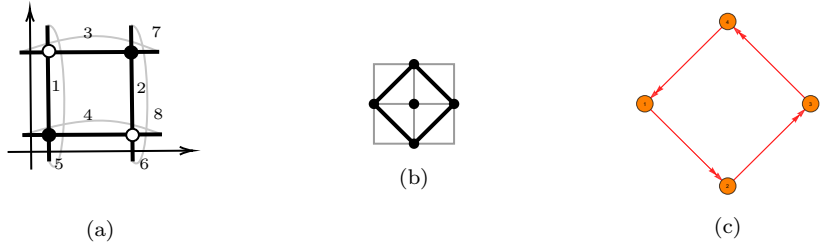
\begin{figure}[h]
    \centering
    \begin{subfigure}{4cm}
		\centering
		\tikzset{every picture/.style={line width=0.75pt}}      
\begin{tikzpicture}[x=0.75pt,y=0.75pt,yscale=-1,xscale=1]
\draw [color={rgb, 255:red, 199; green, 199; blue, 199 }  ,draw opacity=1 ]   (86.07,152.39) .. controls (123.91,140.48) and (137.12,145.25) .. (155.76,152.05) ;
\draw [color={rgb, 255:red, 199; green, 199; blue, 199 }  ,draw opacity=1 ]   (84.68,110.87) .. controls (119.73,98.28) and (132.95,101.68) .. (154.37,110.53) ;
\draw [color={rgb, 255:red, 199; green, 199; blue, 199 }  ,draw opacity=1 ]   (140.32,97.26) .. controls (151.03,96.24) and (151.03,176.56) .. (141.15,165.67) ;
\draw [color={rgb, 255:red, 199; green, 199; blue, 199 }  ,draw opacity=1 ]   (98.59,97.26) .. controls (110.69,94.88) and (108.61,184.72) .. (99.43,165.67) ;
\draw [line width=1.5]    (84.68,110.87) -- (154.37,110.53) ;
\draw [line width=1.5]    (86.07,152.39) -- (155.76,152.05) ;
\draw [line width=1.5]    (98.59,97.26) -- (99.43,165.67) ;
\draw [line width=1.5]    (140.32,97.26) -- (141.15,165.67) ;
\draw  [fill={rgb, 255:red, 255; green, 255; blue, 255 }  ,fill opacity=1 ] (102.53,110.15) .. controls (102.53,108.44) and (101.11,107.05) .. (99.36,107.05) .. controls (97.61,107.05) and (96.19,108.44) .. (96.19,110.15) .. controls (96.19,111.87) and (97.61,113.25) .. (99.36,113.25) .. controls (101.11,113.25) and (102.53,111.87) .. (102.53,110.15) -- cycle ;
\draw  [fill={rgb, 255:red, 255; green, 255; blue, 255 }  ,fill opacity=1 ] (144.25,152.36) .. controls (144.25,150.64) and (142.84,149.26) .. (141.09,149.26) .. controls (139.34,149.26) and (137.92,150.64) .. (137.92,152.36) .. controls (137.92,154.07) and (139.34,155.46) .. (141.09,155.46) .. controls (142.84,155.46) and (144.25,154.07) .. (144.25,152.36) -- cycle ;
\draw  [fill={rgb, 255:red, 0; green, 0; blue, 0 }  ,fill opacity=1 ] (143.56,110.83) .. controls (143.56,109.12) and (142.14,107.73) .. (140.39,107.73) .. controls (138.64,107.73) and (137.22,109.12) .. (137.22,110.83) .. controls (137.22,112.55) and (138.64,113.93) .. (140.39,113.93) .. controls (142.14,113.93) and (143.56,112.55) .. (143.56,110.83) -- cycle ;
\draw  [fill={rgb, 255:red, 0; green, 0; blue, 0 }  ,fill opacity=1 ] (102.53,152.36) .. controls (102.53,150.64) and (101.11,149.26) .. (99.36,149.26) .. controls (97.61,149.26) and (96.19,150.64) .. (96.19,152.36) .. controls (96.19,154.07) and (97.61,155.46) .. (99.36,155.46) .. controls (101.11,155.46) and (102.53,154.07) .. (102.53,152.36) -- cycle ;
\draw [color={rgb, 255:red, 0; green, 0; blue, 0 }  ,draw opacity=1 ][fill={rgb, 255:red, 155; green, 155; blue, 155 }  ,fill opacity=1 ]   (78.7,160.9) -- (169.2,160.24) ;
\draw [shift={(171.2,160.22)}, rotate = 539.5799999999999] [color={rgb, 255:red, 0; green, 0; blue, 0 }  ,draw opacity=1 ][line width=0.75]    (6.56,-1.97) .. controls (4.17,-0.84) and (1.99,-0.18) .. (0,0) .. controls (1.99,0.18) and (4.17,0.84) .. (6.56,1.97)   ;
\draw [color={rgb, 255:red, 0; green, 0; blue, 0 }  ,draw opacity=1 ][fill={rgb, 255:red, 155; green, 155; blue, 155 }  ,fill opacity=1 ]   (90.52,169.75) -- (89.84,90.75) ;
\draw [shift={(89.83,88.75)}, rotate = 449.51] [color={rgb, 255:red, 0; green, 0; blue, 0 }  ,draw opacity=1 ][line width=0.75]    (6.56,-1.97) .. controls (4.17,-0.84) and (1.99,-0.18) .. (0,0) .. controls (1.99,0.18) and (4.17,0.84) .. (6.56,1.97)   ;
\draw (99,123.5) node [anchor=north west][inner sep=0.75pt]  [font=\tiny] [align=left] {{\tiny 1}};
\draw (141.5,124.5) node [anchor=north west][inner sep=0.75pt]  [font=\tiny] [align=left] {2};
\draw (115,97) node [anchor=north west][inner sep=0.75pt]  [font=\tiny] [align=left] {3};
\draw (115,138) node [anchor=north west][inner sep=0.75pt]  [font=\tiny] [align=left] {4};
\draw (100.5,163) node [anchor=north west][inner sep=0.75pt]  [font=\tiny] [align=left] {5};
\draw (143,163) node [anchor=north west][inner sep=0.75pt]  [font=\tiny] [align=left] {6};
\draw (149,96.5) node [anchor=north west][inner sep=0.75pt]  [font=\tiny] [align=left] {7};
\draw (150,137.5) node [anchor=north west][inner sep=0.75pt]  [font=\tiny] [align=left] {8};
\end{tikzpicture}
		\caption{}
	\end{subfigure}
\begin{subfigure}{4cm}
    \centering
    \tikzset{every picture/.style={line width=0.75pt}}
\begin{tikzpicture}[x=0.75pt,y=0.75pt,yscale=-1,xscale=1]
\draw  [draw opacity=0] (259.5,59.7) -- (299.5,59.7) -- (299.5,99.7) -- (259.5,99.7) -- cycle ; \draw  [color={rgb, 255:red, 155; green, 155; blue, 155 }  ,draw opacity=1 ] (279.5,59.7) -- (279.5,99.7) ; \draw  [color={rgb, 255:red, 155; green, 155; blue, 155 }  ,draw opacity=1 ] (259.5,79.7) -- (299.5,79.7) ; \draw  [color={rgb, 255:red, 155; green, 155; blue, 155 }  ,draw opacity=1 ] (259.5,59.7) -- (299.5,59.7) -- (299.5,99.7) -- (259.5,99.7) -- cycle ;
\draw [line width=1.5]    (299.5,79.7) -- (279.5,59.7) ;
\draw [line width=1.5]    (259.5,79.7) -- (279.5,59.7) ;
\draw [line width=1.5]    (299.5,79.7) -- (279.5,99.7) ;
\draw  [fill={rgb, 255:red, 0; green, 0; blue, 0 }  ,fill opacity=1 ] (277.4,99.7) .. controls (277.4,98.54) and (278.34,97.6) .. (279.5,97.6) .. controls (280.66,97.6) and (281.6,98.54) .. (281.6,99.7) .. controls (281.6,100.86) and (280.66,101.8) .. (279.5,101.8) .. controls (278.34,101.8) and (277.4,100.86) .. (277.4,99.7) -- cycle ;
\draw  [fill={rgb, 255:red, 0; green, 0; blue, 0 }  ,fill opacity=1 ] (257.4,79.7) .. controls (257.4,78.54) and (258.34,77.6) .. (259.5,77.6) .. controls (260.66,77.6) and (261.6,78.54) .. (261.6,79.7) .. controls (261.6,80.86) and (260.66,81.8) .. (259.5,81.8) .. controls (258.34,81.8) and (257.4,80.86) .. (257.4,79.7) -- cycle ;
\draw  [fill={rgb, 255:red, 0; green, 0; blue, 0 }  ,fill opacity=1 ] (277.4,79.7) .. controls (277.4,78.54) and (278.34,77.6) .. (279.5,77.6) .. controls (280.66,77.6) and (281.6,78.54) .. (281.6,79.7) .. controls (281.6,80.86) and (280.66,81.8) .. (279.5,81.8) .. controls (278.34,81.8) and (277.4,80.86) .. (277.4,79.7) -- cycle ;
\draw  [fill={rgb, 255:red, 0; green, 0; blue, 0 }  ,fill opacity=1 ] (297.4,79.7) .. controls (297.4,78.54) and (298.34,77.6) .. (299.5,77.6) .. controls (300.66,77.6) and (301.6,78.54) .. (301.6,79.7) .. controls (301.6,80.86) and (300.66,81.8) .. (299.5,81.8) .. controls (298.34,81.8) and (297.4,80.86) .. (297.4,79.7) -- cycle ;
\draw  [fill={rgb, 255:red, 0; green, 0; blue, 0 }  ,fill opacity=1 ] (277.4,59.7) .. controls (277.4,58.54) and (278.34,57.6) .. (279.5,57.6) .. controls (280.66,57.6) and (281.6,58.54) .. (281.6,59.7) .. controls (281.6,60.86) and (280.66,61.8) .. (279.5,61.8) .. controls (278.34,61.8) and (277.4,60.86) .. (277.4,59.7) -- cycle ;
\draw [line width=1.5]    (279.5,99.7) -- (259.5,79.7) ;
\end{tikzpicture}
\caption{}
\end{subfigure}
	\begin{subfigure}{4cm}
		\centering
		\includegraphics[width=2.5cm]{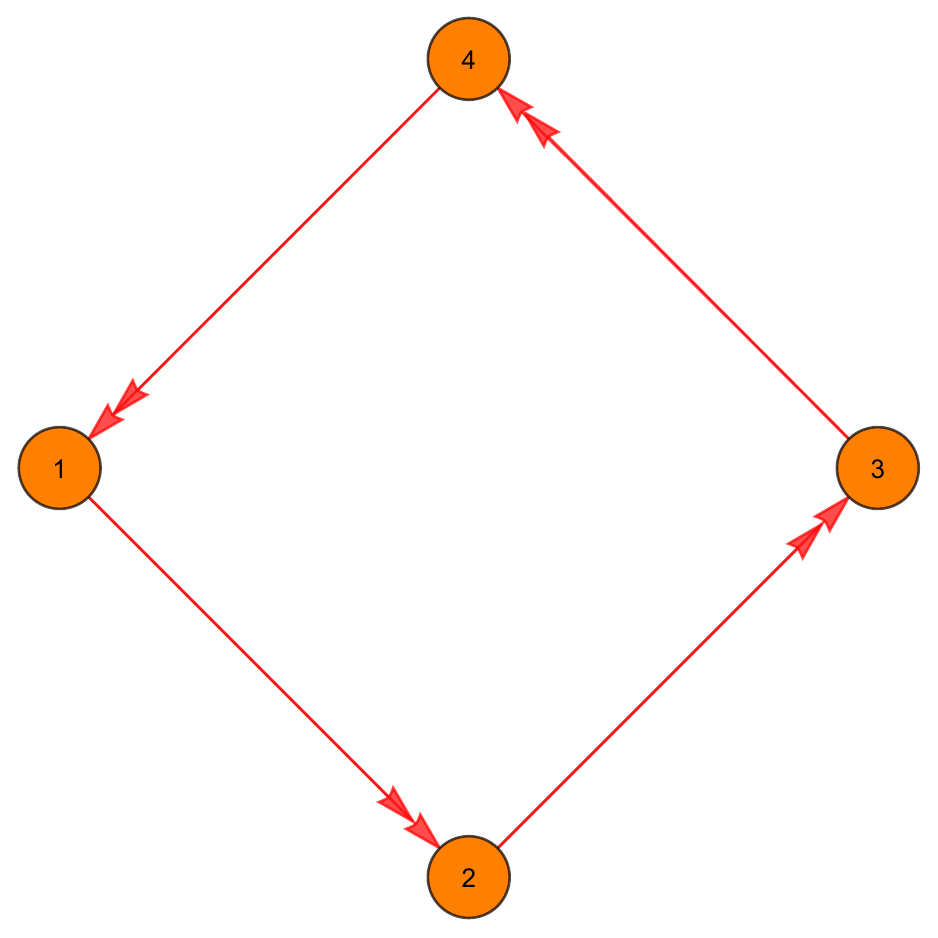}
		\caption{}
	\end{subfigure}
    \caption{(a) The dimer model. (b) The toric diagram. (c) The quiver diagram.}\label{dimerexample2}
\end{figure}
The fundamental region is the square where the numbers are the labels of the edges (rather than weights). Let us take the weight of each edge to be $\sqrt{2}$ (for the reason to be discussed shortly). Now consider for instance the vertex correpsonding to the monomial $z$ in the Newton polygon. Its perfect matching is composed of $X_2,X_5$ where $X_I$ is the arrow dual to edge $I$ \cite{Hanany:2012hi}. Therefore, this gives rise to $(-1)\times\sqrt{2}\times\sqrt{2}z=-2z$ in the spectral curve. Overall, one may check that this agrees with the Kasteleyn matrix
\begin{equation}
    K=\begin{pmatrix}
-\sqrt{2}+\sqrt{2}z & \sqrt{2}-\sqrt{2}w \\
-\sqrt{2}+\sqrt{2}w^{-1} & -\sqrt{2}+\sqrt{2}z^{-1} 
\end{pmatrix},
\end{equation}
where the signs and variables assigned to the edges are $\{1,-1,-1,-1,z,z^{-1},-w,w^{-1}\}$ (ordered by the labelling of edges). The curve is then given by
\begin{equation}
    -2z-2z^{-1}-2w-2w^{-1}+(2+2+2+2)=0,
\end{equation}
or equivalently,
\begin{equation}
    -z-z^{-1}-w-w^{-1}+4=0.\label{isoF0curve}
\end{equation}
It is straightforward to get the Newton polygon as in Figure \ref{dimerexample2}(b). The quiver in Figure \ref{dimerexample2}(c) is the dual graph of the dimer.
\end{example}

\subsubsection{Isoradial dimers}\label{isodimer}
We have introduced some rudiments of dimer models and brane tilings. Of particular interest is the isoradial 
embedding of a dimer model.

\begin{figure}[h]
	\centering
	\tikzset{every picture/.style={line width=0.75pt}}
\begin{tikzpicture}[x=0.75pt,y=0.75pt,yscale=-1,xscale=1]
\draw  [fill={rgb, 255:red, 65; green, 117; blue, 5 }  ,fill opacity=1 ] (478.11,147.4) -- (463.67,155.74) -- (463.67,139.06) -- (478.11,130.72) -- cycle ;
\draw  [fill={rgb, 255:red, 126; green, 211; blue, 33 }  ,fill opacity=1 ] (463.74,122.26) -- (478.18,130.6) -- (463.74,138.94) -- (449.3,130.6) -- cycle ;
\draw  [fill={rgb, 255:red, 139; green, 87; blue, 42 }  ,fill opacity=1 ] (478.26,114.05) -- (478.11,130.72) -- (463.74,122.26) -- (463.89,105.59) -- cycle ;
\draw  [fill={rgb, 255:red, 248; green, 231; blue, 28 }  ,fill opacity=1 ] (492.42,122.38) -- (477.98,130.72) -- (477.98,114.05) -- (492.42,105.71) -- cycle ;
\draw  [fill={rgb, 255:red, 245; green, 166; blue, 35 }  ,fill opacity=1 ] (492.55,122.39) -- (506.99,130.72) -- (492.55,139.06) -- (478.11,130.72) -- cycle ;
\draw  [fill={rgb, 255:red, 0; green, 0; blue, 255 }  ,fill opacity=1 ] (382.48,129.85) -- (396.92,121.51) -- (396.92,138.19) -- (382.48,146.53) -- cycle ;
\draw  [fill={rgb, 255:red, 0; green, 100; blue, 0 }  ,fill opacity=1 ] (382.48,129.85) -- (382.34,146.52) -- (367.97,138.06) -- (368.11,121.39) -- cycle ;
\draw  [fill={rgb, 255:red, 255; green, 0; blue, 0 }  ,fill opacity=1 ] (382.45,113.06) -- (396.86,121.45) -- (382.39,129.74) -- (367.98,121.34) -- cycle ;
\draw    (382.41,146.41) -- (382.42,129.74) ;
\draw    (382.42,129.74) -- (396.86,121.41) ;
\draw    (367.98,121.39) -- (382.42,129.74) ;
\draw    (118.2,137.11) -- (165.93,109.6) ;
\draw    (165.93,109.6) -- (213.63,137.18) ;
\draw    (165.9,164.69) -- (213.63,137.18) ;
\draw    (118.2,137.11) -- (165.9,164.69) ;
\draw   (118.17,192.2) -- (70.47,164.63) -- (70.51,109.53) -- (118.24,82.02) -- (165.93,109.6) -- (165.9,164.69) -- cycle ;
\draw   (213.59,192.27) -- (165.9,164.69) -- (165.93,109.6) -- (213.67,82.09) -- (261.36,109.67) -- (261.32,164.76) -- cycle ;
\draw  [fill={rgb, 255:red, 0; green, 0; blue, 0 }  ,fill opacity=1 ] (73.9,109.53) .. controls (73.9,107.66) and (72.39,106.14) .. (70.51,106.14) .. controls (68.64,106.14) and (67.12,107.66) .. (67.12,109.53) .. controls (67.12,111.41) and (68.64,112.92) .. (70.51,112.92) .. controls (72.39,112.92) and (73.9,111.41) .. (73.9,109.53) -- cycle ;
\draw  [fill={rgb, 255:red, 0; green, 0; blue, 0 }  ,fill opacity=1 ] (121.56,192.2) .. controls (121.56,190.33) and (120.04,188.81) .. (118.17,188.81) .. controls (116.29,188.81) and (114.77,190.33) .. (114.77,192.2) .. controls (114.77,194.08) and (116.29,195.6) .. (118.17,195.6) .. controls (120.04,195.6) and (121.56,194.08) .. (121.56,192.2) -- cycle ;
\draw  [fill={rgb, 255:red, 0; green, 0; blue, 0 }  ,fill opacity=1 ] (216.98,192.27) .. controls (216.98,190.4) and (215.46,188.88) .. (213.59,188.88) .. controls (211.72,188.88) and (210.2,190.4) .. (210.2,192.27) .. controls (210.2,194.14) and (211.72,195.66) .. (213.59,195.66) .. controls (215.46,195.66) and (216.98,194.14) .. (216.98,192.27) -- cycle ;
\draw  [fill={rgb, 255:red, 0; green, 0; blue, 0 }  ,fill opacity=1 ] (264.75,109.67) .. controls (264.75,107.79) and (263.23,106.27) .. (261.36,106.27) .. controls (259.48,106.27) and (257.97,107.79) .. (257.97,109.67) .. controls (257.97,111.54) and (259.48,113.06) .. (261.36,113.06) .. controls (263.23,113.06) and (264.75,111.54) .. (264.75,109.67) -- cycle ;
\draw  [fill={rgb, 255:red, 255; green, 255; blue, 255 }  ,fill opacity=1 ] (121.63,82.02) .. controls (121.63,80.15) and (120.12,78.63) .. (118.24,78.63) .. controls (116.37,78.63) and (114.85,80.15) .. (114.85,82.02) .. controls (114.85,83.89) and (116.37,85.41) .. (118.24,85.41) .. controls (120.12,85.41) and (121.63,83.89) .. (121.63,82.02) -- cycle ;
\draw  [fill={rgb, 255:red, 255; green, 255; blue, 255 }  ,fill opacity=1 ] (73.87,164.63) .. controls (73.87,162.75) and (72.35,161.23) .. (70.47,161.23) .. controls (68.6,161.23) and (67.08,162.75) .. (67.08,164.63) .. controls (67.08,166.5) and (68.6,168.02) .. (70.47,168.02) .. controls (72.35,168.02) and (73.87,166.5) .. (73.87,164.63) -- cycle ;
\draw  [fill={rgb, 255:red, 255; green, 255; blue, 255 }  ,fill opacity=1 ] (264.71,164.76) .. controls (264.71,162.89) and (263.19,161.37) .. (261.32,161.37) .. controls (259.45,161.37) and (257.93,162.89) .. (257.93,164.76) .. controls (257.93,166.63) and (259.45,168.15) .. (261.32,168.15) .. controls (263.19,168.15) and (264.71,166.63) .. (264.71,164.76) -- cycle ;
\draw  [fill={rgb, 255:red, 255; green, 255; blue, 255 }  ,fill opacity=1 ] (217.06,82.09) .. controls (217.06,80.21) and (215.54,78.7) .. (213.67,78.7) .. controls (211.79,78.7) and (210.27,80.21) .. (210.27,82.09) .. controls (210.27,83.96) and (211.79,85.48) .. (213.67,85.48) .. controls (215.54,85.48) and (217.06,83.96) .. (217.06,82.09) -- cycle ;
\draw  [dash pattern={on 0.84pt off 2.51pt}]  (118.2,137.11) -- (213.63,137.18) ;
\draw  [draw opacity=0] (157.99,160.66) .. controls (157.93,160.08) and (157.95,159.48) .. (158.06,158.88) .. controls (158.62,155.88) and (161.28,153.85) .. (163.99,154.36) .. controls (164.64,154.48) and (165.23,154.73) .. (165.75,155.09) -- (162.98,159.79) -- cycle ; \draw   (157.99,160.66) .. controls (157.93,160.08) and (157.95,159.48) .. (158.06,158.88) .. controls (158.62,155.88) and (161.28,153.85) .. (163.99,154.36) .. controls (164.64,154.48) and (165.23,154.73) .. (165.75,155.09) ;
\draw [color={rgb, 255:red, 255; green, 0; blue, 0 }  ,draw opacity=1 ][line width=1.5]    (165.93,109.6) -- (165.9,164.69) ;
\draw  [fill={rgb, 255:red, 0; green, 0; blue, 0 }  ,fill opacity=1 ] (169.33,109.6) .. controls (169.33,107.73) and (167.81,106.21) .. (165.93,106.21) .. controls (164.06,106.21) and (162.54,107.73) .. (162.54,109.6) .. controls (162.54,111.47) and (164.06,112.99) .. (165.93,112.99) .. controls (167.81,112.99) and (169.33,111.47) .. (169.33,109.6) -- cycle ;
\draw  [fill={rgb, 255:red, 255; green, 255; blue, 255 }  ,fill opacity=1 ] (169.29,164.69) .. controls (169.29,162.82) and (167.77,161.3) .. (165.9,161.3) .. controls (164.02,161.3) and (162.5,162.82) .. (162.5,164.69) .. controls (162.5,166.57) and (164.02,168.08) .. (165.9,168.08) .. controls (167.77,168.08) and (169.29,166.57) .. (169.29,164.69) -- cycle ;
\draw   (367.97,138.06) -- (353.53,129.72) -- (353.55,113.04) -- (367.99,104.71) -- (382.43,113.06) -- (382.42,129.74) -- cycle ;
\draw   (396.85,138.08) -- (382.42,129.74) -- (382.43,113.06) -- (396.88,104.73) -- (411.31,113.08) -- (411.3,129.76) -- cycle ;
\draw   (382.39,163.09) -- (367.96,154.74) -- (367.97,138.06) -- (382.42,129.74) -- (396.85,138.08) -- (396.84,154.76) -- cycle ;
\draw  [fill={rgb, 255:red, 255; green, 255; blue, 255 }  ,fill opacity=1 ] (369.69,104.71) .. controls (369.69,103.78) and (368.93,103.02) .. (367.99,103.02) .. controls (367.06,103.02) and (366.3,103.78) .. (366.3,104.71) .. controls (366.3,105.65) and (367.06,106.41) .. (367.99,106.41) .. controls (368.93,106.41) and (369.69,105.65) .. (369.69,104.71) -- cycle ;
\draw  [fill={rgb, 255:red, 255; green, 255; blue, 255 }  ,fill opacity=1 ] (355.23,129.72) .. controls (355.23,128.78) and (354.47,128.02) .. (353.53,128.02) .. controls (352.6,128.02) and (351.84,128.78) .. (351.84,129.72) .. controls (351.84,130.65) and (352.6,131.41) .. (353.53,131.41) .. controls (354.47,131.41) and (355.23,130.65) .. (355.23,129.72) -- cycle ;
\draw  [fill={rgb, 255:red, 255; green, 255; blue, 255 }  ,fill opacity=1 ] (384.11,129.74) .. controls (384.11,128.8) and (383.35,128.04) .. (382.42,128.04) .. controls (381.48,128.04) and (380.72,128.8) .. (380.72,129.74) .. controls (380.72,130.67) and (381.48,131.43) .. (382.42,131.43) .. controls (383.35,131.43) and (384.11,130.67) .. (384.11,129.74) -- cycle ;
\draw  [fill={rgb, 255:red, 255; green, 255; blue, 255 }  ,fill opacity=1 ] (412.99,129.76) .. controls (412.99,128.82) and (412.24,128.06) .. (411.3,128.06) .. controls (410.36,128.06) and (409.6,128.82) .. (409.6,129.76) .. controls (409.6,130.69) and (410.36,131.45) .. (411.3,131.45) .. controls (412.24,131.45) and (412.99,130.69) .. (412.99,129.76) -- cycle ;
\draw  [fill={rgb, 255:red, 255; green, 255; blue, 255 }  ,fill opacity=1 ] (369.65,154.74) .. controls (369.65,153.8) and (368.89,153.04) .. (367.96,153.04) .. controls (367.02,153.04) and (366.26,153.8) .. (366.26,154.74) .. controls (366.26,155.68) and (367.02,156.43) .. (367.96,156.43) .. controls (368.89,156.43) and (369.65,155.68) .. (369.65,154.74) -- cycle ;
\draw  [fill={rgb, 255:red, 255; green, 255; blue, 255 }  ,fill opacity=1 ] (398.54,154.76) .. controls (398.54,153.82) and (397.78,153.06) .. (396.84,153.06) .. controls (395.9,153.06) and (395.14,153.82) .. (395.14,154.76) .. controls (395.14,155.7) and (395.9,156.45) .. (396.84,156.45) .. controls (397.78,156.45) and (398.54,155.7) .. (398.54,154.76) -- cycle ;
\draw  [fill={rgb, 255:red, 255; green, 255; blue, 255 }  ,fill opacity=1 ] (398.57,104.73) .. controls (398.57,103.8) and (397.81,103.04) .. (396.88,103.04) .. controls (395.94,103.04) and (395.18,103.8) .. (395.18,104.73) .. controls (395.18,105.67) and (395.94,106.43) .. (396.88,106.43) .. controls (397.81,106.43) and (398.57,105.67) .. (398.57,104.73) -- cycle ;
\draw  [fill={rgb, 255:red, 0; green, 0; blue, 0 }  ,fill opacity=1 ] (355.24,113.04) .. controls (355.24,112.1) and (354.48,111.35) .. (353.55,111.35) .. controls (352.61,111.35) and (351.85,112.1) .. (351.85,113.04) .. controls (351.85,113.98) and (352.61,114.74) .. (353.55,114.74) .. controls (354.48,114.74) and (355.24,113.98) .. (355.24,113.04) -- cycle ;
\draw  [fill={rgb, 255:red, 0; green, 0; blue, 0 }  ,fill opacity=1 ] (384.12,113.06) .. controls (384.12,112.12) and (383.37,111.37) .. (382.43,111.37) .. controls (381.49,111.37) and (380.73,112.12) .. (380.73,113.06) .. controls (380.73,114) and (381.49,114.76) .. (382.43,114.76) .. controls (383.37,114.76) and (384.12,114) .. (384.12,113.06) -- cycle ;
\draw  [fill={rgb, 255:red, 0; green, 0; blue, 0 }  ,fill opacity=1 ] (413.01,113.08) .. controls (413.01,112.15) and (412.25,111.39) .. (411.31,111.39) .. controls (410.37,111.39) and (409.61,112.15) .. (409.61,113.08) .. controls (409.61,114.02) and (410.37,114.78) .. (411.31,114.78) .. controls (412.25,114.78) and (413.01,114.02) .. (413.01,113.08) -- cycle ;
\draw  [fill={rgb, 255:red, 0; green, 0; blue, 0 }  ,fill opacity=1 ] (369.67,138.06) .. controls (369.67,137.13) and (368.91,136.37) .. (367.97,136.37) .. controls (367.03,136.37) and (366.27,137.13) .. (366.27,138.06) .. controls (366.27,139) and (367.03,139.76) .. (367.97,139.76) .. controls (368.91,139.76) and (369.67,139) .. (369.67,138.06) -- cycle ;
\draw  [fill={rgb, 255:red, 0; green, 0; blue, 0 }  ,fill opacity=1 ] (398.55,138.08) .. controls (398.55,137.15) and (397.79,136.39) .. (396.85,136.39) .. controls (395.92,136.39) and (395.16,137.15) .. (395.16,138.08) .. controls (395.16,139.02) and (395.92,139.78) .. (396.85,139.78) .. controls (397.79,139.78) and (398.55,139.02) .. (398.55,138.08) -- cycle ;
\draw  [fill={rgb, 255:red, 0; green, 0; blue, 0 }  ,fill opacity=1 ] (384.09,163.09) .. controls (384.09,162.15) and (383.33,161.39) .. (382.39,161.39) .. controls (381.46,161.39) and (380.7,162.15) .. (380.7,163.09) .. controls (380.7,164.02) and (381.46,164.78) .. (382.39,164.78) .. controls (383.33,164.78) and (384.09,164.02) .. (384.09,163.09) -- cycle ;
\draw  [fill={rgb, 255:red, 208; green, 2; blue, 27 }  ,fill opacity=1 ] (492.48,139.18) -- (492.34,155.86) -- (477.97,147.4) -- (478.11,130.72) -- cycle ;
\draw   (477.97,147.4) -- (463.53,139.05) -- (463.55,122.37) -- (477.99,114.05) -- (492.43,122.39) -- (492.42,139.07) -- cycle ;
\draw  [fill={rgb, 255:red, 255; green, 255; blue, 255 }  ,fill opacity=1 ] (479.69,114.05) .. controls (479.69,113.11) and (478.93,112.35) .. (477.99,112.35) .. controls (477.06,112.35) and (476.3,113.11) .. (476.3,114.05) .. controls (476.3,114.98) and (477.06,115.74) .. (477.99,115.74) .. controls (478.93,115.74) and (479.69,114.98) .. (479.69,114.05) -- cycle ;
\draw  [fill={rgb, 255:red, 255; green, 255; blue, 255 }  ,fill opacity=1 ] (465.23,139.05) .. controls (465.23,138.11) and (464.47,137.35) .. (463.53,137.35) .. controls (462.6,137.35) and (461.84,138.11) .. (461.84,139.05) .. controls (461.84,139.99) and (462.6,140.75) .. (463.53,140.75) .. controls (464.47,140.75) and (465.23,139.99) .. (465.23,139.05) -- cycle ;
\draw  [fill={rgb, 255:red, 255; green, 255; blue, 255 }  ,fill opacity=1 ] (494.11,139.07) .. controls (494.11,138.13) and (493.35,137.37) .. (492.42,137.37) .. controls (491.48,137.37) and (490.72,138.13) .. (490.72,139.07) .. controls (490.72,140.01) and (491.48,140.77) .. (492.42,140.77) .. controls (493.35,140.77) and (494.11,140.01) .. (494.11,139.07) -- cycle ;
\draw  [fill={rgb, 255:red, 0; green, 0; blue, 0 }  ,fill opacity=1 ] (465.24,122.37) .. controls (465.24,121.44) and (464.48,120.68) .. (463.55,120.68) .. controls (462.61,120.68) and (461.85,121.44) .. (461.85,122.37) .. controls (461.85,123.31) and (462.61,124.07) .. (463.55,124.07) .. controls (464.48,124.07) and (465.24,123.31) .. (465.24,122.37) -- cycle ;
\draw  [fill={rgb, 255:red, 0; green, 0; blue, 0 }  ,fill opacity=1 ] (494.12,122.39) .. controls (494.12,121.46) and (493.37,120.7) .. (492.43,120.7) .. controls (491.49,120.7) and (490.73,121.46) .. (490.73,122.39) .. controls (490.73,123.33) and (491.49,124.09) .. (492.43,124.09) .. controls (493.37,124.09) and (494.12,123.33) .. (494.12,122.39) -- cycle ;
\draw  [fill={rgb, 255:red, 0; green, 0; blue, 0 }  ,fill opacity=1 ] (479.67,147.4) .. controls (479.67,146.46) and (478.91,145.7) .. (477.97,145.7) .. controls (477.03,145.7) and (476.27,146.46) .. (476.27,147.4) .. controls (476.27,148.33) and (477.03,149.09) .. (477.97,149.09) .. controls (478.91,149.09) and (479.67,148.33) .. (479.67,147.4) -- cycle ;
\draw (151,144.5) node [anchor=north west][inner sep=0.75pt]  [font=\fontsize{0.07em}{0.08em}\selectfont] [align=left] {{\fontsize{0.07em}{0.08em}\selectfont $\theta_I$}};
\draw (158,233.67) node [anchor=north west][inner sep=0.75pt]   [align=left] {(a)};
\draw (424,232.33) node [anchor=north west][inner sep=0.75pt]   [align=left] {(b)};
\end{tikzpicture}
    \caption{(a) The edge in red has length $l$. Its dashed dual edge is of length $\sqrt{4-l^2}$ which equals the weight of the edge. The rhombus angle is labelled by $\theta_I$. The corresponding internal angle $2\theta_I$ of the rombus gives the R-charge physically. (b) The left plot indicates $\sum2\theta_I=2\pi$ which corresponds to each superpotential term while the right plot indicates $\sum(\pi-2\theta_I)=2\pi$ which correpsonds to (the fields connected to) each node in the quiver.}\label{theta}
\end{figure}
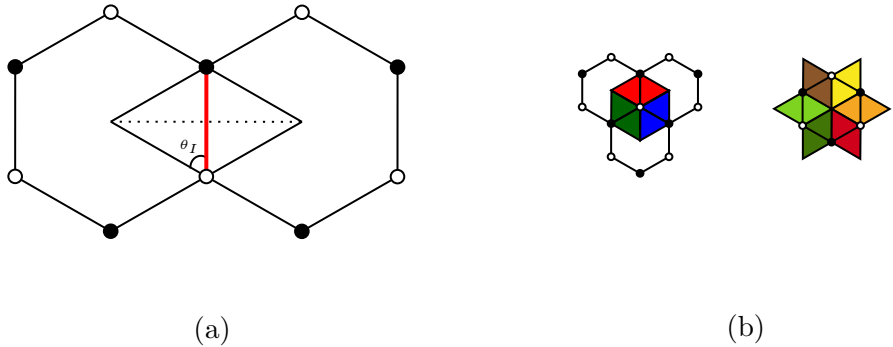

\begin{definition}
A dimer is {\bf isoradial} if every face is inscribed in a circle of the same radius, which we can take to be $1$. In this paper, we will mostly choose the weight of an edge to be $\sqrt{4-l^2}$ for an isoradial dimer where $l$ is the length of the edge.
\end{definition}
The reason for this choice is that the edge weight is equal to the distance of the circumcentres of the two faces adjacent to the edge (i.e., its dual, perpendicular, edge).
We illustrate this in Figure \ref{theta}(a) in a hexagonal tiling example.
As we will see shortly, we can always construct a spectral curve of certain kind (a so-called genus 0 Harnack curve) from such isoradial embedding of a dimer. According to \cite{Kenyon_2002}, this choice of edge weight is \emph{critical} in the sense that it uniquely maximizes the (normalized) determinant of Dirac operator. We will also later see that it is closely related to the Mahler measure for isoradial embeddings.

We can also express this edge weight in terms of the {\bf rhombus angle} $\theta_I$, as shown in Figure \ref{theta}(a). In our convention, $2\theta_I$ is the angle of the rhombus at the vertex that has in common with the edge. 
It is easy to see that our chosen edge weight is 
$\sqrt{4 - l^2} = 2\sin(\theta_I)$. In other words, $l=2\cos(\theta_I)$. The energy function associated to this edge $e_I$, recalling that edge weight is $e^{-\mathcal{E}}$, is then $\mathcal{E}(e_I)=-\log(2\sin(\theta_I))$.

Now, the internal angle $2\theta_I$ is essentially the R-charge of the corresponding chiral multiplet in the dual quiver gauge theory \cite{Franco:2006gc}: for a field $X_I$ with R-charge $R_I$,
\begin{equation}
    2\theta_I=\pi R_I.
\end{equation}
Indeed, we have (i), that $\sum2\theta_I=2\pi$, which is the geometric recasting of the condition on R-charges from the vanishing $\beta$-function, that $\sum R_I = 2$.
Likewise, we have (ii), that $\sum(\pi-2\theta_I)=2\pi$.
Notice the difference between the two sums: (i) is a sum over the angles whose edges are connected to the same black or white node while the (ii) is a sum of angles in the same face.
We depict this in Figure \ref{theta}(b), where a coloured rhombus has rhombus angle $\theta_I$. The left plot represents (i), a sum over the parts of rhombi surrounding a \emph{vertex} (drawn as white in the middle). Every such sum corresponds to a term in the superpotential. 
The right plot represents (ii), a sum over the part of rhombi surrounding 
the (circum)centre of the dimer \emph{face}. Each contributes an angle of $(\pi-2\theta_I)$ so that $\sum(\pi-2\theta_I)=2\pi$. Every such sum corresponds to arrows attached to a node in the dual quiver.

\begin{example}
Recall the dimer in Figure \ref{dimerexample2}.  Since the rhombus angles are all $\pi/4$, each edge weight equals $2\sin(\pi/4)=\sqrt{2}$.
\end{example}

\paragraph{Isoradial Spectral Curve and GLSM Fields}
It is straightforward to obtain the spectral curve in terms of the rhombus angles from Kasteleyn matrix. When taking determinant, each term we get is simply a product of edge weights $2\sin(\theta_I)$ contributed from the corresponding rhombus angles/R-charges\footnote{Note that so far by R-charges, we mean all possible trial R-charges that satisfy the conformality condition. In other words, the rhombus angles are still variables in the spectral curve. We will determine their exact values (and hence exact coefficients for the curve) in \S\ref{isomax}.}. Therefore, we need to figure out which edges contribute to each monomial in the Newton polynomial.
This can be seen from the perfect matching(s) associated to each lattice point in the Newton polygon. Physically, each perfect matching can be interpreted as a gauged linear sigma model (GLSM) field. Every lattice point in the Newton polygon is associated with one or more GLSM fields.

For a vertex/corner point $v_i$, we only have one corresponding GLSM field $p_i$. It can be written as $p_i=\sum\limits_IX_I$ where $X_I$'s are the arrows in the quiver. Recall that $X_I$'s are arrows dual to the edges $e_I$ in the dimer. When computing $\det(K)$, we would get the monomial corresponding to $v_i$ as a product of these $e_I$'s. Since they have weights $2\sin(\theta_I)$, this gives the term
\begin{equation}
    (-1)^{\delta}\prod_I2\sin(\theta_I)z^aw^b\label{monomial1}
\end{equation}
where the factor $z^aw^b$ can be directly read off from the Newton polygon, and we will explain what $\delta$ means shortly.

This can be generalized to any lattice point. For interior points and other boundary points, they correspond to multiple GLSM fields. Suppose one of such points is associated to GLSM fields $q_1,\dots,q_k$, then each $q_i$ can be written as $q_i=\sum\limits_IX_I$. As a result, each $q_i$ gives rise to a product of $2\sin(\theta_I)$ from the determinant. Then the corresponding monomial in the spectral curve is the sum of these products for every $q_i$,
\begin{equation}
    (-1)^\delta\sum_i\begin{bmatrix}\text{weight}\\\text{of~} q_i\end{bmatrix}z^aw^b=(-1)^\delta\sum_i\left(\prod_I2\sin(\theta_I)\right)z^aw^b.\label{monomial2}
\end{equation}

Now let us determine $\delta$. Given a reference perfect matching $M_0$, denote the horizontal (vertical) height change of the perfect matching $M$ to be $h_x$ ($h_y$). Then the above rules of writing the Newton polynomial should agree with the result in \cite{Kenyon:2003uj}:
\begin{equation}
    P(z,w)=\sum_M\text{e}^{-\mathcal{E}(M)}z^{h_x}w^{h_y}(-1)^{h_xh_y+h_x+h_y}.\label{monomials}
\end{equation}
It is straightforward to see that the energy of $M$ is consistent with \eqref{monomial1} and \eqref{monomial2}, that is, $\mathcal{E}(M)=\sum\limits_I\mathcal{E}(e_I)=-\log\left(\prod\limits_I2\sin(\theta_I)\right)$. Now different reference $M_0$ may give different signs for each term, but it would preserve certain properties of the spectral curve (such as its Mahler measure). Here, we will stick to the perfect matching corresponding to $a=b=0$ (i.e. the origin of the Newton polygon) as the reference $M_0$ so that the powers of variables agrees with \eqref{monomial1} and \eqref{monomial2} \footnote{Notice that in \cite{Kenyon:2003uj}, there is also a total factor $z^{x_0}w^{y_0}$ in the front of the right hand side in \eqref{monomials}, where $x_0$ and $y_0$ are the total flows across the horizontal and vertical cycles respectively. This would ensure that the overall powers $z^{x_0+h_x}w^{y_0+h_y}$ is the same as $z^ay^b$. For simplicity, we remove this factor in \eqref{monomials} as long as we choose the one with $x_0=y_0=0$ as our reference perfect matching.}. Then the parity of $(h_xh_y+h_x+h_y)$ is fully determined by $a=h_x$ and $b=h_y$. Thus, we may write $\delta$ as
\begin{equation}
    \delta=\begin{cases}
    0,&\text{both $a$ and $b$ are even}\\
    1,&\text{otherwise}
    \end{cases}.
\end{equation}

\begin{example}
Recall the example in Figure \ref{dimerexample}. Let us choose the green matching as the reference perfect matching. As the blue matching has height change $(0,-1)$, we have $\delta=1$ and this gives rise to the term $-2w^{-1}$. Altogether, we have the spectral curve
\begin{equation}
    -2z-2z^{-1}-2w-2w^{-1}+(2+2+2+2)=0,
\end{equation}
or equivalently,
\begin{equation}
    -z-z^{-1}-w-w^{-1}+4=0.
\end{equation}
This agrees with \eqref{isoF0curve}.
\end{example}

\subsubsection{Amoebae and Harnack Curves}\label{harnack}
Now that we have some familiarity with dimers and Newton polynomials/spectral curves, let us 
collect some facts on amoebae and Harnack curves \footnote{See more details in \cite{Bao:2021olg}, as well recent machine-learning results thereon.
}, adhering to the notation of \cite{Kenyon:2003uj,Kenyon:2003ui}. 

\begin{definition}
An {\bf amoeba} is the set of points in the real plane, of the logarithmic projection of the spectral curve $P=0$:
\begin{equation}
\cA_P=\Big\{ \big(\log|z|,\log|w|\big)\ \big|\ P(z,w)=0 \Big\}.
\end{equation}
\end{definition}
The definition of the amoeba can be easily extended to polynomials of more variables, but in this paper, we will only focus on Newton polynomials in two variables $(z,w)$, and everything will be planar: the Newton polynomial lives in $\mathbb{C}^2$; the amoeba lives in $\mathbb{R}^2$ and the toric diagram lives in $\mathbb{Z}^2$. Let us also introduce the spine as a deformation retract of the amoeba. For simplicity, we shall use the result in \cite{Feng:2005gw} as our definition. For the original definition, see \cite{passare2000amoebas}.
\begin{definition}
The {\bf spine} $\mathcal{S}$ of the amoeba is the dual $(p,q)$-web of the toric diagram associated to $P(z,w)$.
\end{definition}
In parallel, we have
\begin{definition}
A real algebraic curve $C\subset\mathbb{RP}^2$ of degree $d$ is an {\bf M-curve} if it has the maximal number of connected components, i.e., $\frac{(d-1)(d-2)}{2}+1$. Following \cite{Kenyon:2003ui}, we shall call the connected components {\bf ovals}. Ovals that do not intersect the coordinate axes are known as {\bf compact ovals}. An isolated real point on the curve is regarded as a {\bf degenerate oval}. The {\bf genus} $g$ is the number of non-degenerate compact ovals. For an M-curve, the genus is also maximal and equals $\frac{(d-1)(d-2)}{2}$.
\end{definition}

A Harnack curve is a special type of M-curve in the sense that its ovals have the ``best'' possible topological configurations (see Figure 2 in \cite{Kenyon:2003ui} for an illustration).
The definition of Harnack curves is quite intricate  \cite{mikhalkin2000real}. Here, we will take the following characterization as the working definition:
\begin{definition}
A {\bf Harnack curve} $C$ possesses the map
\begin{equation}
    C(\mathbb{C})\ni(z,w)\mapsto(\log|z|,\log|w|)\in\mathcal{A}_C
\end{equation}
such that it is 2-to-1 from the curve to its amoeba (except for a finite number of real nodes where it is 1-to-1). The amoeba of Harnack curve with genus $g$ has exactly $g$ holes (i.e., compact complementary regions). Hence, the number of holes for an amoeba is also called the genus of the amoeba.
\end{definition}
From \cite{Kenyon:2003uj}, we have a practical way to identify Harnack curves associated with amoebae and dimers.
\begin{theorem}
    For any choice of non-negative edge weights on a dimer, the spectral curve $P(z,w)=0$ is a Harnack curve of degree $d$ with $\frac{(d-1)(d-2)}{2}$ compact ovals.
\end{theorem}
There is another remarkable theorem \cite{mikhalkin2001amoebas} that will be crucial to us:
\begin{theorem}
    A curve is Harnack if and only if its amoeba has the maximal possible area for a given Newton polygon $\Delta$. That is, $A(\mathcal{A}_P)=\pi^2A(\Delta)$ where $A(\Delta)$ is the unnormalized area of the Newton polygon.
\end{theorem}
In fact, the 2-to-1 feature for Harnack curves leads to the two important propositions \cite{Kenyon:2003uj}:
\begin{proposition}
The boundary of the amoeba is the image of the real locus of the spectral curve $P(z,w)=0$.
It follows that the amoeba of a Harnack curve can be determined by
\begin{equation}
    \prod_{n_1,n_2\in\mathbb{Z}^2}P\left((-1)^{n_1}\text{e}^{x},(-1)^{n_2}\text{e}^{y}\right)\leq0.
\end{equation}
\end{proposition}

\begin{proposition}
Any interior lattice point of the Newton polygon $\Delta$ corresponds to either a bounded complementary region (i.e., a hole) of the amoeba or an isolated real node in the spectral curve. In particular, the number of holes of the amoeba is equal to the genus $g$ of the curve.\label{prophole}
\end{proposition}

Finally, we have a theorem from \cite{Kenyon:2003ui} for isoradial dimers:
\begin{theorem}
    A dimer corresponding to a genus zero Harnack curve is isoradial if and only if its amoeba contains the origin.
\end{theorem}

On the log plane, we can always shift the amoeba so that it contains the origin, which corresponds to a rescaling of $z$ and $w$ of the spectral curve. This gives a canonical family of isoradial parameterizations for any genus-zero Harnack curve.
\begin{example}
The spectral curve \eqref{isoF0curve} is Harnack and of genus $0$. In fact, $-z-z^{-1}-w-w^{-1}+k=0$ is Harnack when $k\geq4$, with $g=0$ for $k=4$ and $g=1$ otherwise.
\end{example}

\paragraph{Ronkin functions} 
Closely related to the amoeba is the so-called Ronkin function. In fact, to probe different regions of the amoeba, we can use this analytic tool.
\begin{definition}
In two dimensions, the generalized Mahler measure \eqref{genMahler} with $a_1=\exp(x),\, a_2= \exp(y)$ defines the {\bf Ronkin function} $R(x,y) \coloneqq m(P;\mathrm{e}^x,\, \mathrm{e}^y)$. In particular, $R(0,0)=m(P;1,1)=m(P)$.
\end{definition}
Following \cite{ronkin2000zeros,passare2000amoebas,forsberg2000laurent,mikhalkin2004amoebas}, we have\footnote{We are focusing on $\mathbb{R}^2$ in this paper, but the discussions on Ronkin functions here can be directly extended to any $\mathbb{R}^n$.}
\begin{theorem}
    The Ronkin function $R(x,y)$ is convex. It is strictly convex over $\mathcal{A}_P$ and linear over each component of $\mathbb{R}^2\backslash\mathcal{A}_P$. The gradient $\nabla = (\partial_x, \partial_y)$ of the Ronkin function satisfies
    \begin{itemize}
        \item     $\textup{Int}(\Delta)\subset\nabla R(\mathbb{R}^2)\subset\Delta$, where $\Delta$ is the Newton polygon for $P$ and $\textup{Int}(\Delta)$ is its interior; 
        \item for each component $E_i$ of $\mathbb{R}^2\backslash\mathcal{A}_P$, $\nabla R(E_i)=(a_i,b_i)$, where $(a_i,b_i)$ is the lattice point in $\Delta$ corresponding to $E_i$.
    \end{itemize}
\end{theorem}

Given a Harnack curve with ovals (either degenerate or non-degenerate), we can always shift the amoeba such that a hole or a critical point\footnote{By critical point, we mean that this point in the amoeba corresponds to a node in the spectral curve.} is located at the orgin. In terms of $P(z,w)$, this is a rescaling/redefinition of the variables $z,w$ (we will make this more precise in \S\ref{mahlerflow}). Since the Mahler measure is the Ronkin function at $(0,0)$ (which is always a lattice point in $\Delta$), we have $\nabla R(0,0)=(0,0)$ following this theorem.
Moreover, as the Ronkin function is always convex, we have
\begin{corollary}
Given a Laurent polynomial $P$ (with possible rescaling of variables), the Mahler measure $m(P)$ is the minimum of $R(x,y)$.
\end{corollary}

\subsubsection{Crystal Melting and D-branes}\label{crystal}
Another physical system, in contrast to quiver gauge theories of brane tiling, that arise from dimer models is the so-called crystal melting model, which counts certain BPS bound states and from which toric geometry emerges \cite{Ooguri:2009ijd,Ooguri:2009ri}.

In Type IIA string theory, consider D0- and D2-branes with a single D6 on a toric Gorenstein 3-fold $\mathcal{X}$. More generally, we also include D4-branes \cite{Szabo:2009vw,Aganagic:2010qr,Nishinaka:2013mba}. Denote the charges of D$p$-branes as $Q_p$, where $Q_{6,4}$ are magnetic while $Q_{2,0}$ are electric. In this configuration, the D4s wrap an ample divisor class $[C]=\sum\limits_{i=1}^{b_2(\mathcal{X})}Q_{4,i}[C_i]\in H_4(\mathcal{X},\mathbb{Z})$ where $C_i$ is a basis of the 4-cycles and $b_2=b_4$. Likewise, the D2s wrap the 2-cycle $[S]=\sum\limits_{i=1}^{b_2(\mathcal{X})}Q_{2,i}[S_i]\in H_2(\mathcal{X},\mathbb{Z})$. Then, the D6-D4-D2-D0 bound states are counted by the partition function \cite{Szabo:2009vw}
\begin{equation}
    Z_\text{BPS}=\sum_{Q_0,\bm{Q}_2}\Omega(Q_0,\bm{Q}_2,\bm{Q}_4,Q_6)\text{e}^{-Q_0\phi_0-\bm{Q}_2\bm{\phi}_2},
\end{equation}
where $\phi_p$ are the chemical potentials for D$p$-branes and $\Omega$ is the degeneracy (the Witten index) of the bound states. In fact, the chemical potentials $\phi_0$ and $\phi_2$ can be identified with string coupling $g_s$ and K\"ahler moduli respectively \cite{Ooguri:2009ri}.

The profound results of \cite{Okounkov:2003sp,Ooguri:2009ijd,Ooguri:2009ri} then relate the BPS states to melting crystals. A crystal is also a dual diagram of the dimer in the following sense. In the crystal, there are different types of atoms. Each type corresponds to a node in the quiver, and the chemical bonds between atoms are represented by the arrows.

Given an initial crystal, one can melt it by removing atoms from the top of it. The BPS degeneracy is equal to the number of melting crystal configurations with $Q_0$ being the total number of atoms removed and $\bm{Q}_2$ the numbers of atoms of different types.

In the thermodynamic limit where a large number of atoms are removed, the shape of the molten crystal is exactly the (minus) Ronkin function, whose 2d projection is the amoeba of Newton polynomial $P$ associated to the Gorenstein 3-fold $\mathcal{X}$. Therefore, using saddle point approximation, we have \cite{Ooguri:2009ri}
\begin{proposition}
In the thermodynamic limit,
\begin{equation}
    Z_\textup{BPS}\sim\textup{exp}\left(\int\textup{d}x\,\textup{d}y\,R(x,y)\right)\label{ZBPS}
\end{equation}
(where we have omitted a factor of $4/g_s^2$ in the exponential). We may then define the {\bf free energy} as $F\equiv-\log Z_\textup{BPS}$ \cite{Kenyon:2003uj}.
\end{proposition}
In general, \eqref{ZBPS} is divergent. Hence, we need to normalize the partition function by $Z/Z_0$ where $Z_0$ is the partition function of the initial unmolten crystal. Then the volume between the Ronkin functions for $Z$ and $Z_0$ would remain finite\footnote{Note that this volume is different from the volume under Ronkin function discussed in \cite{Kenyon:2003ui}.}.

In particular, the phase structures of crystals are given by amoebae.
\begin{definition}
An unbounded complementary region of the amoeba corresponds to an unmolten part in the crystal and is hence called the {\bf solid phase}. For the parts where atoms are removed in the crystal, the interior of the amoeba is known as the {\bf liquid phase} while a bounded complementary region of the amoeba is known as the {\bf gas phase}.
\end{definition}
With D6-D2-D0 bound states, there would only be liquid and solid phases. Gas phases would appear when we further add D4 branes.

Following Proposition \ref{prophole}, the number of gas phases of a dimer/crystal model is equal to the genus of $P(z,w)=0$. In general, every solid/frozen phase corresponds to a boundary point on $\Delta$ and every gas phase corresponds to an interior point (except for degenerate cases).

\begin{example}
The Ronkin function and amoeba with $P=k-z-z^{-1}-w-w^{-1}$ ($k>4$) are sketched in Figure \ref{Ronkinamoeba}.
\begin{figure}[h]
    \centering
    \begin{subfigure}{6cm}
		\centering
		\includegraphics[width=5cm]{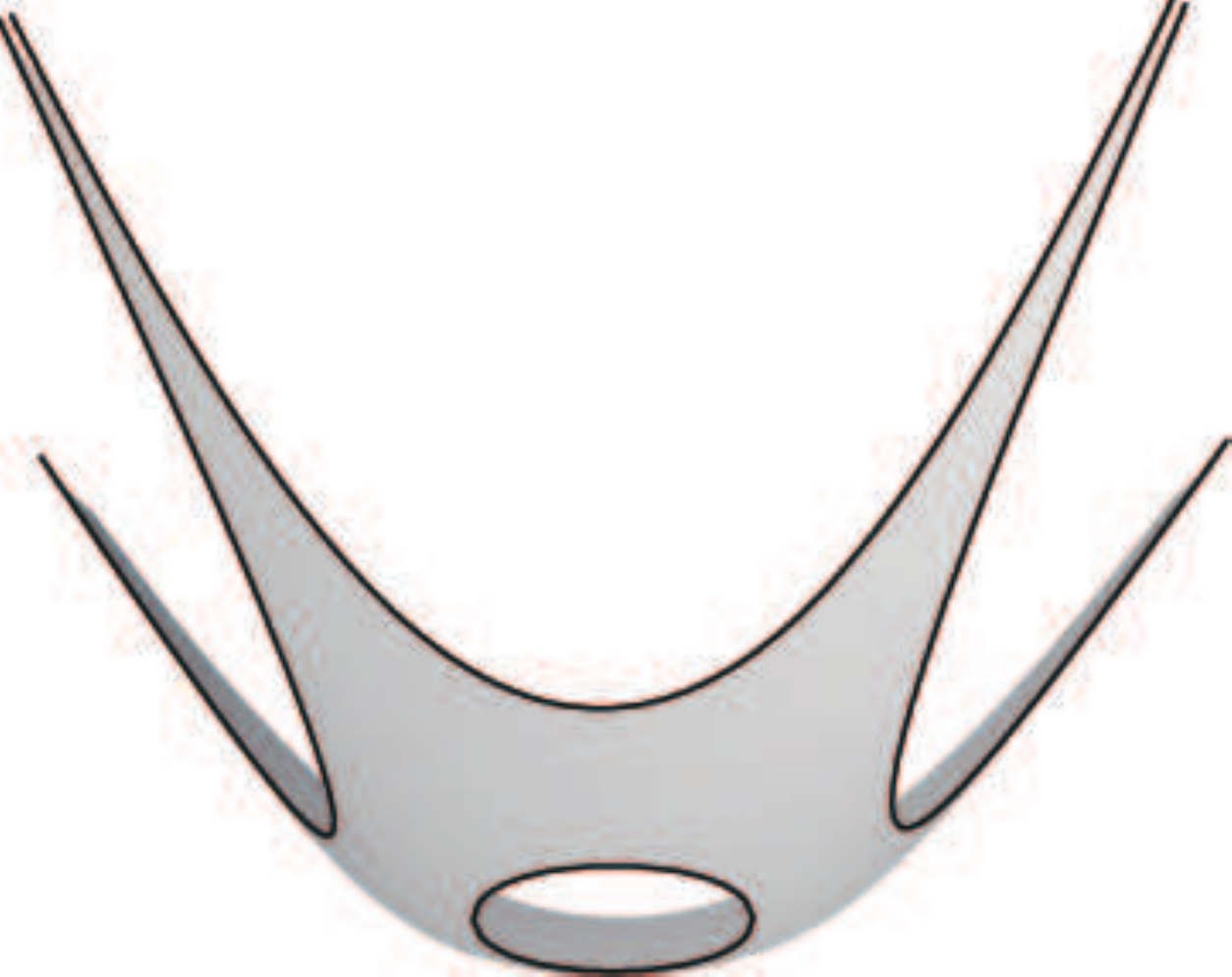}
		\caption{}
	\end{subfigure}
	\begin{subfigure}{6cm}
		\centering
		\includegraphics[width=5cm]{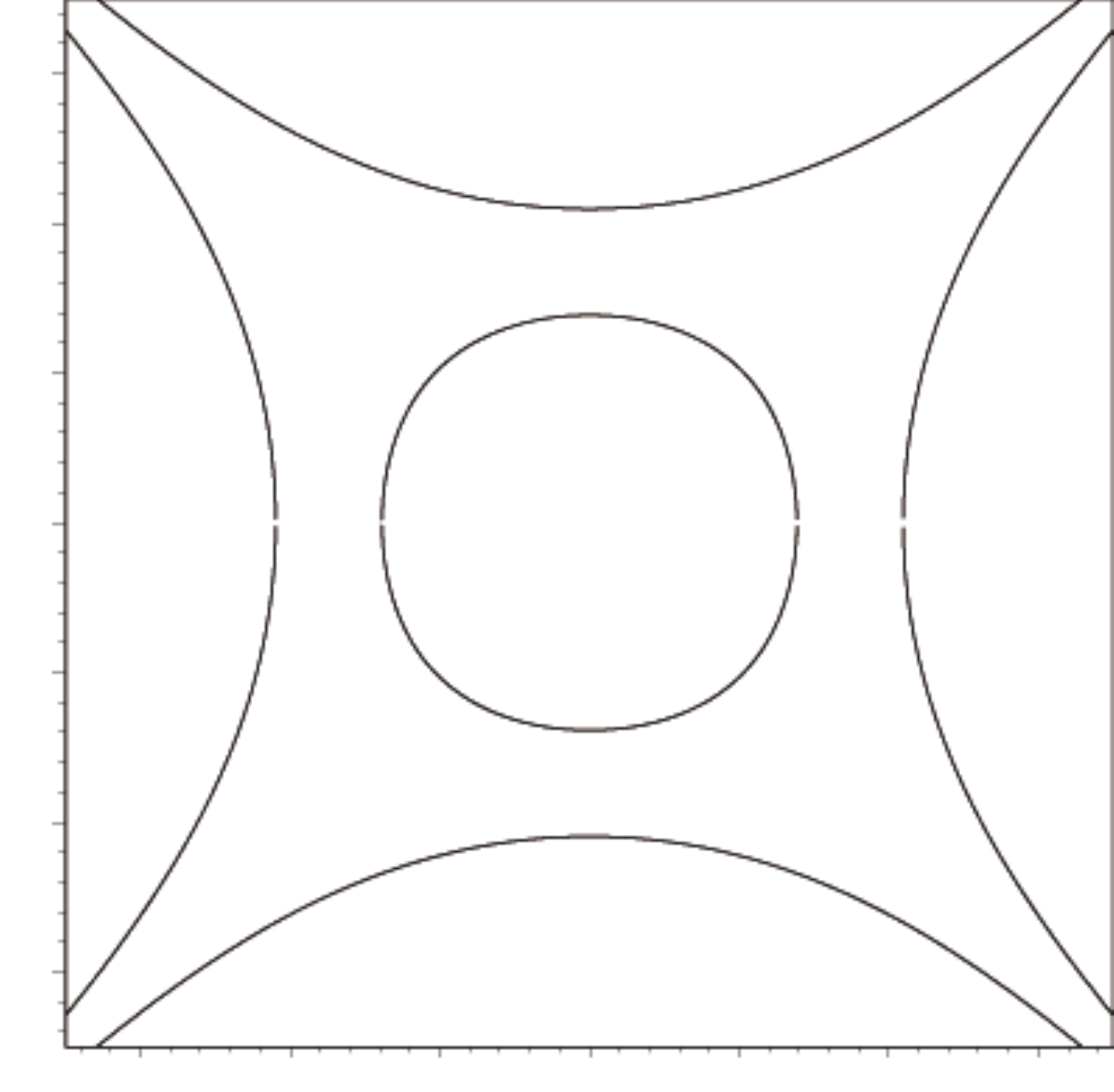}
		\caption{}
	\end{subfigure}
	\caption{The Ronkin function (a) and the amoeba (b) of $\mathbb{F}_0$. Figures are taken from \cite{Kenyon:2003uj} (with slight modifications). The grey part in (a) gives the interior of amoeba in (b), which corresponds to the liquid phase. The hole in (a)/(b) is the gas phase. The unbounded white parts in (a) gives the unbounded complementary regions in (b). They correspond to solid phases.}\label{Ronkinamoeba}
\end{figure}
As we can see, the bounded (unbounded) linear facets in a Ronkin function correspond to the bounded (unbounded) complementary regions in the amoeba while the non-linear part in the Ronkin function is projected to the interior of the amoeba. The (minus) Ronkin function is the limit shape of the crystal.
\end{example}

\paragraph{Quiver quantum mechanics} Let us think of $\mathcal{X}$ as a fibre bundle of $\mathbb{T}^2\times\mathbb{R}$ over the $\mathbb{R}^3$ base. Then we can recover the quiver and brane tiling by performing T-dualities along the $\mathbb{T}^2$ directions. The low energy effective 1d quantum mechanics is in fact the dimensional reduction of the 4d $\mathcal{N}=1$ gauge theory. In the toric diagram, its boundaries specify the singular loci where the $\mathbb{T}^2$ fibre degenrates to a circle. This becomes the NS5-branes under T-dualities, which are (straightened) zig-zag paths on the brane tiling. The D0s become D2s wrapping the whole torus while the D2s are still D2s but restricted in certain domains separated by the NS5 branes. In some of these domains, there would also be NS5s stretched parallel to the D2s. Based on the charges of these NS5s, different domains correspond to black/white nodes and faces in the dimer model. Thus, we can also get the quiver as the dual graph of the dimer. Readers are referred to Figure 1 and 2 in \cite{Ooguri:2009ijd} for illustrations.

For the single D6 brane, as it fills the whole Gorenstein 3-fold, it will become a point on the torus after T-dualities. Hence, it acts as a flavour brane and there is a flavour node added to the quiver. Likewise, the D4s will become flavour D2-branes which are again points on the torus. These would lead to flavour D4-nodes in the quiver diagram \cite{Nishinaka:2013mba}.

\section{Mahler Measure in Quiver Gauge Theories}\label{mahlerquiver}
After going over some fundamentals of Mahler measures and dimer models, we can now study the roles Mahler measures play in quiver gauge theories.

As discussed in \eqref{Pkp}, we can recast Newton polynomials into form (up to shifting the Newton polygon and/or overall multiplication of sign)
\begin{equation}
 P(z,w)=k-p(z,w) \ ,    \qquad
 k>0,
\end{equation}
where $p(z,w)$ has no constant terms and no free parameters. When we start to increase $k$, holes might appear in the amoeba of $P$.
For any dimer, let us call the weights $2\sin(\pi R_I/2)$ the {\bf canonical weight choice}. Nevertheless, let us start with a more general set-up where all coefficients $c_{(m,n)}$ depend on $k$ in the following definition.
\begin{definition}
For a spectral curve associated to a dimer with one free parameter $k$, write it as $P(z,w)=\sum\limits_{(m,n)}c_{(m,n)}(k)z^mw^n$. The {\bf isoradial limit} is defined to be $k=k_\textup{iso}$ such that $c_{(m,n)}(k)$ agrees with the coefficients from the canonical choice.
\end{definition}
\begin{remark}
So far, the canonical choice only has special properties for isoradial dimers. As we will see in \S\ref{Seiberg}, the weights $2\sin(\pi R_I/2)$ which physically come from R-charges also have interesting features for non-isoradial dimers. One may then view a non-isoradial dimer as some sort of ``isoradial dimer'' with ``zero'' or ``negative'' edge lengths. Therefore, we shall always call it the isoradial limit as long as the edge weights follow the canonical choice $2\sin(\pi R_I/2)$ for any dimer regardless of its isoradiality.
\end{remark}
\begin{remark}
For $P(z,w)=k-p(z,w)$, the amoeba would have $g=0$ when $k\leq k_\textup{iso}$. If $k>k_\textup{iso}$, the holes would emerge in the amoeba. In particular, the number of holes would always be the same as the number of interior points in the Newton polygon. These holes would evolve simultaneously when we vary $k$.
\end{remark}

Another interesting limit in the parametrization $P=k-p$ would be the large $k$ limit. First, we introduce a well-known concept.
\begin{definition}
The {\bf Hausdorff distance} between two closed sets $A,B\subset\mathbb{R}^n$ is
\begin{equation}
    d_H=\max\{\sup_{a\in A}(d_E(a,B)),\sup_{b\in B}(d_E(b,A))\},
\end{equation}
where $d_E$ is the usual Euclidean distance.
\end{definition}
We shall define a tropical limit using Hausdorff distance.
\begin{definition}
For a Harnack curve associated to a dimer with one free parameter $k$, write it as $P(z,w)=\sum\limits_{(m,n)}c_{(m,n)}(k)z^mw^n$. Denote its amoeba to be $\mathcal{A}(k)$ and the spine to be $\mathcal{S}$. The {\bf tropical limit} is $k=k_\textup{trop}$ such that $d_H(\mathcal{A}(k),\mathcal{S})$ is minimized at $k_\textup{trop}$.
\end{definition}

\begin{proposition}
For $P(z,w)=k-p(z,w)$, $k\rightarrow\infty$ is a tropical limit.
\end{proposition}
\begin{proof}
Consider the modified amoeba defined by
\begin{equation}
    \mathcal{A}_k:(z,w)\mapsto(\log_k|z|,\log_k|w|)=\frac{1}{\log(k)}(\log|z|,\log|w|).
\end{equation}
Therefore, $\mathcal{A}_k$ is just a rescaling of $\mathcal{A}$, i.e., $\mathcal{A}$ and $\mathcal{A}_k$ are similar in the sense of Euclidean geometry. Following Maslov dequantization in \cite{mikhalkin2004decomposition} (see also \cite{bogaardintroduction}), we learn that $\mathcal{A}_k$ converges to the tropical curve trop$(P)$ when $k\rightarrow\infty$. Hence, $\mathcal{A}$ retracts to its spine when $k\rightarrow\infty$.
\end{proof}

\begin{remark}
Notice that $k\rightarrow\infty$ is {\bf a} tropical limit, but may or may not be {\bf the} only tropical limit. See Figure \ref{flow_amoeba}(b) for example. Nevertheless, in this paper, we will mainly focus on the tropical limit at infinity.
\end{remark}

Since most of the relevant objects diverge at $k_\text{trop}=\infty$, we will mainly discuss sufficiently large $k$.
\begin{definition}
Given an amoeba $\mathcal{A}_P$, denote the set of all vertices $v_i$ of the spine as $\mathcal{V}$. Let $V\subset\mathcal{V}$ be a non-empty proper subset of $\mathcal{V}$. We say that $\mathcal{A}_P$ is {\bf locally} an amoeba $\mathcal{A}_\textup{loc}$ around $V$ if in a neighbourhood of $V$, $P$ can be approximated by dropping some of its terms. The dropped terms correspond to the vertices $v^*$ in the dual graph that are outside the neighbourhood. Moreover, the approximated Newton polynomial has amoeba $\mathcal{A}_\textup{loc}$.

Let $\mathcal{P}(\mathcal{V})\neq\{\mathcal{V}\}$ be a non-trivial partition of $\mathcal{V}$. If $\mathcal{A}_P$ is locally some $\mathcal{A}_\textup{loc}$ for every $V\in\mathcal{P}(\mathcal{V})$, then we say $k$ is {\bf subtropical}. If $\mathcal{P}(\mathcal{V})=\{\{v_1\},\{v_2\},\dots,\{v_n\}\}$, i.e., there is a local amoeba around every single vertex $v_i$, then we say $k$ is {\bf high-subtropical}\footnote{Analogous to the term ``tropical'', we also borrow words ``subtropics'' and ``subtropical high'' from climate science.}.
\end{definition}

\subsection{The Mahler Flow}\label{mahlerflow}
As mentioned above, the Newton polynomials here are constructed by writing down an initial $P_\text{iso}=k_\text{iso}-p(z,w)$ with certain choice of edge weights. Then we simply vary $k$ to get a family of curves.

When varying $k$, the Mahler measure changes continuously. We shall refer to this as the {\bf Mahler flow}. As discussed before, $k_{\min}$ gives the isoradial limit while a tropical limit is reached when $k\rightarrow\infty$.

Recall that when $k>\max(|p(z,w)|)=p(1,1)$, we can compute $u_0(k)$ from \eqref{u0integral} using its Taylor expansion just like $m(P)$. It is not hard to see that
\begin{equation}
    \frac{\text{d}m(P)}{\text{d}\log k}=u_0(k).\label{mfloweqn}
\end{equation}

It is worth noting that this {\bf Mahler flow equation} has the same form as the RG equation, where the coupling and $\beta$-function are replaced by $m(P)$ and $u_0(k)$ respectively, and $k$ controls the scale here instead of energy scale. We may also integrate the Mahler flow equation to get
\begin{equation}
    \int_{m_{k_0}}^{m_{k_1}}\frac{\text{d}m(P)}{u_0(k)}=\log\frac{k_1}{k_0},
\end{equation}
where we shall take $k_0=p(1,1)$. Then from their Taylor expansions, we can see that $u_0(k)$ grows no faster than $m(P)$. As a result, both sides would diverge at large $k$.

For $k\leq p(1,1)$, the behaviour of $u_0(k)$ could be very different. In fact, the right hand side of \eqref{mfloweqn} is not the same $u_0$ as the integral of $\frac{1}{1-k^{-1}p}$ any more. This is because the Taylor expansion has radius of convergence $k>p(1,1)$. Consequently, the right hand side in \eqref{mfloweqn} is no longer a period of the elliptic curve for sufficiently small $k$. Hence, we shall take \eqref{mfloweqn} as the definition of $u_0(k)$ for \emph{any} $k$.

When $k>p(1,1)$, it is straightforward to see that the left hand side of \eqref{mfloweqn} is positive since $u_0(k)$ should be positive as a period. Its positivity can also been seen from its Taylor expansion $u_0(k)=\sum\limits_{n=0}^\infty\frac{p^n(z,w)}{k^n}$. Therefore\footnote{Alternatively, we may also take the derivative of $m(P)$ with respect to $k$:
\begin{equation}
    \frac{\text{d}}{\text{d}k}m(P)=\int_0^1\int_0^1\frac{\text{d}}{\text{d}k}\left(\log|k-p|\right)\text{d}\theta\text{d}\phi.\label{dmdk}
\end{equation}
Now,
\begin{equation}
    \frac{\text{d}}{\text{d}k}\log|k-p|=\frac{1}{|k-p|}\frac{\text{d}}{\text{d}k}\left((k-p)(k-\bar{p})\right)^{1/2}=\frac{k-\text{Re}(p)}{|k-p|^2}.
\end{equation}
This leads to the same result as $\max(|p|)=\max(\text{Re}(p))=p(1,1)$.},
\begin{lemma}
The Mahler measure strictly increases when $k$ increases along the Mahler flow, from $k_0=\max\limits_{|z|=|w|=1}(|p(z,w)|)$ to $k\rightarrow\infty$.
\end{lemma}

In many cases, $k_\text{iso}<k_0$. In terms of amoeba, this means that the holes would not open up at the origin. Then we can always shift the amoeba by
\begin{equation}
    (\log|z|,\log|w|)\rightarrow(\log|z|-\log a,\log|w|-\log b)
\end{equation}
for some positive numbers $a$ and $b$ such that a node is moved to the origin. This gives a rescaling/redefinition of the variables in $P(z,w)$ \footnote{As this is just a shift of the amoeba, the 2-to-1 property between the spectral curve and amoeba still holds. Hence, the curve is still Harnack.}:
\begin{equation}
    k_\text{iso}-p(z,w)\rightarrow k_\text{iso}-p(z/a,w/b).\label{rescale}
\end{equation}
Since $k_\text{iso}-p(z,w)=0$ is Harnack, the pair $(a,b)$ is unique by the 1-to-1 property between amoeba and spectral curve at nodes. In other words, there is a unique solution to $k_\text{iso}-\Tilde{p}(z,w)=0$ where $\Tilde{p}(z,w)\equiv p(z/a,w/b)$. As the parametrization of $k-\Tilde{p}$ (for fixed $\Tilde{p}$) is continuous, we can see that this unique solution is given by
\begin{equation}
    k_\text{iso}=\max_{|z|=|w|=1}(|\Tilde{p}(z,w)|)=\Tilde{p}(1,1)=p(1/a,1/b).
\end{equation}
More importantly, this shows that a hole would now open up at the origin for the amoeba when $k$ gets increased from $k_\text{iso}$, and we have $k\geq\max\limits_{|z|=|w|=1}(|\Tilde{p}|)$ for any $k\geq k_\text{iso}$. Therefore, we can now rewrite the above lemma as
\begin{proposition}
The Mahler measure (with possible rescaling of variables) strictly increases when $k$ increases along the Mahler flow, from $k_\textup{iso}$ to $k\rightarrow\infty$.\label{mflowprop}
\end{proposition}

\begin{remark}
Although the Mahler measure would vary under the rescaling of $z,w$, this is essentially a translation on the $xy$-plane for the Ronkin function as the amoeba is shifted. The change of Mahler measure is just indicating different points on $R(x,y)$. Therefore, the physics would not change. The partition function, which is the volume under $R(x,y)$, remains invariant under the rescaling of variables.
\end{remark}

More generally, we may also consider any $k>0$ (without any rescaling of variables) and compute the integration for Mahler measure numerically. Although the spectral curve is non-Harnack and hence the correspondence between solid/liquid phases and regions of amoeba is not clear, we find that $m(P)$ always increases monotonically for all Newton polynomials $P(z,w)$ we have encountered along the Mahler flow. Thus, we are led to:
\begin{conjecture}
The Mahler measure monotonically increases when $k$ increases along the Mahler flow, from $k=0$ to $k \rightarrow \infty$.\label{mflowconj1}
\end{conjecture}
 
From the viewpoint of crystal melting in the thermodynamic limit, it is natural to expect the increasing of Mahler measure when we increase $k$ since more atoms are removed from the crystal. In terms of amoeba, the size of the hole is controlled by the value of $k$. When $k$ increases, the hole would also become larger and larger, which is consistent with the growing gas phases. When $k$ is (sub)tropical, the gas phase would become dominant. Moreover, $k_\text{iso}$ is the critical point for the existence of the holes/gas phases. The holes would open up for $k>k_\text{iso}$ (even if the holes do not appear from the origin). For $k\leq k_\text{iso}$ (though only $k_\text{iso}$ gives a Harnack curve), the amoeba is of genus zero, and its area would become larger when increasing $k$.

Moreover, the partition function for the crystal melting model should also become larger when increase $k$ \footnote{Again, for $k<k_\text{iso}$, the physical interpretation of Ronkin functions, in particular for different phases, may not be clear. Nevertheless, this would still make sense mathematically. More importantly, it is still possible that Ronkin functions for non-Harnack curves are closely related to crystal melting etc in physics, but in a more subtle way.}. In terms of \eqref{ZBPS}, this implies that we may extend the above conjecture to Ronkin functions.
\begin{conjecture}
The Ronkin function $R(x,y)$ (for any fixed $(x,y)$) does not decrease when we increase $k$ along the Mahler flow for $k>0$. More precisely, when $k_2>k_1$,
\begin{equation}
    \begin{cases}
        R_{k_2}(x,y)>R_{k_1}(x,y),&(x,y)\text{ in a non-linear or bounded linear region for }k_2;\\
        R_{k_2}(x,y)=R_{k_1}(x,y),&(x,y)\text{ in an unbounded linear facet for }k_2.
    \end{cases}
\end{equation}
Notice that for $k_2>k_1\geq k_\textup{iso}$, the non-linear region is the liquid phase and a bounded (unbounded) facet is a gas (solid) phase.\label{mflowconj2}
\end{conjecture}

Since the Ronkin function at $(x,y)$ is essentially the Mahler measure for $P(\text{e}^xz,\text{e}^yw)$, $R(x,y)$ for $P(z,w)$ is exactly the Mahler measure for $\Tilde{P}(z,w)=P(\text{e}^xz,\text{e}^yw)$ and shifted amoeba. Hence, we conclude that
\begin{proposition}
Conjecture \ref{mflowconj1} and Conjecture \ref{mflowconj2} are equivalent.
\end{proposition}

\begin{remark}
Notice that the increase of Mahler measure is strict in Proposition \ref{mflowprop} while the increase is monotonic (i.e., only non-decreasing required) in Conjecture \ref{mflowconj1}. The reason for non-strict increasing is more clear in terms of Conjecture \ref{mflowconj2}: the Mahler measure $m(P)=R(0,0)$ may lie in an unbounded linear facet of the Ronkin function.
\end{remark}

\begin{example}
For $\mathbb{F}_0$ with $P=4-z-z^{-1}-w-w^{-1}$, we have
\begin{equation}
    m(P)=\log k-2k^{-2}{}_4F_3\left(1,1,\frac{3}{2},\frac{3}{2};2,2,2;16k^{-2}\right),\quad u_0(k)={}_2F_1\left(\frac{1}{2},\frac{1}{2};1;16k^{-2}\right)
\end{equation}
for $k\geq4$. The detailed steps can be found in Appendix \ref{exF0}. One may check that they satisfy the Mahler flow equation. At $k=k_\textup{iso}=4$, we have
\begin{equation}
    m(P)=\frac{4\mathcal{K}}{\pi};\quad u_0(4)\rightarrow\infty\label{m4}
\end{equation}
where $\mathcal{K}$ is Catalan's constant. At $k\rightarrow k_\textup{trop}=\infty$, we have
\begin{equation}
    m(P)\rightarrow\log k_\textup{trop}=\infty;\quad u_0(k_\textup{trop})=1.
\end{equation}
We can also plot the Mahler flow and $u_0(k)$ from $k_\textup{iso}$ to $k_\textup{trop}$ as follows:
\begin{equation}
    \includegraphics[width=6cm]{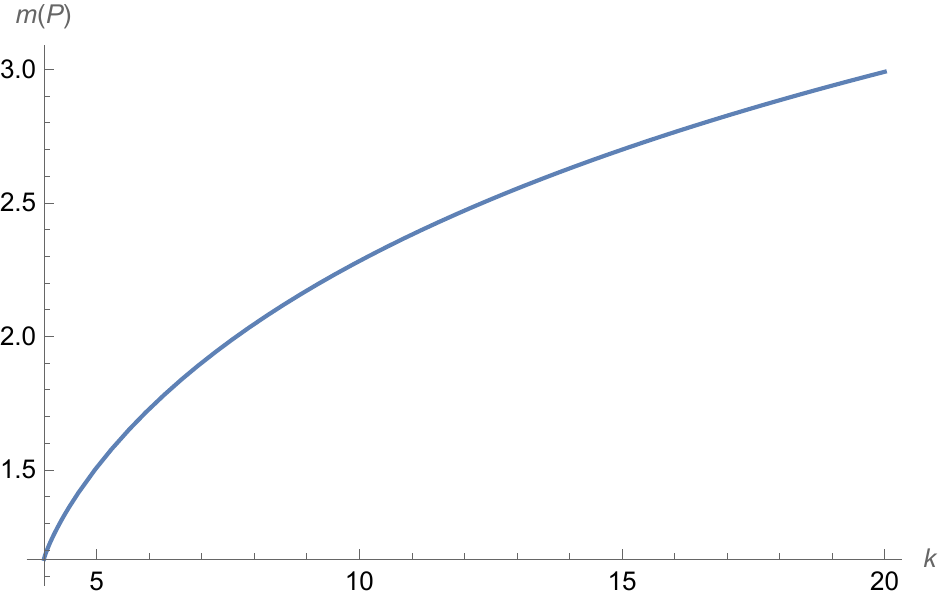},
    \includegraphics[width=6cm]{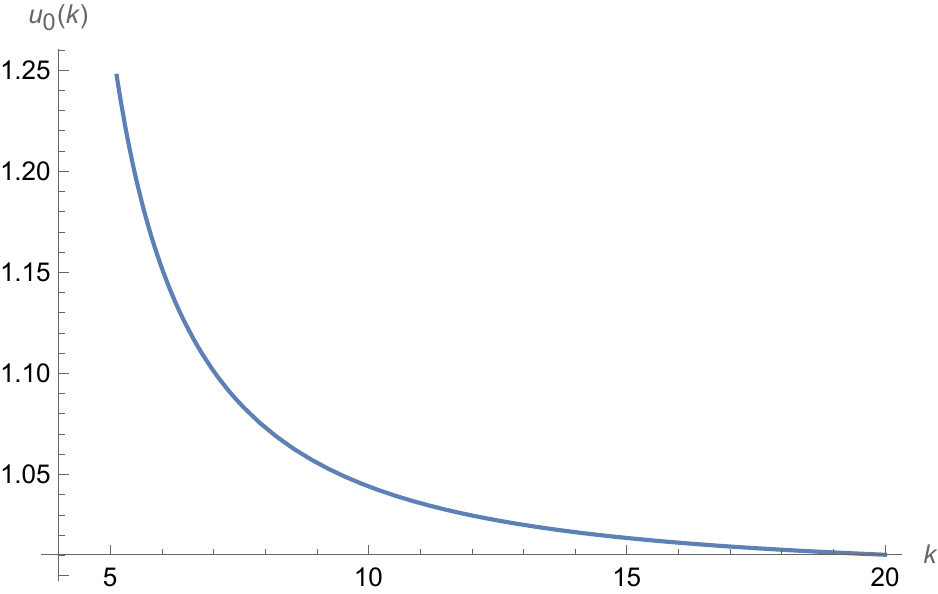}.
\end{equation}
\end{example}
\subsection{Tropical Geometry of the Mahler Flow}\label{holes}
A geometric interpretation of the Mahler flow could be revealed by the holes of the amoeba. In general, it is hard to determine the area(s) of the holes $A_h$. However, when $k$ is sufficiently large, we might be able to calculate $A_h$ using the spines as a tropical limit of the amoeba.

Consider an example, say, $Y^{2,2}$ with vertices of $\Delta$ being $\{
(0,0), (1,0), (0,-1), (1,-1), (-1,-1)
\}$.
The associated  $P=-w-z^{-1}w^{-1}-zw^{-1}-2w^{-1}+k$. For very large $k$, we find that the amoeba is close to its spine as in Figure \ref{hole_ex}(a).
\begin{figure}[h]
    \centering
    \begin{subfigure}{4cm}
		\centering
		\includegraphics[width=4cm]{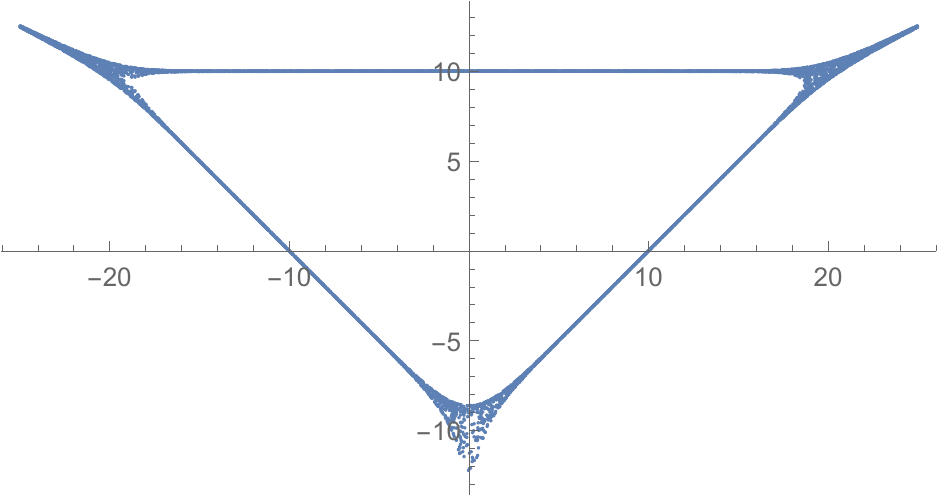}
		\caption{}
	\end{subfigure}
	\begin{subfigure}{5cm}
		\centering
		\tikzset{every picture/.style={line width=0.75pt}}
\begin{tikzpicture}[x=0.75pt,y=0.75pt,yscale=-1,xscale=1]
\draw  [draw opacity=0] (40.5,99.56) -- (200.5,99.56) -- (200.5,179.56) -- (40.5,179.56) -- cycle ; \draw  [color={rgb, 255:red, 155; green, 155; blue, 155 }  ,draw opacity=1 ] (60.5,99.56) -- (60.5,179.56)(80.5,99.56) -- (80.5,179.56)(100.5,99.56) -- (100.5,179.56)(120.5,99.56) -- (120.5,179.56)(140.5,99.56) -- (140.5,179.56)(160.5,99.56) -- (160.5,179.56)(180.5,99.56) -- (180.5,179.56) ; \draw  [color={rgb, 255:red, 155; green, 155; blue, 155 }  ,draw opacity=1 ] (40.5,119.56) -- (200.5,119.56)(40.5,139.56) -- (200.5,139.56)(40.5,159.56) -- (200.5,159.56) ; \draw  [color={rgb, 255:red, 155; green, 155; blue, 155 }  ,draw opacity=1 ] (40.5,99.56) -- (200.5,99.56) -- (200.5,179.56) -- (40.5,179.56) -- cycle ;
\draw [line width=1.5]    (140.5,160.2) -- (120.5,140.2) ;
\draw [line width=1.5]    (120.5,140.2) -- (100.5,160.2) ;
\draw [line width=1.5]    (120.5,120.2) -- (100.5,160.2) ;
\draw [line width=1.5]    (140.5,160.2) -- (120.5,120.2) ;
\draw [line width=1.5]    (120.5,140.2) -- (120.5,120.2) ;
\draw [line width=1.5]    (140.5,160.2) -- (100.5,160.2) ;
\draw [color={rgb, 255:red, 255; green, 0; blue, 0 }  ,draw opacity=1 ][line width=1.5]    (80.5,120.2) -- (160.5,120.2) ;
\draw [color={rgb, 255:red, 255; green, 0; blue, 0 }  ,draw opacity=1 ][line width=1.5]    (120.5,160.2) -- (160.5,120.2) ;
\draw [color={rgb, 255:red, 255; green, 0; blue, 0 }  ,draw opacity=1 ][line width=1.5]    (80.5,120.2) -- (120.5,160.2) ;
\draw [color={rgb, 255:red, 255; green, 0; blue, 0 }  ,draw opacity=1 ][line width=1.5]    (40.5,99.7) -- (80.5,120.2) ;
\draw [color={rgb, 255:red, 255; green, 0; blue, 0 }  ,draw opacity=1 ][line width=1.5]    (160.5,120.2) -- (200.5,99.7) ;
\draw [color={rgb, 255:red, 255; green, 0; blue, 0 }  ,draw opacity=1 ][line width=1.5]    (120.5,159.7) -- (120.5,179.7) ;
\end{tikzpicture}
\caption{}
\end{subfigure}
\begin{subfigure}{4cm}
		\centering
		\includegraphics[width=4cm]{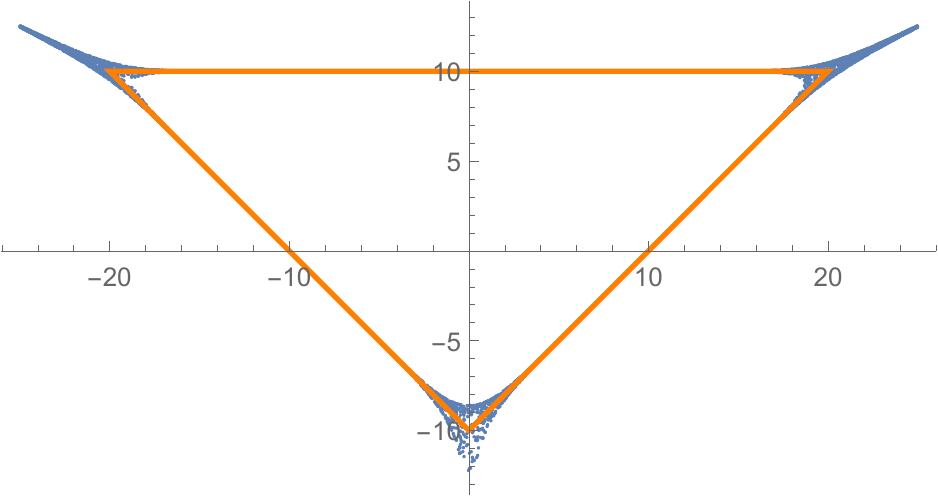}
		\caption{}
	\end{subfigure}
\caption{(a) The amoeba for the Newton polynomial $P(z,w)=-w-z^{-1}w^{-1}-zw^{-1}-2w^{-1}+k$. As an example, we choose $k=\text{e}^{10}$. (b) The spine (in red) is the dual of the triangulated Newton polygon $\Delta$ (in black). (c) The internal triangle in the spine is drawn explicitly in orange in the amoeba. In this example, we can see that the three vertices of the orange triangle are $(20,10)$, $(-20,10)$ and $(0,-10)$ respectively.}\label{hole_ex}
\end{figure}
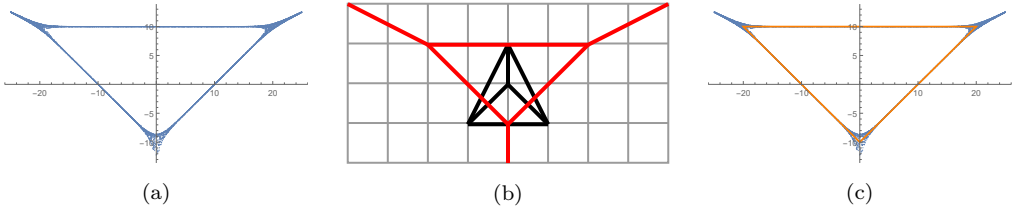
As further shown in Figure \ref{hole_ex}(b), the spine (in red) is the dual of the triangulation (in black) of the Newton polygon $\Delta$. Of particular interest here is the red triangle which is the dual of the three internal lines of the triangulation of $\Delta$ as shown. This is made more clear in Figure \ref{hole_ex}(c): the interior of the orange triangle (i.e., the bounded lines of the spine) consists of the hole and certain parts of the amoeba. At large $k$ here, the hole approaches to this triangle.

Quantitatively, we observe that the three vertices in the spine are $(2\log k,\log k)$, $(-2\log k,\log k)$ and $(0,-\log k)$. 
Note that this is not only true for large $k$ but also for {\bf any} $k$ since this is the consequence of the spine. Around each vertex of the spine, the amoeba locally looks like an amoeba $\mathcal{A}_\text{loc}$ whose Newton polygon is the corresponding subdivision in $\Delta$. These local parts then connect with each other through their thin tentacles. For instance, the upper left part in Figure \ref{hole_ex} is locally a $\mathbb{C}^3$-amoeba. 
For a (global) $\mathbb{C}^3$-amoeba, its tentacles would become thinner and thinner (i.e., asymptotic to the spines) when it goes to infinity. This ensures that the area of the amoeba remains finite ($\pi^2/2$). 

Now, in the $Y^{2,2}$ amoeba, the local $\mathbb{C}^3$ part becomes semi-infinite. The thin finite tentacles will become longer and thinner when $k$ is increased. Therefore, it is natural to conjecture that the area of the local amoeba would be divided equally by the spine. One may check that the sum of areas of the local parts is equal to the area of the whole amoeba since they are all proportional to the areas of the Newton polygons and the local parts correspond to subdivisions of the whole polygon. Then the area of the hole $A_h$ for large $k$ may be computed as
\begin{equation}
    A_h\simeq A(\blacktriangle)-2\times\frac{1}{3}A(\mathcal{A}_{\mathbb{C}^3})-\frac{1}{3}A(\mathcal{A}_{Y^{1,1}}),
\end{equation}
where $\blacktriangle$ denotes the bounded polygon in the spine of $\Delta$ (e.g. the orange triangle here), and $Y^{1,1}$ corresponds to the black triangle at the bottom in the tessellation in Figure \ref{hole_ex}(b). Therefore,
\begin{equation}
    A_h\simeq 4\log^2k-2\times\frac{1}{3}\times\frac{\pi^2}{2}-\frac{1}{3}\times\pi^2=4\log^2k-\frac{2}{3}\pi^2.
\end{equation}

If we increase the coefficient for the term $w^{-1}$, we find that the shape of the amoeba (especially its tentacles) would change as shown in Figure \ref{hole_ex2}(a,b).
\begin{figure}[h]
    \centering
    \begin{subfigure}{4cm}
		\centering
		\includegraphics[width=4cm]{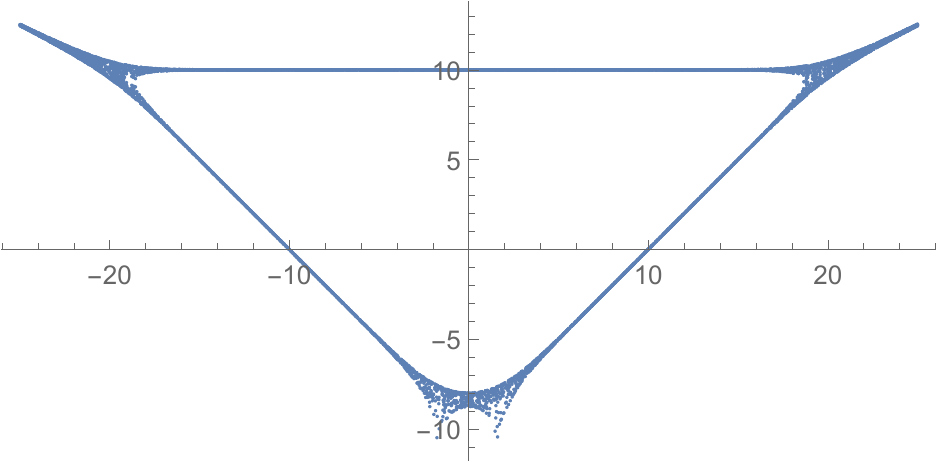}
		\caption{}
	\end{subfigure}
	\begin{subfigure}{4cm}
		\centering
		\includegraphics[width=4cm]{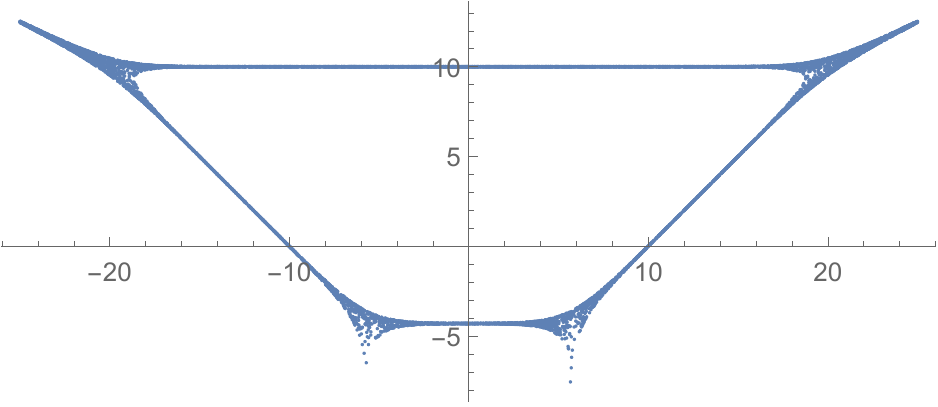}
		\caption{}
	\end{subfigure}
	\begin{subfigure}{5cm}
		\centering
		\tikzset{every picture/.style={line width=0.75pt}}
\begin{tikzpicture}[x=0.75pt,y=0.75pt,yscale=-1,xscale=1]
\draw  [draw opacity=0] (40.5,99.56) -- (200.5,99.56) -- (200.5,179.56) -- (40.5,179.56) -- cycle ; \draw  [color={rgb, 255:red, 155; green, 155; blue, 155 }  ,draw opacity=1 ] (60.5,99.56) -- (60.5,179.56)(80.5,99.56) -- (80.5,179.56)(100.5,99.56) -- (100.5,179.56)(120.5,99.56) -- (120.5,179.56)(140.5,99.56) -- (140.5,179.56)(160.5,99.56) -- (160.5,179.56)(180.5,99.56) -- (180.5,179.56) ; \draw  [color={rgb, 255:red, 155; green, 155; blue, 155 }  ,draw opacity=1 ] (40.5,119.56) -- (200.5,119.56)(40.5,139.56) -- (200.5,139.56)(40.5,159.56) -- (200.5,159.56) ; \draw  [color={rgb, 255:red, 155; green, 155; blue, 155 }  ,draw opacity=1 ] (40.5,99.56) -- (200.5,99.56) -- (200.5,179.56) -- (40.5,179.56) -- cycle ;
\draw [line width=1.5]    (140.5,160.2) -- (120.5,140.2) ;
\draw [line width=1.5]    (120.5,140.2) -- (100.5,160.2) ;
\draw [line width=1.5]    (120.5,120.2) -- (100.5,160.2) ;
\draw [line width=1.5]    (140.5,160.2) -- (120.5,120.2) ;
\draw [line width=1.5]    (120.5,140.2) -- (120.5,120.2) ;
\draw [line width=1.5]    (140.5,160.2) -- (100.5,160.2) ;
\draw [color={rgb, 255:red, 255; green, 0; blue, 0 }  ,draw opacity=1 ][line width=1.5]    (80.5,120.2) -- (160.5,120.2) ;
\draw [color={rgb, 255:red, 255; green, 0; blue, 0 }  ,draw opacity=1 ][line width=1.5]    (127.5,152.75) -- (160.5,120.2) ;
\draw [color={rgb, 255:red, 255; green, 0; blue, 0 }  ,draw opacity=1 ][line width=1.5]    (80.5,120.2) -- (112.5,152.45) ;
\draw [color={rgb, 255:red, 255; green, 0; blue, 0 }  ,draw opacity=1 ][line width=1.5]    (40.5,99.7) -- (80.5,120.2) ;
\draw [color={rgb, 255:red, 255; green, 0; blue, 0 }  ,draw opacity=1 ][line width=1.5]    (160.5,120.2) -- (200.5,99.7) ;
\draw [color={rgb, 255:red, 255; green, 0; blue, 0 }  ,draw opacity=1 ][line width=1.5]    (112.5,152.45) -- (112.5,179.7) ;
\draw [line width=1.5]    (120.5,160.2) -- (120.5,140.2) ;
\draw [color={rgb, 255:red, 255; green, 0; blue, 0 }  ,draw opacity=1 ][line width=1.5]    (112.5,152.45) -- (127.5,152.75) ;
\draw [color={rgb, 255:red, 255; green, 0; blue, 0 }  ,draw opacity=1 ][line width=1.5]    (127.5,152.75) -- (127.5,180) ;
\end{tikzpicture}
	\caption{}
	\end{subfigure}
    \caption{(a) The amoeba for the Newton polynomial $P(z,w)=-w-z^{-1}w^{-1}-zw^{-1}-2\text{e}w^{-1}+\text{e}^{10}$. (b) The amoeba for the Newton polynomial $P(z,w)=-w-z^{-1}w^{-1}-zw^{-1}-2\text{e}^5w^{-1}+\text{e}^{10}$. (c) The spine (in red) is the dual of the triangulated Newton polygon $\Delta$ (in black).}\label{hole_ex2}
\end{figure}
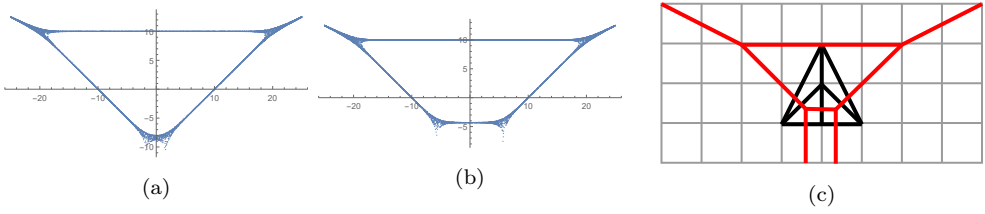
As we can see, the tentacle at the bottom has now been divided into two. The subdivision of the Newton polygon and its dual spine have changed as in Figure \ref{hole_ex2}. Now the bottom local part becomes the triangulated $Y^{1,1}$ for (a) and two $\mathbb{C}^3$'s for (b).

In general, it turns out that for the canonical choice of coefficients, there is no split in the spines caused by boundary lattice points that are not at the corner. In the Ronkin function, this means that the unbounded linear facets would only appear for vertices in the Newton polygon.

We shall make the above discussion more rigorous using the definition for (sub)tropical $k$. The amoebae sketched in Figure \ref{hole_ex}(a) and Figure \ref{hole_ex2}(a,b) all have subtropical $k$. However, only Figure \ref{hole_ex}(a) and Figure \ref{hole_ex2}(b) have high-subtropical $k$. Now consider $P=-w-z^{-1}w^{-1}-zw^{-1}-2w^{-1}+k_\textup{st}$ in Figure \ref{hole_ex}. When $|w|$ is large enough while $|z|$ is small enough such that $\log|w|\sim\mathcal{O}(\log(k_\textup{st}))$, $\log|z|\sim\mathcal{O}(1/\log(k_\textup{st}))$ and $1/|zw|\sim\mathcal{O}(1)$, the Newton polynomial can be approximated by $-w-z^{-1}w^{-1}+k_\textup{st}$. This corresponds to the local $\mathcal{A}_{\mathbb{C}^3}$ at the upper left corner in Figure \ref{hole_ex}(a).

For (high-sub)tropical $k$, we may try to compute $A_h$ for any amoeba using the above method. We first need to determine the vertices of $\blacktriangle$. Let us consider $\mathbb{C}^3$ whose Harnack curve $c_1z+c_2w+k$ (for fixed $c_{1,2}$) as an example. Then by looking at its asymptotic behaviour, we can find its spine as follows:
\begin{equation}
    \begin{split}
        &x\rightarrow0:~y=\log|w|=\log|k/c_2|;\\
        &y\rightarrow0:~x=\log|z|=\log|k/c_1|;\\
        &x,y\rightarrow\infty,z/w\sim\mathcal{O}(1):~y=\frac{x+\log|c_1|}{\log|c_2|}.
    \end{split}
\end{equation}
Hence, the intersection point of the three lines is $(\log|k/c_1|,\log|k/c_2|)$. Now we can use the following theorem \cite{theobald2002computing,bogaardintroduction}.

\begin{theorem}
	For $(\alpha,M)\in(\mathbb{C}^*)^2\rtimes\textup{GL}(2,\mathbb{Z})$ and $P(\bm{z})\equiv P(z,w)\in\mathbb{C}\left[z,w,z^{-1},w^{-1}\right]$, the map $\Psi:(\mathbb{C}^*)^2\rtimes\textup{GL}(2,\mathbb{Z})\rightarrow\textup{Aut}\left(\mathbb{C}\left[z,w,z^{-1},w^{-1}\right]\right)$ defined by $\Psi(\alpha,M)(P(\bm{z}))=P\left(\alpha\cdot\bm{z}^M\right)$ is an isomorphism. Moreover, their Newton polytopes satisfy $\Delta(\Psi(P))=M\cdot\Delta(P)$.
	Denote the amoeba of $P$ as $\mathcal{A}_P$, then for $\det(M)\neq0$, we have $\mathcal{A}_P=M\mathcal{A}_{\Psi(P)}-\textup{Log}(\alpha)$.
\end{theorem}
Then for $M=(M_{ij})$ (and $\alpha=(1,1)$), the curve becomes $P=c_1z^{M_{11}}w^{M_{21}}+c_2z^{M_{12}}w^{M_{22}}+k$. The vertex has coordinates
\begin{equation}
    \frac{1}{\det M}(M_{22}\log|k/c_1|-M_{21}\log|k/c_2|,-M_{12}\log|k/c_1|+M_{11}\log|k/c_2|).
\end{equation}
For instance, the local approximation $-w-z^{-1}w^{-1}+k$ is a GL$(2,\mathbb{Z})$ transformation of $-z-w+k$ given by $\begin{pmatrix}-1&0\\-1&1\end{pmatrix}$. This gives the vertex $(-2\log k,\log k)$ for $Y^{2,2}$ in Figure \ref{hole_ex}. For any local amoebas, the corresponding vertex in the spine can be obtained in this way. Once we know the coordinates of the vertices, we may compute the area of the hole.

\begin{conjecture}
Given a Newton polygon with Newton polynomial $P(z,w)=k-p(z,w)$ and $k$ high-subtropical, the area of a hole (labelled by $i$) in the amoeba reads
\begin{equation}
    A_{h,i}\simeq A(\blacktriangle_i)-\sum_{v_j\in\mathcal{V}_i}\frac{1}{n_j}A(\mathcal{A}_{\textup{loc},j})=A(\blacktriangle_i)-\sum_{v_j\in\mathcal{V}_i}\frac{\pi^2}{n_j}A(\Delta_j),
\end{equation}
where $\blacktriangle_i$ denotes the corresponding separated polygon in the spine and $A(\blacktriangle_i)\propto\log^2(k)$. Moreover, $\mathcal{V}_i$ is the set of spine vertices surrounding the hole $i$. The local amoeba $\mathcal{A}_{\textup{loc},j}$ around the vertex $v_j$ of the spine corresponds to the $n_j$-gon $\Delta_j$ in the tessellation of $\Delta$.

The total area $A_h$ of the holes is then
\begin{equation}
    A_h\simeq A(\blacktriangle)-\sum_{v_j\in\mathcal{V}}\frac{m_j}{n_j}A(\mathcal{A}_{\textup{loc},j})=\sum_iA(\blacktriangle_i)-\sum_{v_j\in\mathcal{V}}\frac{m_j\pi^2}{n_j}A(\Delta_j),
\end{equation}
where $m_j$ is the number of $\blacktriangle_i$'s that have $v_j$ as a vertex, and $\mathcal{V}$ is the set of all vertices of the spine.
\end{conjecture}

\begin{example}
For polygons with a single interior lattice point, we have $m_j=1$, and there is a single $\blacktriangle$.
\end{example}

\begin{remark}
It is often more useful to consider the simplification as follows. For Newton polynomials of form $P(z,w)=k-p(z,w)$ considered in this paper, the (high-sub)tropical $k$ is also the large $k$ limit. We always have the dominating contribution $A_h\sim\log^2(k)$. We may also recast the Mahler flow equation in terms of the area of the hole $\frac{\text{d}m(P)}{\text{d}A_h}$. Then at large $k$, $\frac{\text{d}m(P)}{\text{d}A_h}\sim\frac{1}{2\log k}>0$.
\end{remark}

\paragraph{Integral approximations} As a byproduct, this helps us understand certain integrals in the large $k$ limit. For instance, for $\mathbb{F}_0$, by solving $\text{e}^{y}+\text{e}^{-y}-\text{e}^{x}-\text{e}^{-x}-1=0$, we get part of the boundary of $\mathcal{A}_{\mathbb{F}_0}$ (i.e., one solution to the equation) which reads
\begin{equation}
    y=\log\left(\frac{1}{2}\text{e}^{-x}\left(1+\text{e}^{2x}+k\text{e}^{x}+\sqrt{-4\text{e}^{2x}+(-1-\text{e}^{2x}-k\text{e}^{x})^2}\right)\right).
\end{equation} 
As we can see, this is the upper right boundary of the amoeba:
\begin{equation}
    \includegraphics[width=5cm]{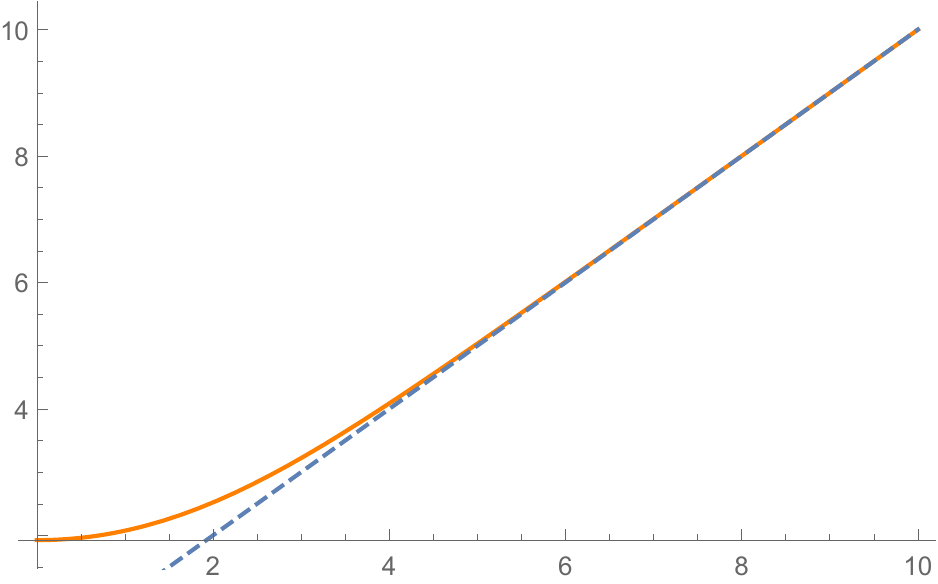},
\end{equation}
where we have used $k=5$ to illustrate this and the dashed line indicates the spine. Therefore, the area of the hole is
\begin{equation}
    A_h=8\left(\int_0^\infty y\text{d}x-\int_0^\infty x\text{d}x\right)-A(\mathcal{A}_{\mathbb{F}_0})=8\int_0^\infty(y-x)\text{d}x-2\pi^2.
\end{equation}
In general, it is not straightforward to determine the large $k$ behaviour for such kind of integral with integrand
\begin{equation}
    I\equiv y-x=\log\left(\frac{1+\text{e}^{2x}+k\text{e}^{x}+\sqrt{-4\text{e}^{2x}+(-1-\text{e}^{2x}-k\text{e}^{x})^2}}{2\text{e}^{2x}}\right)
\end{equation}
because $x$ is also integrated to $\infty$. Nevertheless, from the above analysis, we learn that $A_h\sim\log^2(k)$ for large $k$. Therefore,
\begin{equation}
    \int_0^{\infty}I\,\text{d}x\sim\log^2(k)
\end{equation}
in the large $k$ limit.

\paragraph{Higher dimensions} It would also be natural to conjecture that similar patterns would happen for any dimension $n$.
\begin{conjecture}
Let $P(\bm{z})=k-p(\bm{z})$ be a Laurent polynomial of $\bm{z}\in\mathbb{C}^n$. In the large $k$ limit, we have $V_\textup{bdd}\sim\log^nk$, where $V_\textup{bdd}$ is the volume of a bounded complementary region of the amoeba $\mathcal{A}(P)$.
\end{conjecture}

\subsection{The K\"ahler Parameter}\label{kahler}
So far, the variable $k$ has not been endowed with any further physical interpretations, except being the scale of the Mahler flow and controlling the size of the hole in the amoeba. In this subsection, we shall discuss the physical interpretations of the parameter $k$ in quiver theories.

It is well-known that the variables in the coefficients of $P(z,w)$ are related to K\"ahler moduli of the toric Gorenstein singularity \cite{Hori:2000kt,Hori:2000ck}. In quiver quantum mechanics, these K\"ahler moduli are Fayet-Iliopoulos (FI) parameters. Since every FI parameter is associated with a gauge node in the quiver, or equivalently, a face in the dimer, they should be related to edge weights/energies and magnetic fluxes on the dimer model as pointed out in \cite{Yamazaki:2010fz}.

Hence, it would be natural to relate $k$ to the K\"ahler/FI parameters. However, there are more than one FI paramter in general while we only have one variable $k$ in our Newton polynomial. In fact, the discrepancy between the numbers of parameters are compensated by D-term relations. Let $G$ be the number of nodes and $E$ be the number of arrows in the quivers. In terms of the arrows $X_I$, the D-terms are \cite{Witten:1993yc}
\begin{equation}
    D_i=-e^2\left(\sum_Id_{iI}|X_I|^2-\zeta_i\right),
\end{equation}
where $e$ is the gauge coupling\footnote{In general, the gauge coupling should be $e_i$, but one often sets $e=e_i$ for all $i$ in GLSM.} and $d$ is the incidence matrix\footnote{An incidence matrix has rows and columns represent gauge nodes and arrows in the quiver respectively. The entry is $\pm1$ if the arrow is attached to the node (either coming in or out). Otherwise, it is zero.} of size $G\times E$. Here, $\zeta_i$ denotes the FI parameters. The moment map is reproduced in the matrix form \cite{Feng:2000mi}
\begin{equation}
    \delta\cdot|X_I|^2=\bm{\zeta},
\end{equation}
where $\delta$ is the reduced quiver matrix obtained by removing one row in $d$. Hence, $\delta$ is a $(G-1)\times E$ matrix. This is because there is always one row in $d$ that is redundant. Physically, this indicates the removal of an overall U(1) factor. This can also be related to perfect matching matrix $\mathcal{P}$ via $\delta=Q_D\mathcal{P}^\text{T}$. The matrix $Q_D$ with $(G-1)$ rows encodes the GLSM charges under D-term relations, and hence the name D-term charge matrix.

As the coefficients in the Newton polynomial are obtained from perfect matchings/GLSM fields in the dimer model, there is only one free parameter left with the constraints from $Q_D$. Varying $k$ along the Mahler flow can therefore be interpreted as varying this free parameter. In general, all the coefficients should be functions of these FI parameters while we consider a simplification where we only have one variable $k$ as constant term in this paper.

\paragraph{General coefficients} Let us also make a brief comment on more general choices of coefficients in $P(z,w)$. Since these coefficients are determined by the FI parameters, they would be related to wall crossings in the moduli space of the quiver quantum mechanics. Following the 4d/1d correspondence \cite{Yamazaki:2012cp}, wall crossing corresponds to Seiberg duality in the 4d $\mathcal{N}=1$ gauge theories. As the Ronkin function is closely related to the 4d superconformal index and topological string partition function (see \S\ref{universal} and also \cite{Yamazaki:2012cp}) which are invariant under Seiberg duality and wall crossing respectively, we expect that the Mahler measure and Ronkin function would enjoy certain property under Seiberg duality/wall crossing. In \S\ref{Seiberg}, we will discuss this explicitly for the isoradial limit.

\paragraph{Intercepts of Ronkin functions} Different coefficients in the Newton polynomial would also lead to different intercepts of the linear facets in the Ronkin function. When a solid phase corresponds to a vertex/corner point $(m,n)$ in the Newton polygon, the facet would have slope $(m,n)$ in the Ronkin function. Besides, the Newton polynomial would contain a term $c_{(m,n)}z^mw^n$. According to \cite{passare2000amoebas},
\begin{theorem}
The intercept of this facet is $\log|c_{(m,n)}|$.\label{intercept}
\end{theorem}
As a result, if the coefficients of these linear facets do not equal, their extension would not meet at the origin in the plot of Ronkin function\footnote{If these (extensions of) facets meet at the same point but not the origin, we can always rescale the whole Newton polynomial to make $c_{(m,n)}=1$ and hence shift it to $(0,0,0)$.}. In crystal melting, this means that the unmolten crystal would not have a point as the tip. Instead, the top of the crystal would be some ridge or face. In \cite{Ooguri:2009ijd}, the one-to-one correspondence between crystal melting and dimer model was proven for the crystal with a single atom at the tip. Nevertheless, different coefficients for crystal melting has also been considered in various literature, and it would be natural to expect that this one-to-one correspondence could be extended to such situation for crystal melting\footnote{We should emphasize that the analysis of quiver gauge theories should not be affected since they can be directly read off from the dimers with different edge weights.}. In light of \cite{Kenyon:2003uj}, this would imply a non-trivial magnetic field.

When Newton polynomials cannot be rescaled to have coefficients 1 for corner points, there are actually two different types.
\begin{example}
The Newton polynomial for $L^{1,3,1}/\mathbb{Z}_2$ $(0,1,1,1)$ reads
\begin{equation}
    P=-\frac{AB}{C}(zw+w)-\frac{C}{AB}(zw^{-1}+z^{-2}w^{-1}+3w^{-1}+3z^{-1}w^{-1})-(2z+2z^{-1})+k,
\end{equation}
where
\begin{equation}
    A=\sin\left(\frac{5-\sqrt{7}}{12}\pi\right),B=\sin^2\left(\frac{5-\sqrt{7}}{6}\pi\right),C=\sin^3\left(\frac{1+\sqrt{7}}{12}\pi\right).\label{L131Z2}
\end{equation}
Although it does not have same coefficients for the corner points, we can absorb the extra factor by $\frac{AB}{C}w\rightarrow w$ so that the four coefficients would all become $1$. As we can see, this is a rescaling of $z,w$.

However, there is another type whose coefficients cannot be made the same even with such rescaling. For instance, the Newton polynomial for $\textup{PdP}_{4a}$ reads
\begin{equation}
    P=-B_1^2B_2^2(z+z^{-1})-B_2^2B_3B_4(zw^{-1}+w^{-1})-B_5^4z^{-1}w^2-2B_1B_2B_5^2(z^{-1}w+w)+k,
\end{equation}
where $B_i=\sin(\pi b_i/2)$ and
\begin{equation}
    \begin{split}
        &b_1\approx0.427\text{ is a root of }3x^3-340x^2-24x+72=0;\\
        &b_2\approx0.725\text{ is a root of }3x^3-134x^2+228x-96=0;\\
        &b_3\approx0.298\text{ is a root of }3x^3+206x^2-384x+96=0;\\
        &b_4\approx0.596\text{ is a root of }3x^3+412x^2-1536x+768=0;\\
        &b_5\approx0.550\text{ is a root of }3x^3+250x^2-124x-8=0.
    \end{split}\label{PdP4a}
\end{equation}\label{twotypes}
It turns out that no matter how we rescale the variables, the five vertices $z$, $z^{-1}$, $zw^{-1}$, $w^{-1}$ and $z^{-1}w^2$ would not have same coefficients.
\end{example}

\subsection{Isoradial Limit}\label{isolim}
Now, let us focus on one of the special points, that is, $k_\text{iso}$, in the Mahler flow. When the embedding is isoradial with $2\sin(\theta_I)$ as (canonical) edge weights, we shall call $m(P)$ and  $R(x,y)$ the isoradial Mahler measure and Ronkin function respectively. 
For a different choice of coefficients in the Newton polynomial, it is still possible to have an isoradial point $k=k_\text{iso}$. Now we will show that there is a set of $\theta_I^*$ that maximizes the isoradial Mahler measure and this coincides with the R-charges in $a$-maximization.

More importantly, after our discussion on Seiberg duality in \S\ref{Seiberg}, we will find that the followings in this subsection can be applied to any toric quivers and brane tilings regardless of isoradiality.

\subsubsection{Isoradial Mahler measure}\label{isom}
It is useful to introduce a simple and remarkable formula for isoradial Mahler measure obtained in \cite{de2007partition}:
\begin{equation}
    m(P)=\sum_I\left(\frac{\theta_I}{\pi}\log(2\sin(\theta_I))+\frac{1}{\pi}\Lambda(\theta_I)\right),\label{isomahler}
\end{equation}
where
\begin{equation}
\begin{split}
    \Lambda(x)&=-\int_0^x\log(2\sin(t))\text{d}t\\
    &=\frac{i}{12}\pi^2-\frac{i}{2}x^2+x\log(1-\text{e}^{2ix})-x\log(\sin(x))-\frac{i}{2}\text{Li}_2(\text{e}^{2ix})
\end{split}
\end{equation}
is the {\it Lobachevsky function}. 
In terms of the energy functions, we may also write \eqref{isomahler} as
\begin{equation}
    m(P)=-\frac{1}{\pi}\sum_{I=1}^m\left(\theta_I\mathcal{E}(e_I)+\int_0^{\theta_I}\log(2\sin(t))\text{d}t\right).
\end{equation}
It is important to notice that
\begin{equation}
    \frac{\partial m(P)}{\partial\theta_I}=\frac{\theta_I}{\pi}\cot(\theta_I).
\end{equation}
\begin{example}
Consider our running example of $m(4-z-z^{-1}-w-w^{-1})$ for $\mathbb{F}_0$. Since $\theta_I=\pi/4$ for all $I$, we get
\begin{equation}
    \begin{split}
        m(P)&=\frac{1}{(2\pi i)^2}\int_0^1\int_0^1\left(\log(8-2z-2z^{-1}-2w-2w^{-1})-\log(2)\right)\frac{\textup{d}z}{z}\frac{\textup{d}w}{w}\\
        &=m(8-2z-2z^{-1}-2w-2w^{-1})-\frac{\log(2)}{(2\pi i)^2}\int_0^1\int_0^1\frac{\textup{d}z}{z}\frac{\textup{d}w}{w}\\
        &=8\left(\frac{\pi}{4\pi}\log(2\sin(\pi/4))+\frac{1}{\pi}\Lambda(\pi/4)\right)-\log(2)\\
        &=\frac{4\mathcal{K}}{\pi}.
    \end{split}
\end{equation}
This agrees with our result in \eqref{m4}.
\end{example}

\subsubsection{Maximization of Isoradial Mahler Measure}\label{isomax}
When finding R-charges under RG trajectory in 4d, we need to maximize the $a$-function,
\begin{equation}
    a=\sum_I(R_I-1)^3=\sum_I\left(\frac{2\theta_I}{\pi}-1\right)^3,
\end{equation}
where we have normalized the coefficient and omitted the part from number of quiver nodes (which corresponds to faces in the dimer). On the other hand, the isoradial Mahler measure is also explicitly given in terms of $\theta_I$ by \eqref{isomahler}. 
As $\theta_I$ is a rhombus angle in the dimer model, we have $\theta_I\in[0,\pi/2]$. The shapes of the two functions are sketched in Figure \ref{mahlera}, heuristically drawn against 1 and 2 of the $\theta_I$ angles.
\begin{figure}[h]
    \centering
    \begin{subfigure}{6cm}
		\centering
		\includegraphics[width=6cm]{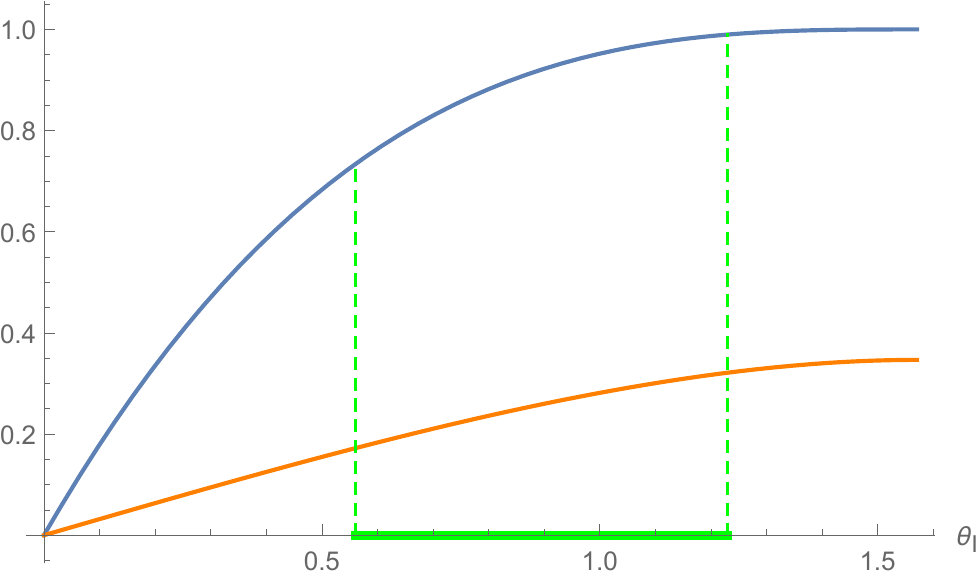}
		\caption{}
	\end{subfigure}
	\begin{subfigure}{6cm}
		\centering
		\includegraphics[width=6cm]{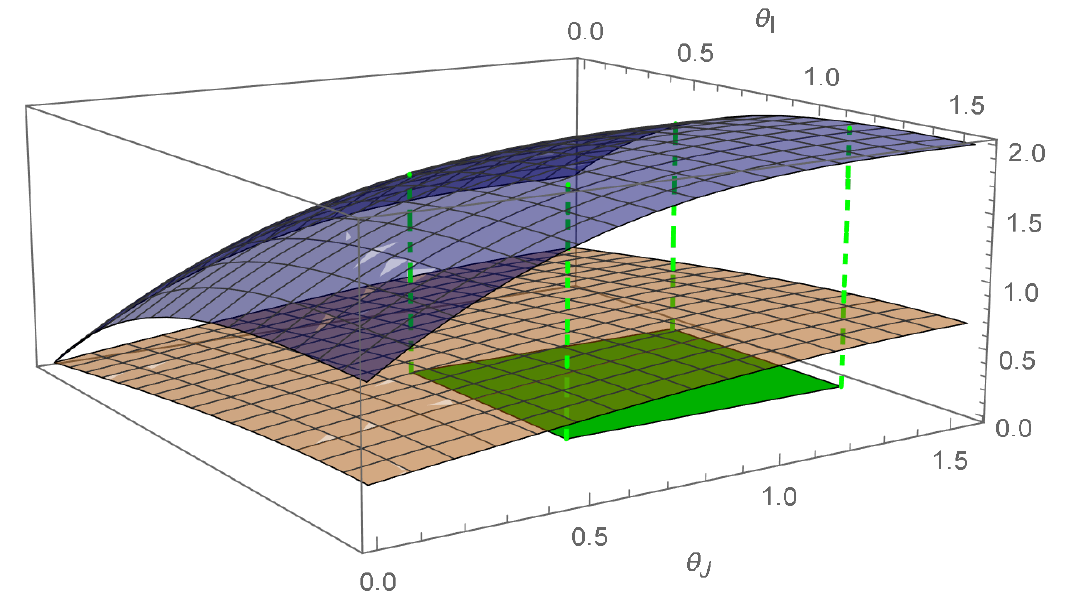}
		\caption{}
	\end{subfigure}
	\caption{The $a$-function (blue) and isoradial Mahler measure (orange) plotted as functions of $\theta_I$. We can only visualize the (a) 1d and (b) 2d versions while physically there should be at least 3 $\theta_I$'s ($R_I$'s) in our theory. Notice that we have shifted the $a$-function by (a) +1 and (b) +2 for better visualization and comparison. The green region is the conformal manifold on the $\theta_I$ axis/$\theta_I$-$\theta_J$ plane.}\label{mahlera}
\end{figure}
The derivatives for the two functions are respectively
\begin{equation}
    \frac{\partial a}{\partial\theta_I}=\frac{6}{\pi}\left(\frac{2\theta_I}{\pi}-1\right)^2>0 \ ,
    \quad 
    \frac{\partial m_\text{iso}}{\partial\theta_I}=\frac{\theta_I}{\pi}\cot(\theta_I)>0 \ ;
    \qquad
    \theta_I\in[0,\pi/2] \ .    
\end{equation}
Hence, both of the functions are strictly increasing for any R-charges. For $a$-maximization, the rhombus angles are also subject to the conditions $\sum2\theta_I=2\pi$ and $\sum(\pi-2\theta_I)=2\pi$ as discussed in Sec.~\ref{isodimer}. 
The region of rhombus angles/R-charges that satisfy these conditions is the so-called conformal manifold $\mathcal{M}$. 
Pictorially, we sketch the conformal manifold in green in the $\theta_I$-space in Figure \ref{mahlera}, and the dashed lines cut out the possible values for $a$ and $m_\text{iso}$.
In other words, we need to find the maximum of $a$ for the set of points $\{(\theta_1,\dots,\theta_n)\in\mathcal{M}\}$. Suppose that $a$ is maximized at $(\theta_1^*,\dots,\theta_n^*)$, then the isoradial Mahler measure would also reach its maximum at $(\theta_1^*,\dots,\theta_n^*)$ on the conformal manifold as $a$ and $m_\text{iso}$ are both monotonically increasing. Therefore, we have shown that
\begin{proposition}\label{propmax}
For toric quiver gauge theories, we have
\begin{equation}
    m_\textup{iso}\text{-maximization}=a\text{-maximization}.
\end{equation}
\end{proposition}

\begin{remark}
Following this proposition, we may also conclude that $m_\textup{iso}$-maximization is equivalent to volume minimization \cite{Gubser:1998vd} and K-semistability (for product test configurations) \cite{Collins:2012dh}. We will not expound upon this here, and readers are referred to \cite{Collins:2012dh,collins2019sasaki,Collins:2016icw,Bao:2020ugf} for more details.
\end{remark}

\begin{remark}
Since the weights for the dimer edges are always non-negative for any trial R-charges, the spectral curve is always Harnack. Hence, with different trial R-charges, the amoeba is deformed but its area is invariant. Furthermore, it would always be genus $0$.
\end{remark}

We emphasize that this proposition is still valid for non-isoradial dimers. Though for the latter we do not have a well-defined rhombus angle, but we can write the edge weights as $2\sin(\pi R_I/2)$. In other words, we replace $\theta_I$ by $\pi R_I/2$ in \eqref{isomahler}. Then \eqref{isomahler}, which we shall still call $m_\text{iso}$ for convenience, again reaches its maximum at $R_I^*$, where $R_I^*$ maximizes the $a$-function. Therefore, it is always true that $m_\text{iso}$-maximization is equivalent to $a$-maximization\footnote{Strictly speaking, so far we can only say that maximizing \eqref{isomahler} (rather than $m_\text{iso}$) is equivalent to $a$-maximization as we have not shown that \eqref{isomahler} gives the correct Mahler measures for non-isoradial embeddings with canonical weight choices. We will come back to this point in \S\ref{Seiberg}.}.

\begin{example}
Let us vary the shape of the isoradial dimer for $\mathbb{F}_0$ in the following way:
\begin{equation}
    \tikzset{every picture/.style={line width=0.75pt}}
    \begin{tikzpicture}[x=0.75pt,y=0.75pt,yscale=-1,xscale=1]
\draw  [color={rgb, 255:red, 155; green, 155; blue, 155 }  ,draw opacity=1 ][dash pattern={on 4.5pt off 4.5pt}] (506.94,126.85) .. controls (506.94,98.21) and (529.77,75) .. (557.93,75) .. controls (586.09,75) and (608.92,98.21) .. (608.92,126.85) .. controls (608.92,155.49) and (586.09,178.7) .. (557.93,178.7) .. controls (529.77,178.7) and (506.94,155.49) .. (506.94,126.85) -- cycle ;
\draw  [line width=1.5]  (575.75,79.6) -- (575.75,174.1) -- (540.1,174.1) -- (540.1,79.6) -- cycle ;
\draw  [color={rgb, 255:red, 155; green, 155; blue, 155 }  ,draw opacity=1 ][dash pattern={on 4.5pt off 4.5pt}] (391.94,125.85) .. controls (391.94,97.21) and (414.77,74) .. (442.93,74) .. controls (471.09,74) and (493.92,97.21) .. (493.92,125.85) .. controls (493.92,154.49) and (471.09,177.7) .. (442.93,177.7) .. controls (414.77,177.7) and (391.94,154.49) .. (391.94,125.85) -- cycle ;
\draw  [line width=1.5]  (465.24,82.35) -- (465.24,169.35) -- (420.62,169.35) -- (420.62,82.35) -- cycle ;
\draw  [color={rgb, 255:red, 155; green, 155; blue, 155 }  ,draw opacity=1 ][dash pattern={on 4.5pt off 4.5pt}] (162.94,123.85) .. controls (162.94,95.21) and (185.77,72) .. (213.93,72) .. controls (242.09,72) and (264.92,95.21) .. (264.92,123.85) .. controls (264.92,152.49) and (242.09,175.7) .. (213.93,175.7) .. controls (185.77,175.7) and (162.94,152.49) .. (162.94,123.85) -- cycle ;
\draw  [line width=1.5]  (170.43,101.54) -- (257.43,101.54) -- (257.43,146.16) -- (170.43,146.16) -- cycle ;
\draw  [color={rgb, 255:red, 155; green, 155; blue, 155 }  ,draw opacity=1 ][dash pattern={on 4.5pt off 4.5pt}] (47.94,122.85) .. controls (47.94,94.21) and (70.77,71) .. (98.93,71) .. controls (127.09,71) and (149.92,94.21) .. (149.92,122.85) .. controls (149.92,151.49) and (127.09,174.7) .. (98.93,174.7) .. controls (70.77,174.7) and (47.94,151.49) .. (47.94,122.85) -- cycle ;
\draw  [line width=1.5]  (51.68,105.02) -- (146.18,105.02) -- (146.18,140.67) -- (51.68,140.67) -- cycle ;
\draw  [color={rgb, 255:red, 155; green, 155; blue, 155 }  ,draw opacity=1 ][dash pattern={on 4.5pt off 4.5pt}] (276.94,124.85) .. controls (276.94,96.21) and (299.77,73) .. (327.93,73) .. controls (356.09,73) and (378.92,96.21) .. (378.92,124.85) .. controls (378.92,153.49) and (356.09,176.7) .. (327.93,176.7) .. controls (299.77,176.7) and (276.94,153.49) .. (276.94,124.85) -- cycle ;
\draw  [color={rgb, 255:red, 0; green, 0; blue, 0 }  ,draw opacity=1 ][line width=1.5]  (293.17,89.5) -- (362.69,89.5) -- (362.69,160.2) -- (293.17,160.2) -- cycle ;
\draw  [fill={rgb, 255:red, 255; green, 255; blue, 255 }  ,fill opacity=1 ] (296.34,89.5) .. controls (296.34,87.79) and (294.92,86.4) .. (293.17,86.4) .. controls (291.42,86.4) and (290,87.79) .. (290,89.5) .. controls (290,91.22) and (291.42,92.6) .. (293.17,92.6) .. controls (294.92,92.6) and (296.34,91.22) .. (296.34,89.5) -- cycle ;
\draw  [fill={rgb, 255:red, 255; green, 255; blue, 255 }  ,fill opacity=1 ] (365.86,160.2) .. controls (365.86,158.48) and (364.44,157.1) .. (362.69,157.1) .. controls (360.94,157.1) and (359.52,158.48) .. (359.52,160.2) .. controls (359.52,161.91) and (360.94,163.3) .. (362.69,163.3) .. controls (364.44,163.3) and (365.86,161.91) .. (365.86,160.2) -- cycle ;
\draw  [fill={rgb, 255:red, 0; green, 0; blue, 0 }  ,fill opacity=1 ] (296.34,160.2) .. controls (296.34,158.48) and (294.92,157.1) .. (293.17,157.1) .. controls (291.42,157.1) and (290,158.48) .. (290,160.2) .. controls (290,161.91) and (291.42,163.3) .. (293.17,163.3) .. controls (294.92,163.3) and (296.34,161.91) .. (296.34,160.2) -- cycle ;
\draw  [fill={rgb, 255:red, 0; green, 0; blue, 0 }  ,fill opacity=1 ] (365.86,89.5) .. controls (365.86,87.79) and (364.44,86.4) .. (362.69,86.4) .. controls (360.94,86.4) and (359.52,87.79) .. (359.52,89.5) .. controls (359.52,91.22) and (360.94,92.6) .. (362.69,92.6) .. controls (364.44,92.6) and (365.86,91.22) .. (365.86,89.5) -- cycle ;
\draw  [fill={rgb, 255:red, 255; green, 255; blue, 255 }  ,fill opacity=1 ] (173.6,101.54) .. controls (173.6,99.83) and (172.18,98.44) .. (170.43,98.44) .. controls (168.68,98.44) and (167.26,99.83) .. (167.26,101.54) .. controls (167.26,103.25) and (168.68,104.64) .. (170.43,104.64) .. controls (172.18,104.64) and (173.6,103.25) .. (173.6,101.54) -- cycle ;
\draw  [fill={rgb, 255:red, 255; green, 255; blue, 255 }  ,fill opacity=1 ] (260.6,146.16) .. controls (260.6,144.45) and (259.18,143.06) .. (257.43,143.06) .. controls (255.68,143.06) and (254.26,144.45) .. (254.26,146.16) .. controls (254.26,147.87) and (255.68,149.26) .. (257.43,149.26) .. controls (259.18,149.26) and (260.6,147.87) .. (260.6,146.16) -- cycle ;
\draw  [fill={rgb, 255:red, 0; green, 0; blue, 0 }  ,fill opacity=1 ] (173.6,146.16) .. controls (173.6,144.45) and (172.18,143.06) .. (170.43,143.06) .. controls (168.68,143.06) and (167.26,144.45) .. (167.26,146.16) .. controls (167.26,147.87) and (168.68,149.26) .. (170.43,149.26) .. controls (172.18,149.26) and (173.6,147.87) .. (173.6,146.16) -- cycle ;
\draw  [fill={rgb, 255:red, 0; green, 0; blue, 0 }  ,fill opacity=1 ] (260.6,101.54) .. controls (260.6,99.83) and (259.18,98.44) .. (257.43,98.44) .. controls (255.68,98.44) and (254.26,99.83) .. (254.26,101.54) .. controls (254.26,103.25) and (255.68,104.64) .. (257.43,104.64) .. controls (259.18,104.64) and (260.6,103.25) .. (260.6,101.54) -- cycle ;
\draw  [fill={rgb, 255:red, 255; green, 255; blue, 255 }  ,fill opacity=1 ] (54.85,105.02) .. controls (54.85,103.31) and (53.43,101.92) .. (51.68,101.92) .. controls (49.93,101.92) and (48.51,103.31) .. (48.51,105.02) .. controls (48.51,106.74) and (49.93,108.13) .. (51.68,108.13) .. controls (53.43,108.13) and (54.85,106.74) .. (54.85,105.02) -- cycle ;
\draw  [fill={rgb, 255:red, 255; green, 255; blue, 255 }  ,fill opacity=1 ] (149.35,140.67) .. controls (149.35,138.96) and (147.93,137.57) .. (146.18,137.57) .. controls (144.43,137.57) and (143.01,138.96) .. (143.01,140.67) .. controls (143.01,142.39) and (144.43,143.78) .. (146.18,143.78) .. controls (147.93,143.78) and (149.35,142.39) .. (149.35,140.67) -- cycle ;
\draw  [fill={rgb, 255:red, 0; green, 0; blue, 0 }  ,fill opacity=1 ] (54.85,140.67) .. controls (54.85,138.96) and (53.43,137.57) .. (51.68,137.57) .. controls (49.93,137.57) and (48.51,138.96) .. (48.51,140.67) .. controls (48.51,142.39) and (49.93,143.78) .. (51.68,143.78) .. controls (53.43,143.78) and (54.85,142.39) .. (54.85,140.67) -- cycle ;
\draw  [fill={rgb, 255:red, 0; green, 0; blue, 0 }  ,fill opacity=1 ] (149.35,105.02) .. controls (149.35,103.31) and (147.93,101.92) .. (146.18,101.92) .. controls (144.43,101.92) and (143.01,103.31) .. (143.01,105.02) .. controls (143.01,106.74) and (144.43,108.13) .. (146.18,108.13) .. controls (147.93,108.13) and (149.35,106.74) .. (149.35,105.02) -- cycle ;
\draw  [fill={rgb, 255:red, 255; green, 255; blue, 255 }  ,fill opacity=1 ] (423.78,82.35) .. controls (423.78,80.64) and (422.37,79.25) .. (420.62,79.25) .. controls (418.87,79.25) and (417.45,80.64) .. (417.45,82.35) .. controls (417.45,84.06) and (418.87,85.45) .. (420.62,85.45) .. controls (422.37,85.45) and (423.78,84.06) .. (423.78,82.35) -- cycle ;
\draw  [fill={rgb, 255:red, 255; green, 255; blue, 255 }  ,fill opacity=1 ] (468.41,169.35) .. controls (468.41,167.64) and (466.99,166.25) .. (465.24,166.25) .. controls (463.49,166.25) and (462.07,167.64) .. (462.07,169.35) .. controls (462.07,171.06) and (463.49,172.45) .. (465.24,172.45) .. controls (466.99,172.45) and (468.41,171.06) .. (468.41,169.35) -- cycle ;
\draw  [fill={rgb, 255:red, 0; green, 0; blue, 0 }  ,fill opacity=1 ] (423.78,169.35) .. controls (423.78,167.64) and (422.37,166.25) .. (420.62,166.25) .. controls (418.87,166.25) and (417.45,167.64) .. (417.45,169.35) .. controls (417.45,171.06) and (418.87,172.45) .. (420.62,172.45) .. controls (422.37,172.45) and (423.78,171.06) .. (423.78,169.35) -- cycle ;
\draw  [fill={rgb, 255:red, 0; green, 0; blue, 0 }  ,fill opacity=1 ] (468.41,82.35) .. controls (468.41,80.64) and (466.99,79.25) .. (465.24,79.25) .. controls (463.49,79.25) and (462.07,80.64) .. (462.07,82.35) .. controls (462.07,84.06) and (463.49,85.45) .. (465.24,85.45) .. controls (466.99,85.45) and (468.41,84.06) .. (468.41,82.35) -- cycle ;
\draw  [fill={rgb, 255:red, 255; green, 255; blue, 255 }  ,fill opacity=1 ] (543.27,79.6) .. controls (543.27,77.89) and (541.85,76.5) .. (540.1,76.5) .. controls (538.35,76.5) and (536.94,77.89) .. (536.94,79.6) .. controls (536.94,81.31) and (538.35,82.7) .. (540.1,82.7) .. controls (541.85,82.7) and (543.27,81.31) .. (543.27,79.6) -- cycle ;
\draw  [fill={rgb, 255:red, 255; green, 255; blue, 255 }  ,fill opacity=1 ] (578.92,174.1) .. controls (578.92,172.39) and (577.5,171) .. (575.75,171) .. controls (574,171) and (572.59,172.39) .. (572.59,174.1) .. controls (572.59,175.81) and (574,177.2) .. (575.75,177.2) .. controls (577.5,177.2) and (578.92,175.81) .. (578.92,174.1) -- cycle ;
\draw  [fill={rgb, 255:red, 0; green, 0; blue, 0 }  ,fill opacity=1 ] (543.27,174.1) .. controls (543.27,172.39) and (541.85,171) .. (540.1,171) .. controls (538.35,171) and (536.94,172.39) .. (536.94,174.1) .. controls (536.94,175.81) and (538.35,177.2) .. (540.1,177.2) .. controls (541.85,177.2) and (543.27,175.81) .. (543.27,174.1) -- cycle ;
\draw  [fill={rgb, 255:red, 0; green, 0; blue, 0 }  ,fill opacity=1 ] (578.92,79.6) .. controls (578.92,77.89) and (577.5,76.5) .. (575.75,76.5) .. controls (574,76.5) and (572.59,77.89) .. (572.59,79.6) .. controls (572.59,81.31) and (574,82.7) .. (575.75,82.7) .. controls (577.5,82.7) and (578.92,81.31) .. (578.92,79.6) -- cycle ;
\end{tikzpicture},
\end{equation}
where we are restricting to rectangles as an illustration. Suppose the two distinct edges have weights $2\nu_1^{1/2}$ and $2\nu_2^{1/2}$, then we have $\nu_I=\sin^2(\theta_I)$ and $\nu_1+\nu_2=1$. The Mahler measure reads
\begin{equation}
    m_\textup{iso}=\int_0^{2\pi}\int_0^{2\pi}\log|2\nu_1+2\nu_2-\nu_1\textup{e}^{is}-\nu_1\textup{e}^{-is}-\nu_2\textup{e}^{it}-\nu_2\textup{e}^{-it}|\textup{d}s\textup{d}t.
\end{equation}
The isoradial Mahler measure as a function of $\nu_1$ and its derivative are plotted in Figure \ref{isomaxex}.
\begin{figure}[h]
	\centering
	\begin{subfigure}{5cm}
	    \centering
	    \includegraphics[width=5cm]{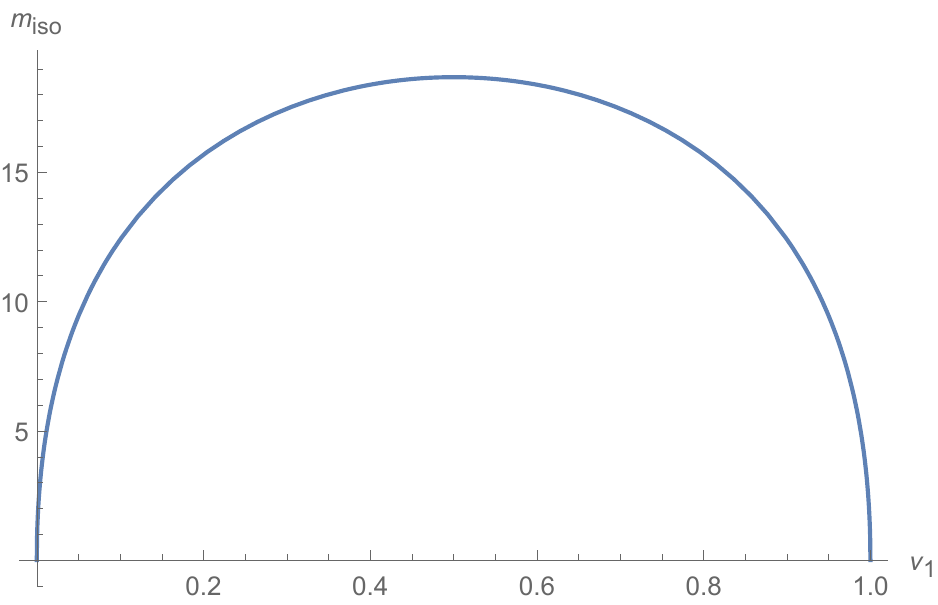}
	    \caption{}
	\end{subfigure}
	\begin{subfigure}{5cm}
	    \centering
	    \includegraphics[width=5cm]{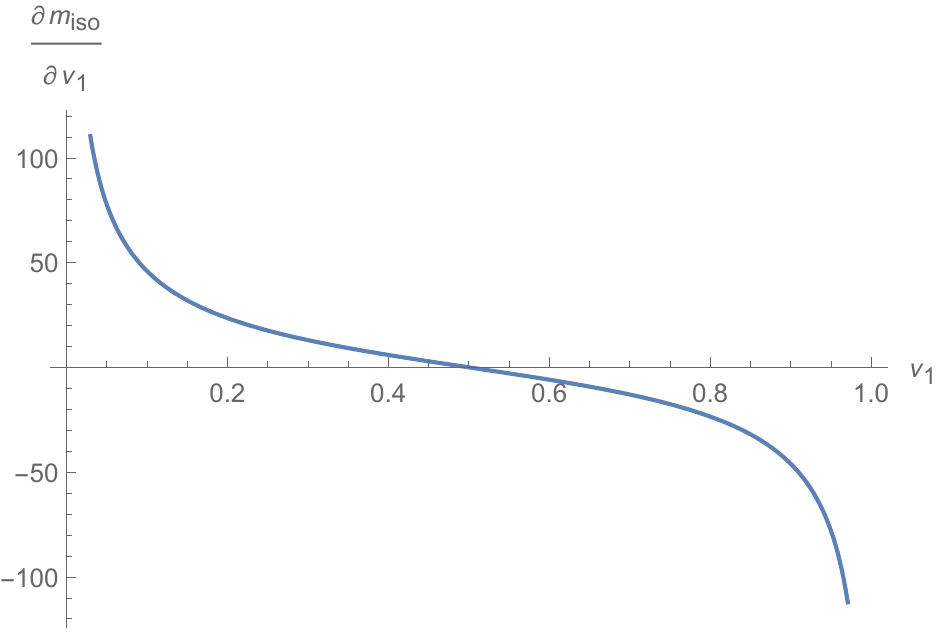}
	    \caption{}
	\end{subfigure}
	\caption{The plot of (a) $m_\text{iso}$ and (b) $\partial m_\text{iso}/\partial\nu_1$. Here $\nu_2=1-\nu_1$.}\label{isomaxex}
\end{figure}
We can see that the maximum is reached at $\nu_1=1/2$. Indeed, \eqref{isomahler} yields
\begin{equation}
    m_\textup{iso}=4\left(\frac{\theta_1}{\pi}\log(2\sin(\theta_1))+\frac{1}{\pi}\Lambda(\theta_1)\right)+4\left(\frac{\theta_2}{\pi}\log(2\sin(\theta_2))+\frac{1}{\pi}\Lambda(\theta_2)\right)
\end{equation}
with $\theta_2=\pi/2-\theta_1$. Its derivatives read
\begin{equation}
    \begin{split}
        &\frac{\partial m_\textup{iso}}{\partial\theta_1}=\frac{4\theta_1}{\pi}\cot(\theta_1)+\frac{4\theta_1-2\pi}{\pi}\tan(\theta_1),\\
        &\frac{\partial^2m_\textup{iso}}{\partial\theta_1^2}=\frac{4\cot(\theta_1)-4\theta_1\csc^2(\theta_1)+2\sec^2(\theta)(2\theta_1+2\sin(\theta_1)-\pi)}{\pi}.
    \end{split}
\end{equation}
One can find that the first derivative vanishes when $\theta_1=\pi/4$ (and the second derivative equals $8/\pi-4<0$). Therefore, the Mahler measure is maximized at $\theta_1=\theta_2=\pi/4$ (or equivalently, $\nu_1=\nu_2=1/2$). This agrees with the result from $a$-maximization. In general, if the faces in the dimer are not rectangles but general isoradial quadrilaterals, the Mahler measure should still be maximized at $\theta_{1,\dots,8}=\pi/4$, that is, $R_{1,\dots,8}=1/2$.
\end{example}

\subsubsection{Seiberg Duality}\label{Seiberg}
Theories engineered by D-branes on toric Gorenstein singularities enjoy certain dualities including Seiberg duality and specular duality. In this subsection, we show that the Mahler measure also exhibits certain properties under these dualities. 
Let us start with Seiberg duality \cite{Seiberg:1994pq,Feng:2001bn,Beasley:2001zp}. In terms of quivers, this is essentially mutations (plus non-trivial superpotentials). For brane tilings, this gives rise to so-called urban renewal \cite{propp2003generalized,Franco:2005rj}. Seiberg duals correspond to the same toric diagram as moduli space. However, there is no reason why the Newton polynomials for the duals would remain the same since the coefficients are obtained from different brane tilings whose perfect matchings can change. Nevertheless, we will see that at the isoradial point, the polynomials are the same.

Assume that the dimer is embedded isoradially on the torus, then we can use our canonical choice of edge weights to write down the Newton polynomial. It turns out that for all the duals we checked, although the total numbers of multiplets/rhombus angles  change, the coefficient for every single monomial, viz, $\sum\text{e}^{-\mathcal{E}(M)}=\sum\left(\prod\limits_I2\log(\sin(2\theta_I))\right)$, is invariant up to a factor of $2^{n_2-n_1}$ for some integers $n_i$. In fact, $n_1$ and $n_2$ are nothing but the numbers of edges/multiplets in each perfect matching/GLSM field for the two tilings respectively. Therefore, we can simply cancel the factor $2^{n_1}$ or $2^{n_2}$ from every monomial in the two Newton polynomials.

Note that there is a subtlety in the above discussion. In general, a dimer (which is a dual of an isoradial one) does not admit an isoradial embedding. Nevertheless, let us still use $2\sin(\pi R_I/2)$ as the edge weights\footnote{We can only write the weights in terms of the R-charges as the concept of rhombus angle is not really well-defined for non-isoradial embeddings.}. It turns out that the resulting Newton polynomial is again the same as its Seiberg dual(s). To ensure that this is the desired $P(z,w)$ even though the embedding is non-isoradial, one may check that in such case the Mahler measure is equal to the value computed from \eqref{isomahler} using $\theta_I=\pi R_I/2$ for the physical R-charges $R_I$. In fact, this is because the non-isoradial dimer has been continuously deformed so that some of the edges shrink to zero (or even negative lengths). After such deformation, the non-isoradial dimer degenerates to an isoradial dimer with some of the rhombus angles being $\pi/2$ (or even obtuse). Therefore, the formula \eqref{isomahler} for isoradial Mahler measures would still work in this situation.

\begin{conjecture}
Seiberg duals have exactly the same Newton polynomials. Hence, Mahler measure is trivially invariant under Seiberg duality.

Equivalently, we may keep the two factors $2^{n_1}$ and $2^{n_2}$ such that the Newton polynomials are related by $P_2(z,w)=2^{n_2-n_1}P_1(z,w)$, where $n_i$ is the number of edges/multiplets in each perfect matching/GLSM field. Then we can say that the Mahler measure is invariant under Seiberg duality up to $\log(2^{n_2-n_1})$, that is, $m_2=m_1+\log(2^{n_2-n_1})$.\label{Seibergconj}
\end{conjecture}

\begin{remark}
Since the Newton polynomial $P(z,w)=k_\textup{iso}-p(z,w)$ is invariant, if we leave the isoradial point and increase $k$, the resulting $P=k-p$ and its Mahler measure would also be the same for Seiberg duals.
\end{remark}

\begin{example}
Let us consider $L^{1,3,1}/\mathbb{Z}_2$ $(0,1,1,1)$ as an example. It has two toric phases whose brane tilings are
\begin{equation}
    \includegraphics[width=4cm]{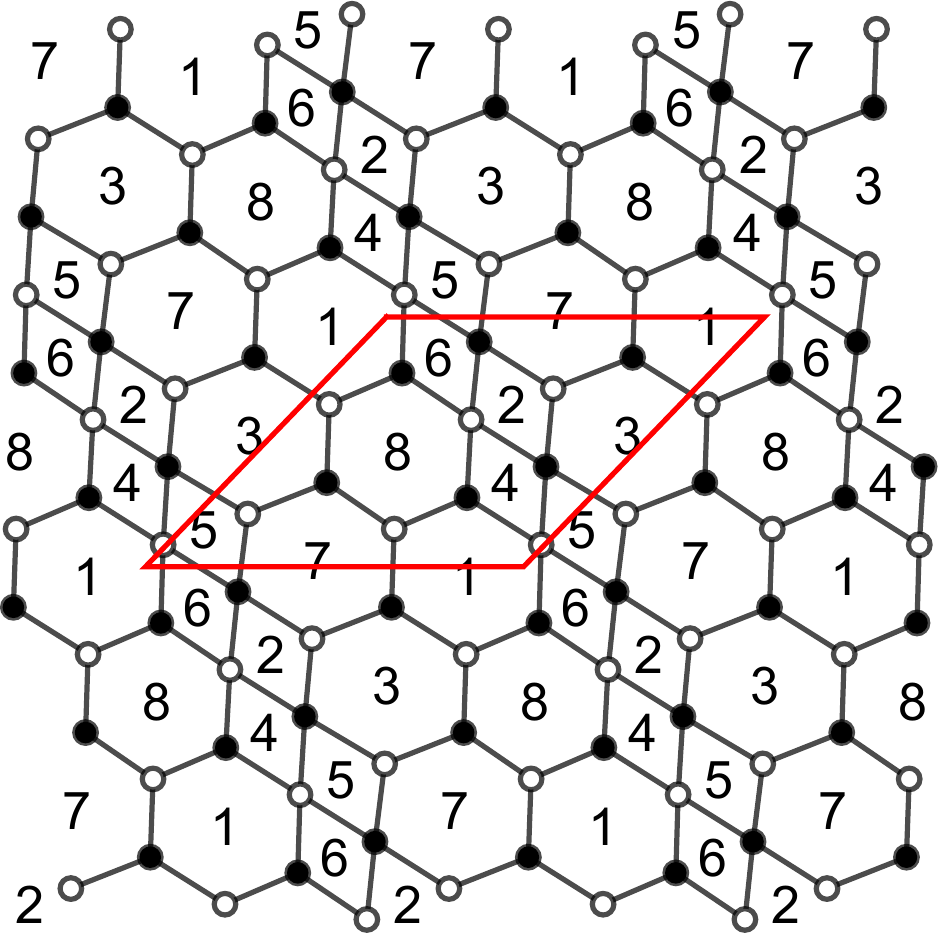},~~~~~~
    \includegraphics[width=4cm]{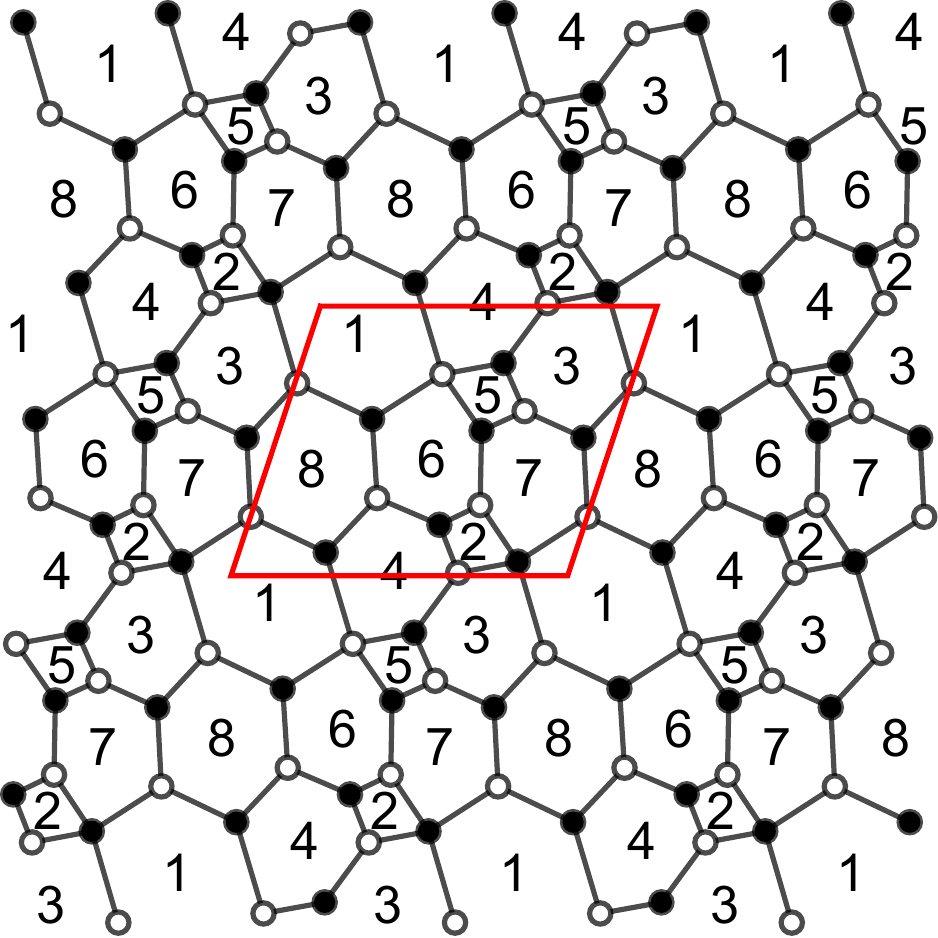}.
\end{equation}
Using the data including perfect matrices in \cite{Hanany:2012hi}, we find that both of them yield (after cancelling a common factor of $2^n$ for each case respectively)
\begin{equation}
    P=-\frac{AB}{C}(zw+w)-\frac{C}{AB}(zw^{-1}+z^{-2}w^{-1}+3w^{-1}+3z^{-1}w^{-1})-(2z+2z^{-1})+12,
\end{equation}
where $A,B,C$ are given in \eqref{L131Z2}. The dual theories have $20$ and $22$ bifundamentals respectively, but the corresponding edges weighted by their R-charges lead to the same Newton polynomial. In particular, the second dimer does not admit an isoradial embedding. However, we find that two of the R-charges in this case become $1$. This is equivalent to two of the edges shrinking to zero in the dimer. One edge $e_{34}$ is between face 3 and 4 while the other $e_{67}$ is between face 6 and 7. After such deformation, the dimer degenerates to an isoradial embedding with two rhombus angles being $\pi/2$. One may also check that \eqref{isomahler} in this case agrees with the result from the first dimer and gives the correct Mahler measure.

It is worth noting that when the two edges shrink in the second brane tiling, it does not degenerate to the first tiling but instead some dimer that is ``partially urban renewed''. More precisely, from the second tiling to the first one under Seiberg duality, $e_{34,67,64}$ would vanish while a new edge $e_{73}$ would be created. In the ``partially renewed'' dimer, only the removal of $e_{34}$ and $e_{67}$ has been done. In terms of the quivers, Seiberg duality flips node 2 in the quiver. The ``partially renewed'' dimer has bifundamentals $X_{34}$ and $X_{67}$ removed in the quiver while $X_{64}$ has not yet been removed and $X_{73}$ has not not been added. Therefore, this partially mutated quiver theory is anomalous\footnote{We should emphasize that when we say the tiling \emph{degenerates}, it only ``looks like'' the ``partially renewed'' dimer but does \emph{not} ``become'' that dimer. The tiling still gives an anomaly-free physical quiver theory. Only the two edges $e_{34,67}$ have length 0 (viz, weight 2) due to the R-charges for $X_{34,67}$ being 1.}.
\end{example}

In fact, one may also check that some of the matter fields have $R_I>1$ in some toric phases. In such cases, we find that the above discussion still holds. We may therefore say that the weights are assigned to be $2\sin(\theta_I)$ with $\theta=\pi R_I/2>\pi/2$. Eqvuialently, we can also regard the edges as of ``negative'' lengths, that is, $2\cos(\pi R_I/2)<0$.

\begin{example}
Let us consider $\textup{PdP}_{4a}$ as an example. It has three toric phases whose brane tilings are
\begin{equation}
    \includegraphics[width=4cm]{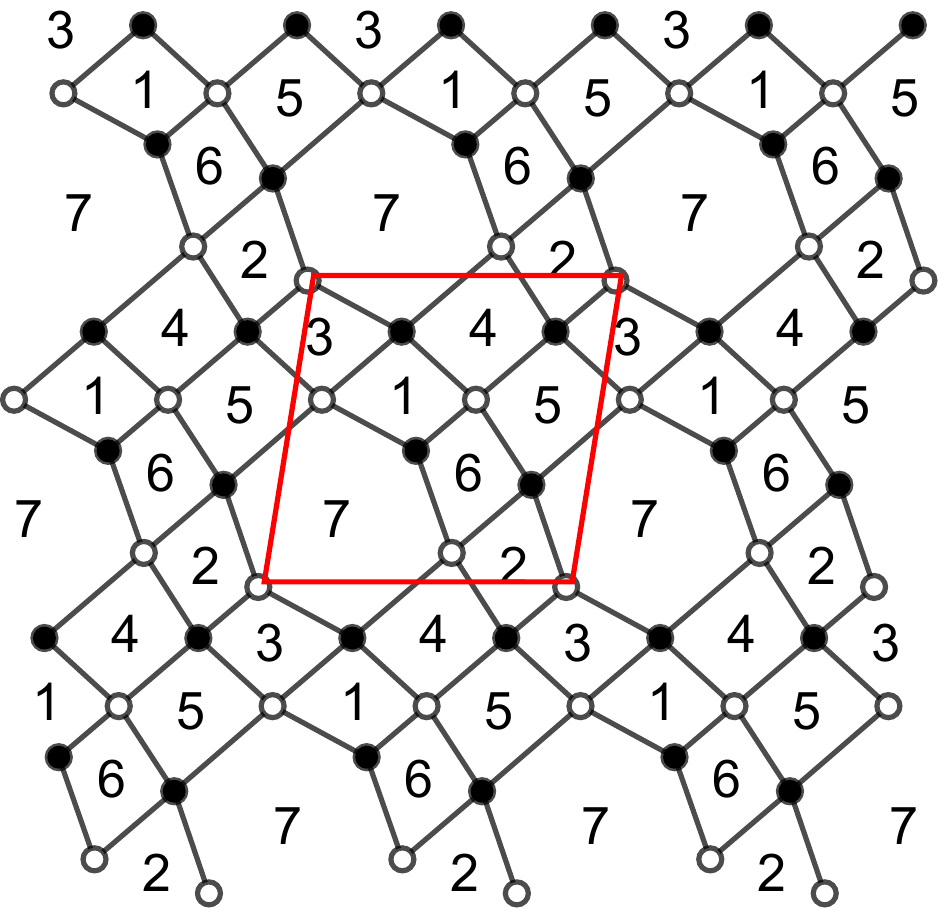},~~~
    \includegraphics[width=4cm]{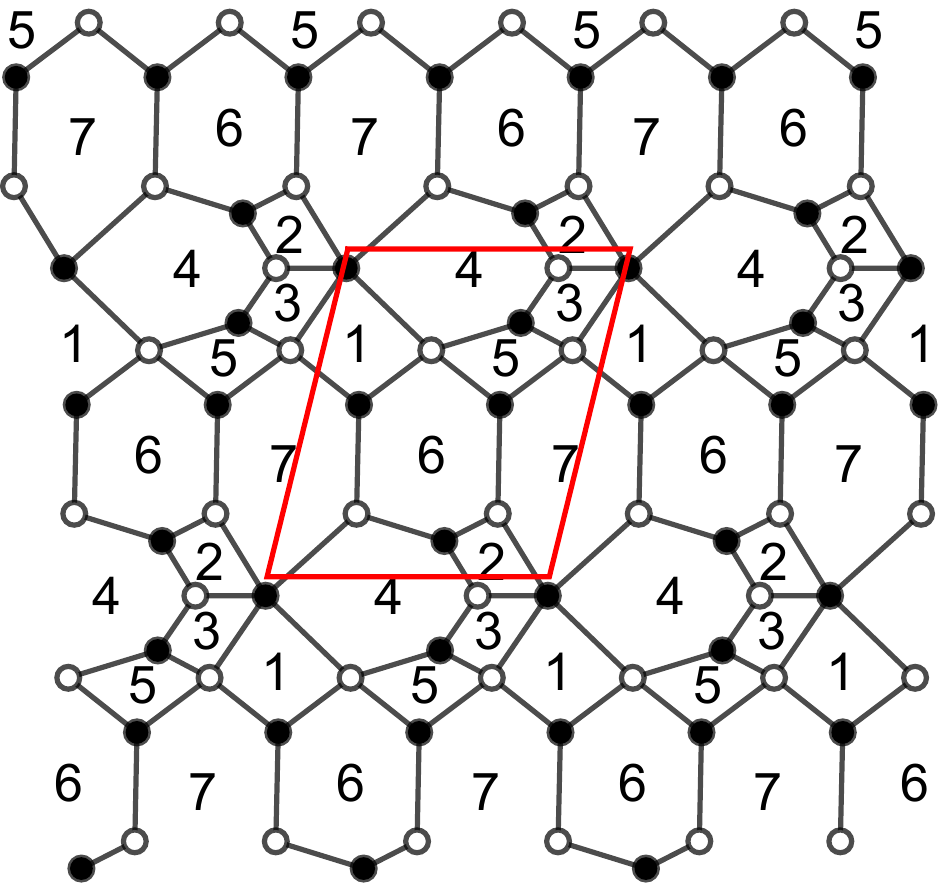},~~~
    \includegraphics[width=4cm]{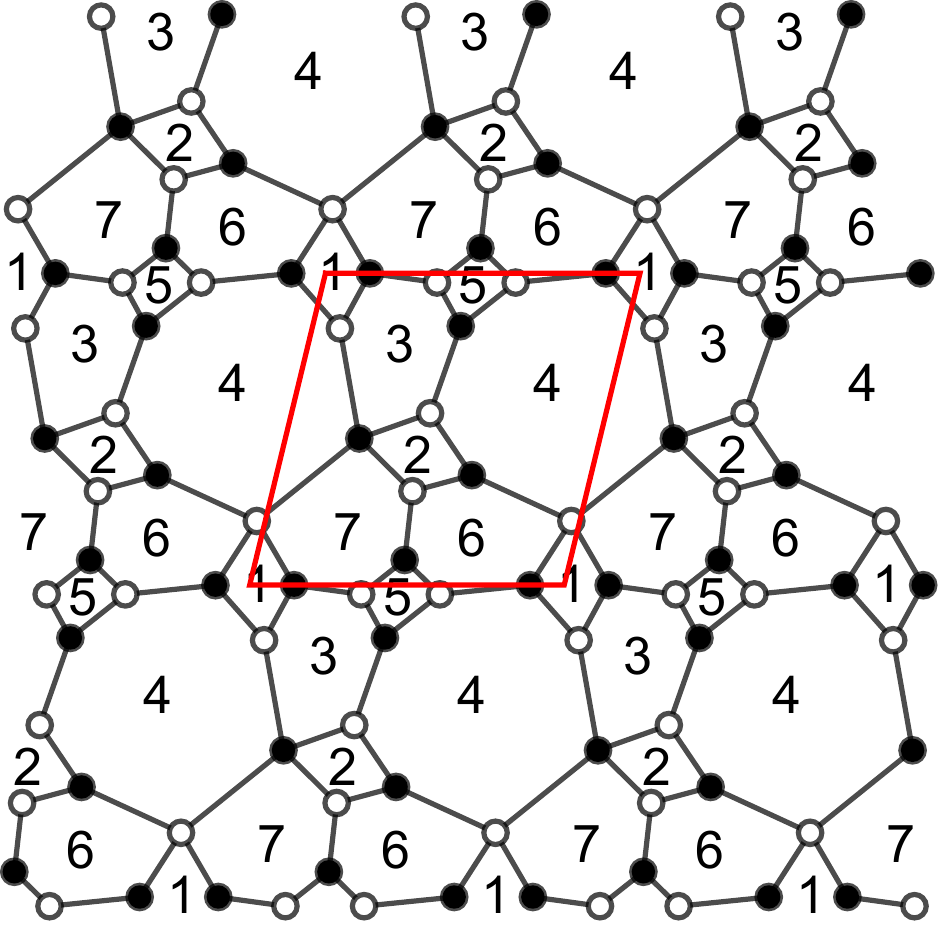}.\label{PdP4atilings}
\end{equation}
Using the data including perfect matrices in \cite{Hanany:2012hi}, we find that all of them yield (after cancelling a common factor of $2^n$ for each case respectively)
\begin{equation}
    \begin{split}
        P=&-B_1^2B_2^2(z+z^{-1})-B_2^2B_3B_4(zw^{-1}+w^{-1})-B_5^4z^{-1}w^2-2B_1B_2B_5^2(z^{-1}w+w)\\
        &+2B_1^2B_2^2+4B_1B_2B_3B_5+2B_2^2B_4B_5+B_3^2B_5^2,\label{PdP4a}
    \end{split}
\end{equation}
where $B_i=\sin(\pi b_i/2)$ are given in \eqref{PdP4a}. The dual theories have $15$, $17$ and $19$ bifundamentals respectively, but the corresponding edges weighted by their R-charges lead to the same Newton polynomial.

In particular, for the second tiling, $X_{67}$ has R-charge $R_{67}=\mathfrak{R}$, where
\begin{equation}
    \mathfrak{R}\approx1.023\text{ is a root of }x^3+24x^2-96x+72=0.
\end{equation}
Therefore, the edge $e_{67}$ (not to be confused with $e_{76}$) has ``negative'' length in the dimer. When performing Seiberg duality between the second tiling to the first one. Node 2 in the quiver is flipped. Therefore, the white node where face 2, 6, 7 meet in the second dimer would become black. Then the edge $e_{67}$ would have its white and black nodes reversed, and hence of negative length. Notice that $e_{67}$ vanishes in the first dimer. This is because the ``isoradial'' embedding for the second dimer with negative $e_{67}$ degenerates to one being ``partially urban renewed''.

Likewise, in the third tiling, the fields $X_{37}$ and $X_{67}$ would both have R-charge $\mathfrak{R}>1$. Again, compared to the first dimer, the two edges $e_{37}$ and $e_{67}$ would have their white and black nodes reversed, and hence of negative lengths.
\end{example}

From the above discussion, we have actually found a canonical choice for edge weights regardless of isoradiality.
\begin{remark}
For any brane tiling, the canonical choice (in the sense of $m_\textup{iso}$-/$a$-maximization and Seiberg duality) for the weight of edge $e_I$ is $2\sin(\pi R_I/2)$. The Mahler measure for the corresponding Newton polynomial can be computed using \eqref{isomahler}.

For isoradial embedding, this edge weight is just the critical choice (in the sense of \cite{Kenyon_2002}) $2\sin(\theta_I)$. For non-isoradial embedding, we have two equivalent viewpoints:
\begin{itemize}
    \item The dimer still has the original shape with positive edge lengths, and hence non-isoradial. We are just assigning the so-called canonical weights to the edges.
    \item As before, we may also say that the dimer degenerates and becomes isoradial. However, as we have emphasized, it is \emph{different} from the dimer that leads to anomalous theory. Some of the edges are \emph{not} removed, and they just have zero or negative lengths.
\end{itemize}
\end{remark}
The first point of view emphasizes the {\bf universal} weight choice for all (both isoradial and non-isoradial) embeddings. On the other hand, the second one explains why we can always apply \eqref{isomahler} in the calculations of Mahler measures for any toric quiver theories.

In fact, when computing Mahler measures using \eqref{isomahler} for Seiberg duals, they would imply some non-trivial mathematical identities. For instance,
\begin{remark}
In the example for $\textup{PdP}_{4a}$ above, we have
\begin{equation}
    \alpha\log(2\sin(\alpha))+\Lambda(\alpha)+\beta\log(2\sin(\beta))+\Lambda(\beta)=\pi\log(2),
\end{equation}
where
\begin{equation}
    \begin{split}
        \frac{2\alpha}{\pi}&=10-\frac{24(1-i\sqrt{3})}{\left(\frac{1}{2}(233+i\sqrt{1007})\right)^{1/3}}-(1+i\sqrt{3})\left(\frac{1}{2}(233+i\sqrt{1007})\right)^{1/3}\approx0.977473,\\
        \frac{2\beta}{\pi}&=-8-\frac{24(1-i\sqrt{3})}{\left(\frac{1}{2}(-233+i\sqrt{1007})\right)^{1/3}}-(1+i\sqrt{3})\left(\frac{1}{2}(-233+i\sqrt{1007})\right)^{1/3}\approx1.02253.
    \end{split}
\end{equation}
When there are some R-charges equal to $1$ (such as the above example for $L^{1,3,1}/\mathbb{Z}_2$), we always have the same identity
\begin{equation}
    \Lambda(\pi/2)\equiv-\int_0^{\pi/2}\log(2\sin(t))\textup{d}t=0.
\end{equation}
This is expected due to the periodicity of $\sin(t)$.
\end{remark}

To end this subsection, let us make a comment on the Ronkin functions. Although there is no simple formula like \eqref{isomahler} for Mahler measures,
\begin{conjecture}
The Ronkin function is invariant under Seiberg duality since the Newton polynomial (with the canonical weight choice) does not change (regardless of the isoradiality of the dimer).
\end{conjecture}

\begin{remark}
We mainly focused on $k=k_\text{iso}$ in this subsection. However, we may also leave the isoradial point and consider general $P(z,w)=k-p(z,w)$. Since Seiberg duals have the same starting isoradial point, the Newton polynomials, Mahler measures and Ronkin functions would also be the same for these dual theories along the Mahler flow.
\end{remark}

\subsection{The Master Space}\label{master}
So far, we have engaged in many discussions on Mahler measures at $k=k_\text{iso}$. Let us now leave the isoradial point and treat $k$ as a general variable. As we are now going to see, this would encode certain information of the master space.

Recall that for $k>\max\limits_{|z|=|w|=1}(|p(z,w)|)=p(1,1)$, we have the expansions
\begin{equation}
     m(P)=\log k-\sum_{n=2}^\infty\frac{f_n}{nk^n};\quad u_0(k)=1+\sum_{n=2}^\infty\frac{f_n}{k^n} ,
\end{equation}
for some expansion coefficients $f_n$.
Since the coefficients of $p(z,w)$ come from the perfect matchings/GLSM fields, we shall write it as
\begin{equation}
    p(z,w)=\sum_{i}r_iz^{m_i}w^{n_i},
\end{equation}
where each $r_i$ denotes a perfect matching that does not correspond to a constant term in $P(z,w)$. Likewise, we shall denote the perfect matchings giving a constant term in $P$ as $s_i$. Henceforth, we will refer to them as $r$-matchings and $s$-matchings respectively. Note that here we are keeping the coefficients general (rather than just the canonical choices). Following the multinomial theorem, we have \cite{villegas1999modular}
\begin{equation}
    f_n=\sum_{\substack{\bm{l}=(l_1,\dots,l_N)\in\mathbb{Z}_{\geq0}^N\\l_1+\dots+l_N=n\\M\bm{l}=0}}\frac{n!}{l_1!\dots l_N!}r_1^{l_1}\dots r_N^{l_N},
\end{equation}
where $N$ is the number of $r$-matchings and
\begin{equation}
    M=\begin{pmatrix}
m_1 & m_2 & \dots & m_N\\
n_1 & n_2 & \dots & n_N
\end{pmatrix}
\end{equation}
is the $2\times N$ matrix of the corresponding lattice points. Notice that there could be duplicated columns in $M$ since some lattice points can correspond to multiple ($r$-)matchings.
We have that
\begin{proposition}\label{masteru0}
The period $u_0(k)$ is a generating function of the master space in terms of $F$-term charge matrix.
\end{proposition}
\begin{proof}
The {\bf master space} $\mathcal{F}^{\flat}$ \cite{Forcella:2008eh,Forcella:2008bb} is the spectrum of the coordinate ring quotiented by F-term relations, $\mathbb{C}[X_I]/\langle\partial W=0\rangle$. Its largest irreducible component is known as the {\bf coherent component} ${}^\text{Irr}\mathcal{F}^{\flat}$, which can be written as a symplectic quotient in terms of GLSM fields $q_i$:
\begin{equation}
    {}^\text{Irr}\mathcal{F}^{\flat}=\mathbb{C}[q_i]//Q_F,
\end{equation}
where $Q_F$ is the F-term charge matrix which is equal to the kernel of perfect matching matrix \cite{Feng:2000mi}.

Consider the expansion of $u_0(k)$ as
\begin{equation}
    u_0(k)=1+\sum_{n=2}^\infty\frac{[p^n(z,w)]_0}{k^n}.
\end{equation}
Recall that on the fundamental domain $\mathcal{G}_1$ of the dimer, $p(z,w)$ has coefficients $r_i$. (If we have a monomial with coefficient say $r_i+r_j$, we then treat it as two terms.) Moreover, each $r$-matching is composed of edges $e_I$ on $\mathcal{G}_1$.

At order $n$, the coefficients of $p^n(z,w)$ (up to multinomial coefficients) are
\begin{equation}
    \prod_{l_1+\dots+l_N=n}r_i^{l_i},
\end{equation}
where $l_i$ denotes the multiplicity of the matching $r_i$. They correspond to ``perfect matchings'' on the $n$ copies of the fundamental domain $n\mathcal{G}_1$ (not to be confused with $\mathcal{G}_n$). Therefore, let us write such perfect matching on $n\mathcal{G}_1$ as $l_1r_1+l_2r_2+\dots+l_Nr_N$. Since we are considering the constant terms $[p^n(z,w)]_0$, they correspond to ``$s$-matchings'' on $n\mathcal{G}_1$. By regrouping the edges in this internal perfect matching, we can write
\begin{equation}
    l_1r_1+l_2r_2+\dots+l_Nr_N=\sum_Ie_I=l'_1s_1+l'_2s_2+\dots+l'_Ns_N,\label{rsreln}
\end{equation}
where $s_i$ are $s$-matchings on $\mathcal{G}_1$. This is always plausible because the number of edges $e_I$ in each perfect matching is the same.

In the F-term charge matrix, there are two types of relations.
\begin{itemize}
    \item If there is a row giving the relation $\sum r_i=\sum s_j$, then this is simply from $l_i=l'_j=1$.
    \item If there is a relation involving only $r$-matchings (without $s$-matchings), say $\sum r_i=\sum r'_j$, then we can replace $r'_j$'s using the $s$-matchings following the relations in the first type. This would reduce the relation of second type to \eqref{rsreln} with possibly non-trivial multiplicities.
\end{itemize}
Hence, all the relations in $Q_F$ are revealed by coefficients in the expansion of $u_0(k)$. The multinomial coefficients $\frac{n!}{l_1!\dots l_N!}$ are just the numbers of equivalent ways of writing/ordering the $r$-matchings. This follows directly from the combinatorial interpretation of multinomial coefficients.
\end{proof}

\begin{remark}
At the lowest non-zero order $n=2$, the relations are of first type. The relations of second type would appear in higher orders. All the other relations are just (redundant) relations of those two-type relations.
\end{remark}

\begin{remark}
Since we are working with general coefficients for $P(z,w)$, the canonical choice would certainly satisfy this proposition. We can simply replace $r_i$ with $\prod\limits_{I}2\sin(\pi R_I/2)$ for every perfect matching.
\end{remark}

\begin{remark}
Since we always have D$\,4$-branes as flavour branes in the system when the Taylor expansion is valid, the superpotential would change from $W_0$ to $(W_0+W_\textup{flav})$. Nevertheless, the F-term relations would still be $\partial W_0/\partial X_I=0$ as shown in \cite{Nishinaka:2013mba}. Alternatively, $Q_F$ is only determined by perfect matchings on the dimer which remain unchanged regardless the existence of D$\,4$-branes. Therefore, the above proposition should always hold\footnote{Incidentally, the moduli space of D4-D2-D0 states is just a subspace of the moduli space of D6-D2-D0 states \cite{Nishinaka:2013mba}.}.
\end{remark}

It is best to illustrate the foregoing discussions with an example.
\begin{example}
Consider $Y^{2,2}$ whose Newton polynomial is $P=k-(r_1w+r_2w^{-1}+r_3z^2+r_4z^{-1}+r_5z^{-1})$. Then
\begin{equation}
    u_0(k)=1+\frac{2r_3r_4+2r_3r_5}{k^2}+\frac{12r_1r_2r_3^2+12r_3^2r_4r_5+6r_3^2r_4^2+6r_3^2r_5^2}{k^4}+\dots
\end{equation}
The F-term charge matrix is
\begin{equation}
    Q_F=
\left(
\begin{array}{ccccc|cccc}
 r_{1} & r_{2} & r_{3} & r_{4} & r_{5} & s_{1} & s_{2} & s_{3} & s_{4} \\
   \hline
 0 & 0 & -1 & -1 & 0 & 1 & 1 & 0 & 0 \\
 0 & 0 & -1 & 0 & -1 & 0 & 0 & 1 & 1 \\
 1 & 1 & 0 & -1 & -1 & 0 & 0 & 0 & 0 
\end{array}
\right).
\end{equation}
At order two, the coefficients give the first two rows in $Q_F$: $r_3+r_4=s_1+s_2$ and $r_3+r_5=s_3+s_4$. The factors $2$ are just equivalent ways of writing the relations, e.g. $r_3+r_4$ and $r_4+r_3$. At order four, $r_1r_2r_3^2$ decodes the third row in $Q_F$ since
\begin{equation}
    r_1+r_2=r_4+r_5=s_1+s_2-r_3+s_3+s_4-r_3.
\end{equation}
Again, the factor $12$ is just the number of ways to arrange $r_1$, $r_2$ and two $r_3$'s. The second term $12r_3^2r_4r_5$ reveals the same relation but with $r_4+r_5+2r_3=s_1+s_2+s_3+s_4$ instead. The remaining two terms are (redundant) relations of lower-order relations: $2(r_3+r_4)=2(s_1+s_2)$ and $2(r_3+r_5)=2(s_3+s_4)$.
\end{example}

\subsubsection{Specular Duality}\label{specular}
For toric quiver gauge theories, a duality known as specular duality was proposed in \cite{Hanany:2012hi,Hanany:2012vc}. In general, specular duality does not preserve the mesonic moduli spaces (except for self-dual cases) although Hilbert series for duals are the same up to some fugacity map. Instead,
\begin{definition}
Specular duality is a duality that preserves master spaces.
\end{definition}
Therefore, a consequence of Proposition \ref{masteru0} is:
\begin{corollary}\label{speccor}
Given a pair of specular duals $a$ and $b$, suppose the GLSM fields $p^a_i$ and $p^b_i$ are mapped under
\begin{equation}
    p^a_1\leftrightarrow p^b_1,\quad p^a_2\leftrightarrow p^b_2,\quad\dots,\quad p^a_L\leftrightarrow p^b_L.
\end{equation}
If their Newton polynomials are $P_a(z,w)=k-p_a(z,w)$ and $P_b(z,w)=k-p_b(z,w)$, then for $k\geq\max\limits_{|z|=|w|=1}(|p_a(z,w)|,|p_b(z,w)|)$, the two Mahler measures have the series expansions
\begin{equation}
    m(P_a)=\log(k)-\sum_{n=2}^\infty\frac{f_n(p^a_i)}{nk^n},\quad m(P_b)=\log(k)-\sum_{n=2}^\infty\frac{f_n(p^b_i)}{nk^n}.
\end{equation}
Likewise,
\begin{equation}
    u_0(P_a)=1+\sum_{n=2}^\infty\frac{f_n(p^a_i)}{k^n},\quad u_0(P_b)=1+\sum_{n=2}^\infty\frac{f_n(p^b_i)}{k^n}.
\end{equation}
Here, $f_n$ are functions of $p_i^{a,b}$ whose variables are ordered as
\begin{equation}
    f_n(p^a_1,p^a_2,\dots,p^a_L)\text{ and }f_n(p^b_1,p^b_2,\dots,p^b_L).
\end{equation}
\end{corollary}
\begin{proof}
Since specular duality preserves the master space and $u_0(k)$ is a generating function whose expansion reveals the F-term relations, the coefficients for $u_0^a$ and $u_0^b$ should match as functions of perfect matchings. It is straightforward to see that this is the same for the Mahler measures.
\end{proof}

\begin{remark}
Since the Newton polynomials and Mahler measures are invariant under Seiberg duality, Corollary \ref{speccor} is transitive. If a toric phase of polygon $\Delta_1$ is specular dual to toric phase A of $\Delta_2$ and a toric phase of $\Delta_3$ is dual to toric phase B of $\Delta_2$, then the Mahler measures and $u_0$ for $\Delta_1$ and $\Delta_3$ would also satisfy Corollary \ref{speccor}.
\end{remark}

\begin{example}
One of the toric phases for $\mathbb{F}_0$ is specular dual to the single phase for $Y^{2,2}$. Their Newton polynomials are
\begin{equation}
    P_{\mathbb{F}_0}=k-8(z+w+z^{-1}+w^{-1}),\quad P_{Y^{2,2}}=k-9(z+z^{-1}w^{-1}+z^{-1}w+2z^{-1}),
\end{equation}
where the coefficients are taken to be the canonical choice from R-charges, and $k_\textup{iso}$ is reached at $k=32,36$ respectively. For instance, at order 2, one of the ``internal perfect matchings'' on $2\mathcal{G}^{\mathbb{F}_0}_1$ is
\begin{equation}
\tikzset{every picture/.style={line width=0.75pt}}        
.
\end{equation}
On the left hand side, we have two external perfect matchings on $\mathcal{G}_1^{\mathbb{F}_0}$ (in red and orange respectively). Together they form an internal perfect matching on $2\mathcal{G}_1^{\mathbb{F}_0}$. This can be regrouped into two internal perfect matchings on $\mathcal{G}_1^{\mathbb{F}_0}$ as shown on the right hand side in blue and green respectively.

Under specular duality, the perfect matchings in different colours are mapped to
\begin{equation}
\tikzset{every picture/.style={line width=0.75pt}}
%
.
\end{equation}
Notice that now the red and orange ones are internal perfect matchings while the blue and green ones are external on $\mathcal{G}_1^{Y^{2,2}}$.

Overall, at order 2, we have
\begin{equation}
    \begin{split}
        f_2=&2\times2^8\sin\left(\frac{\pi R_H}{2}\right)\sin\left(\frac{\pi R_A}{2}\right)\sin\left(\frac{\pi R_B}{2}\right)\sin\left(\frac{\pi R_G}{2}\right)\\
        &\times\left(\sin\left(\frac{\pi R_K}{2}\right)\sin\left(\frac{\pi R_L}{2}\right)\sin\left(\frac{\pi R_C}{2}\right)\sin\left(\frac{\pi R_D}{2}\right)\right.\\
        &\qquad\left.+\sin\left(\frac{\pi R_I}{2}\right)\sin\left(\frac{\pi R_J}{2}\right)\sin\left(\frac{\pi R_E}{2}\right)\sin\left(\frac{\pi R_F}{2}\right)\right),
    \end{split}
\end{equation}
where the subscripts of the R-charges follow the notations in the perfect matching matrices in \cite[\S4.2]{Hanany:2012vc}. One may check that $f_n$ for the two theories should match at any order $n$.
\end{example}

In the above example, if we plug in the R-charge values for the two theories, we find that
\begin{equation}
    \begin{split}
        &m(P_{\mathbb{F}_0})=\log(k)-\frac{128}{k^2}-\frac{36864}{k^4}-\frac{5242800}{3k^6}-\frac{10276044800}{k^8}-\frac{34093450395648}{k^{10}}-\dots,\\
        &m(P_{Y^{2,2}})=\log(k)-\frac{162}{k^2}-\frac{59049}{k^4}-\frac{35429400}{3k^6}-\frac{52732233225}{2k^8}-\frac{110712378300552}{5k^{10}}-\dots
    \end{split}
\end{equation}
At different orders, the ratios of the coefficients are
\begin{equation}
    \sigma_2=\frac{81}{64},~\sigma_4=\frac{6561}{4096},~\sigma_6=\frac{531441}{262144},~\sigma_8=\frac{43046721}{16777216},~\sigma_{10}=\frac{3486784401}{1073741824},\dots
\end{equation}
One can actually find that
\begin{equation}
    \sigma_n=\frac{f_n^{Y^{2,2}}}{f_n^{\mathbb{F}_0}}=\left(\frac{9}{8}\right)^n.
\end{equation}
This means that we can rescale/normalize the Newton polynomial by dividing $P_{\mathbb{F}_0}$ ($P_{Y^{2,2}}$) by a factor of 8 (9):
\begin{equation}
    P_{\mathbb{F}_0}=k-(z+w+z^{-1}+w^{-1}),\quad P_{Y^{2,2}}=k-(z+z^{-1}w^{-1}+z^{-1}w+2z^{-1}),
\end{equation}
where we have absorbed the factor for the constant term under the redefinition $k\rightarrow k/8$ ($k\rightarrow k/9$). Then the two normalized Newton polynomials have equal Mahler measures (as well as $u_0(k)$).

More generally, if the duals have Mahler measures
\begin{equation}
    m(P_a)=\log(k)-\sum_{n=1}^\infty\frac{\mathfrak{f}_a^nC_n}{nk^n},\quad m(P_b)=\log(k)-\sum_{n=1}^\infty\frac{\mathfrak{f}_b^nC_n}{nk^n}
\end{equation}
for $C_n,\mathfrak{f}^n_{a,b}\in\mathbb{C}$, then we can normalize $P_{a,b}=k-p_{a,b}$ to be $P_{a,b}^\text{norm}=k-p_{a,b}^\text{norm}=k-\frac{p_{a,b}}{\mathfrak{f}_{a,b}}$ such that the two would have equal Mahler measures. However, this is not true in general, and we have found several counterexamples. For instance, the third tiling for PdP$_{4a}$ in \eqref{PdP4atilings} is dual to the single toric phase of PdP$_{4b}$ (see \cite{Hanany:2012hi,Hanany:2012vc} for its details). It turns out that no matter how we normalize the Newton polynomials, their Mahler measure expansions would not have the same numerical coefficients though the discrepancies are very small\footnote{To avoid any possible confusion, we should emphasize that only the numerical coefficients differ in these cases. Corollary \ref{speccor} holds for any polygons.}.

Remarkably, for all the examples whose \emph{numerical} $f_n$'s are different from their specular duals (under any normalization), we find that they coincide with polynomials of second type in Example \ref{twotypes}. It would be interesting to explore this more deeply in future.

It is also natural to ask how specular duals are related at isoradial point. Since the number of chiral multiplets is invariant under specular duality, the number of summands in \eqref{isomahler} would not change. Moreover, a {\bf zig-zag path} is mapped to a face in the specular dual tiling \cite{Hanany:2012vc}. A zig-zag path is a collection of edges that forms a closed path on the brane tiling. It maximally turns left/right at a black/white node. The winding number $(p,q)$ of the zig-zag path corresponds to a direction in the dual web of the toric diagram. Physically, zig-zag paths can be interpreted as gauge invariant operators \cite{Hanany:2005ss}. On the other hand, a node in the tiling is mapped to a node in the dual tiling. In terms of superpotentials, this reverses the order of half the terms based on the convention of untwisting the zig-zag paths. Now that we have the $\text{zig-zag}\leftrightarrow\text{face}$ and $\text{node}\leftrightarrow\text{node}$ mappings, we can write down how Mahler measures would transform.
\begin{proposition}\label{specprop}
Suppose a brane tiling $G$ has $c$ perfect matchings. Then for the specular dual tiling $G'$ with edges $e_I$, the Mahler measure is
\begin{equation}
    m=\sum_I\left(\frac{\theta_I}{\pi}\log(2\sin(\theta_I))+\frac{1}{\pi}\Lambda(\theta_I)\right),
\end{equation}
where $\theta_I$ (and $m$) can be obtained by maximizing $m$ subject to the conditions
\begin{equation}
    \sum_{e_I\in\mathcal{I}}\theta_I=\pi,~\sum_{e_I\in\mathcal{Z}}(\pi-2\theta_I)=2\pi
\end{equation}
for all nodes $\mathcal{I}$ in $G$ \footnote{In this case, ``$\in$'' indicates the edges attached to a node.} and all zig-zag paths $\mathcal{Z}$ in $G$.
\end{proposition}

Notice that specular duality, in particular for non-reflexive polygons, may require the dimers to be embedded on a Riemann surface of any genus $g$ rather than just a torus. Nevertheless, the discussions in this subsection would still hold for any polygons. This is because Corollary \ref{speccor} only requires the invariance of the master spaces, and for Proposition \ref{specprop}, \eqref{isomahler} is always true for dimers which are (doubly) periodic on general bi-dimensional lattices \cite{de2007partition}.

\subsection{Tropical Limit}\label{troplim}
In this part, we shall focus on another special point in the Mahler flow, that is, $k_\text{trop}$. Recall that in the large/(sub)tropical $k$ limit, the amoeba tends to its spine and the gas phase becomes dominant. In particular, in this limit the $k^{-n}$ terms in the Mahler measure tend to zero and we have $m(P)\sim\log k$. Moreover, the area of the hole in the amoeba is $A_h\sim\log^2k$. Therefore, in the large $k$ limit, the Mahler measure in the tropical limit follows an area law
\begin{equation}
    m(P)\sim A_h^{1/2}.
\end{equation}
Likewise, the Ronkin function would be dominated by the linear facets. More precisely,
\begin{equation}
    R(x,y)\simeq\log|c_{(m,n)}|+mx+ny
\end{equation}
for $P=\sum\limits_{(m,n)}c_{(m,n)}z^mw^n$. For each linear facet, this equality is exact. For the liquid phases, it gives a good approximation as the non-linear regions tends to the spine at large $k$. In fact, as the amoeba $\mathcal{A}$ in the (sub)tropical limit is a union of local amoebae $\mathcal{A}_\text{loc}$, the (local) non-linear part of $R(x,y)$ for $\mathcal{A}$ would be the translations and rotations of the Ronkin functions for $\mathcal{A}_\text{loc}$, where the detailed translations and rotations can be determined by the neighbour linear facets.

From the perspective of crystal melting, (almost) the whole crystal is molten. In other words, the system becomes a gas of atoms. Since the linear facets of the Ronkin function are sent to infinity, the partition function would also diverge as expected.

In this large $k$ limit, we may then estimate the free energy in \eqref{ZBPS} as
\begin{equation}
    F=-\log Z\sim-A_h\log k\sim-\log^3k.
\end{equation}
In other words,
\begin{equation}
    F\sim-m^3(P).
\end{equation}

Now we shall consider the meaning of tropical limit from the perspective of gauge theories. The F-term relations are encoded by $u_0(k)$ for the master space. In the large $k$ limit, we have $u_0\rightarrow1$. Therefore, all the constraints on the GLSM fields from $Q_F$ are lost in the tropical limit. As a result, the master space would become the trivial $\mathbb{C}^n$, and all the GLSM fields become free.

This is also reflected by the amoeba. In the large $k$ limit, the amoeba is composed of local vertices in the spine. These vertices are only connected by thin long channels/lines.

\subsection{Phase Transition}\label{transition}
We have now discussed the physics at $k_\text{iso}$ and $k_\text{(sub)trop}$. Let us consider an interesting intermediate phenomenon for the Mahler measure along the Mahler flow.

When $k\leq k_\text{iso}$, there is no hole in the amoeba. In particular, for $k=k_\text{iso}$, we know that the amoeba and its complements are precisely the liquid and solid phases respectively. When $k>k_\text{iso}$, holes would open up in the amoeba\footnote{For the Newton polynomials $P(z,w)=k-p(z,w)$ considered in this paper, all possible holes would appear when we cross the isoradial point. In other words, the number of gas phase regions would be equal to the number of interior points in the Newton polygon. When we further increase $k$, the sizes of all holes would increase.}.
\begin{remark}
The isoradial point $k_\textup{iso}$ is the critical point of liquid-gas phase transition. Gas phases would appear when $k>k_\textup{iso}$. Mathematically, $k_\textup{iso}$ is the critical point of the spectral curve transition from non-Harnack to Harnack. From the perspective of D-branes, there are only D\,$2$- and D\,$0$-branes when $k\leq k_\textup{iso}$. The D\,$4$-branes are added to the system when $k>k_\textup{iso}$.
\end{remark}

Therefore, we might think of the Mahler measure receiving contributions from both phases.
\begin{definition}
Let $m_l(P)$ and $m_g(P)$ be the liquid and gas phase contributions to the Mahler measure respectively, which satisfy the following properties:
\begin{itemize}
    \item They give the full contribution to Mahler measure, i.e.,
    \begin{equation}
        m(P)=m_g(P)+m_l(P);
    \end{equation}
    \item When there is no gas phase, only $m_l(P)$ would contribute, i.e.,
    \begin{equation}
        m(P)=m_l(P),\quad m_g(P)=0
    \end{equation}
    for $k\leq k_\textup{iso}$ \footnote{When the spectral curve is non-Harnack, let us still define the ``liquid'' contribution as $m_l=m$ and the ``gas'' contribution as $m_g=0$ since the amoeba has no holes and the crystal has no gas phases. Although how the amoeba region and its complements are related to different crystal phases is still not clear, it should not affect the well-definedness here.};
    \item When there are gas phases, only the gas phase contribution would change  along the Mahler flow, i.e.,
    \begin{equation}
        m_l(P)=m_\textup{iso}(P)=\textup{const},
    \end{equation}
    for $k>k_\textup{iso}$, and $m_\textup{iso}(P)$ can be computed using \eqref{isomahler}.
\end{itemize}
\end{definition}
Following the definition, it is straightforward to see that along the Mahler flow, we have
\begin{equation}
    \begin{cases}
    m_{l,k_1}(P)<m_{l,k_2}(P)\leq m_{l,\max}(P)=m_\text{iso}(P),&\quad k_1<k_2\leq k_\text{iso}\\
    m_{g,k_1}(P)<m_{g,k_2}(P)\text{ and }m_l(P)=m_\text{iso}(P)=\text{const},&\quad k_2>k_1\geq k_\text{iso}
    \end{cases}.
\end{equation}

Recall that for the area of amoeba, we have
\begin{equation}
    A(\mathcal{A}_{k_1})<A(\mathcal{A}_{k_2})\leq\pi^2A(\Delta),\quad A_{\max}(\mathcal{A})=A(\mathcal{A}_\text{iso})=\pi^2A(\Delta),
\end{equation}
when $k_1<k_2<k_\text{iso}$, where $\Delta$ is the corresponding Newton polygon. When $k\geq k_\text{iso}$, the curve is always Harnack. Hence,
\begin{equation}
    A(\mathcal{A})=A(\mathcal{A}_\text{iso})=\pi^2A(\Delta).
\end{equation}
Therefore,
\begin{conjecture}
Given a Newton polynomial $P(z,w)=k-p(z,w)$ with fixed $p(z,w)$, $m_l(P)$ is only determined by the area of the amoeba $A(\mathcal{A})$, and $m_g(P)$ is only determined by the area of its holes $A_h$.
\end{conjecture}

For (sub)tropical $k$,
\begin{equation}
    m(P)\sim m_g(P)\sim\log k,\quad m_g(P)\gg m_l(P).
\end{equation}
In terms of the hole of the amoeba, we have
\begin{equation}
    m_g(P)\sim\log k\sim A_h^{1/2}.
\end{equation}

\begin{remark}
We shall again emphasize that $P(z,w)=0$ is not Harnack for $k<k_\textup{iso}$. We do not have the 2-to-1 property between spectral curves and amoebae. Hence, the solid/liquid phase separation may not precisely corrspond to the amoeba boundary. Nevertheless, the area of the amoeba would still increase when we increase $k$. Moreover, we would not expect any gas phase since the Mahler flow is continuous and there is no gas phase for the Harnack $P_\textup{iso}(z,w)=0$. Therefore, it is still natural to expect that the Mahler measure (and Ronkin function) would account for liquid and solid phases in terms of the area of the amoeba in a more subtle way. In other words, the amoeba and its complementary regions may still be related to different phases, but the separation of these phases may not coincide with the boundary of amoeba.
\end{remark}

Now let us describe how $m(P)$ serves an indicator of the phase transition. When computing the Mahler measure using Taylor expansion, this is only valid for $k\geq k_0=\max\limits_{|z|=|w|=1}(|p|)=p(1,1)$. Therefore, one would expect some changes for $m(P)$ and $u_0(k)$ at $k_0$. Mathematically, $k_0$ is the point such that for $k<k_0$,
\begin{equation}
    m(P)=\text{Re}\left(\int_{\mathbb{T}^2}\log(k-p)\frac{\text{d}z}{z}\frac{\text{d}w}{w}\right)\neq\int_{\mathbb{T}^2}\log(k-p)\frac{\text{d}z}{z}\frac{\text{d}w}{w}
\end{equation}
and $\frac{\text{d}m(P)}{\text{d}\log k}$ is no longer the period of the holomorphic differential on the curve $1-k^{-1}p=0$. In fact,
\begin{proposition}
For $m(P)$ and $u_0(k)$ with generic $P(z,w)$, the changes of their behaviours at $k=k_0$ indicate the liquid-gas phase transition. Since the transition of the system starts at $k=k_\textup{iso}\leq k_0$, there is a ``delay'' in the indication by Mahler measure along the Mahler flow. By ``delay'', we mean that the changes of $m(P)$ and $u_0(k)$ appear later than the emergence of gas phases in the whole system.\label{transitionprop}
\end{proposition}
\begin{proof}
Since the spectral curve and the amoeba $\mathcal{A}$ are 2-to-1 for int$(\mathcal{A})$ and 1-to-1 for $\partial\mathcal{A}$, the origin would transit to the gas phase precisely at $k_0=\max\limits_{|z|=|w|=1}(|p|)$. Hence, although the gas phases start to appear at $k=k_\text{iso}\leq k_0$, the origin would only be included in a hole of the amoeba when $k\geq k_0$. The changes of $m(P)$ and $u_0$ at $k_0$ discussed above would then account for this since the Mahler measure is the Ronkin function at $(0,0)$.
\end{proof}

\begin{example}
As an example, consider $\mathbb{C}^3/(\mathbb{Z}_4\times\mathbb{Z}_2)$ $(1,0,3)(0,1,1)$ whose Newton polynomial is
\begin{equation}
    P=-zw-zw^{-1}-z^{-3}w^{-1}-2z-2z^{-1}-4w^{-1}-6z^{-1}w^{-1}-4z^{-2}w^{-1}+k.\label{C3Z4Z2}
\end{equation}
The isoradial point is $k_\textup{iso}=12$ while $k_0=21>k_\textup{iso}$. We plot $m(P)$ and $u_0(k)$ along the Mahler flow numerically:
\begin{equation}
    \includegraphics[width=4cm]{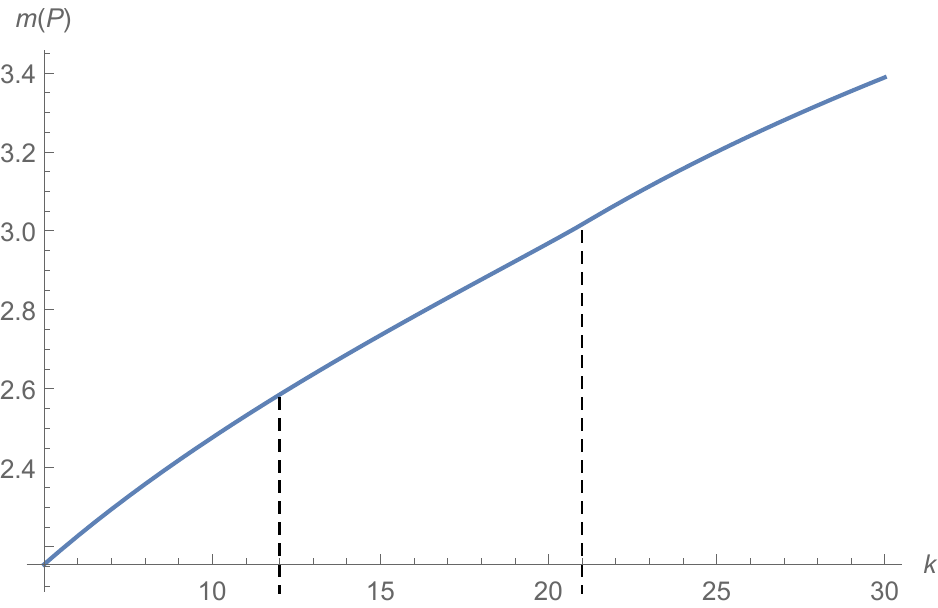}\quad\quad
    \includegraphics[width=4cm]{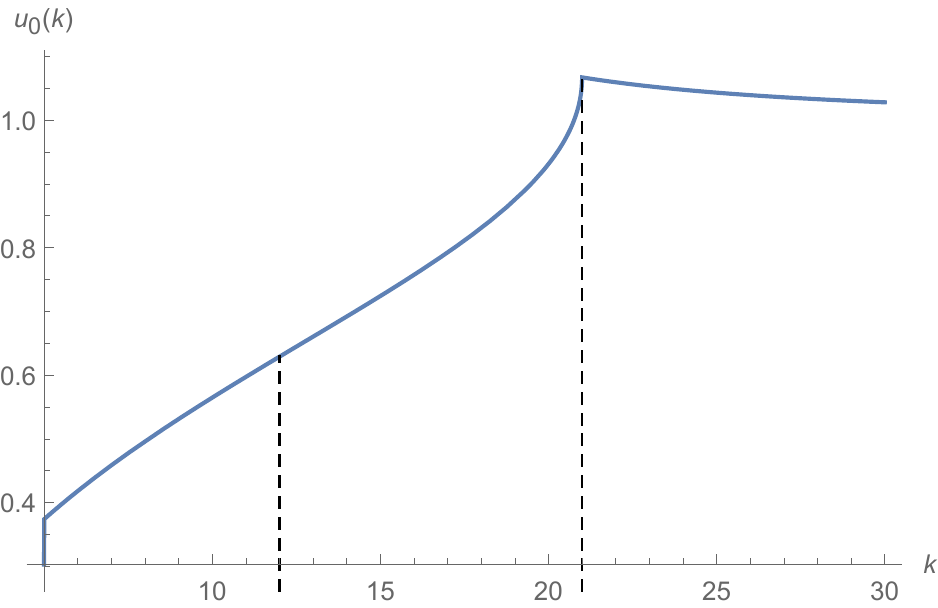},\label{transitionex}
\end{equation}
where we have used the Mahler flow equation $\textup{d}m(P)/\textup{d}\log k=u_0(k)$ to plot $u_0(k)$ for $k<k_0$. The two dashed lines are $k=k_\textup{iso}$ and $k=k_0$ respectively. As we can see, the curve for $u_0(k)$ has an abrupt change at $k=k_0$.
\end{example}

As we can see, since the origin is still in liquid phase for $k_\text{iso}\leq k<k_0$, the trajectory of $m(P)$ and its derivative have no changes at $k_\text{iso}$. The liquid-gas transition is reflected by $m(P)$ and $u_0$ at a later $k$, i.e., $k=k_0$.

\begin{remark}
As mentioned above, we may consider generalized Mahler measure at the point $(x_0,y_0)$ where a hole opens up to resolve this ``delay''. Equivalently, we can again make a rescaling of variables in $P(z,w)$ so that $k_\textup{iso}$ and $k_0$ would coincide (recall \eqref{rescale}). This just shifts the amoeba/Ronkin function by $(x_0,y_0)\rightarrow(0,0)$.
\end{remark}

\begin{corollary}
For $m(P)$ and $u_0(k)$ of $P(z,w)$ (with possible rescaling of variables), the emergence of liquid-gas phase transition at $k=k_\textup{iso}$ is indicated by the changes of their behaviours.\label{transitioncor}
\end{corollary}

\begin{example}
Let us rescale the variables in \eqref{C3Z4Z2} by $z\rightarrow z$, $w\rightarrow 4w$. This yields
\begin{equation}
    \Tilde{P}=-4zw-\frac{1}{4}zw^{-1}-\frac{1}{4}z^{-3}w^{-1}-2z-2z^{-1}-w^{-1}-\frac{3}{2}z^{-1}w^{-1}-z^{-2}w^{-1}+k.
\end{equation}
One may check that now $k_0=12=k_\textup{iso}$. Again, we plot $m(P)$ and $u_0(k)$ numerically:
\begin{equation}
    \includegraphics[width=4cm]{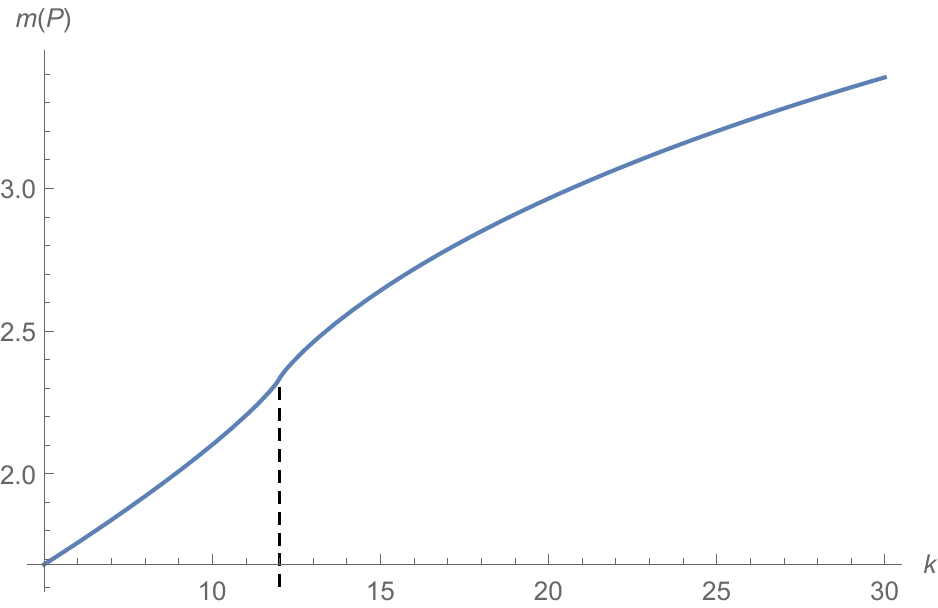}\quad\quad
    \includegraphics[width=4cm]{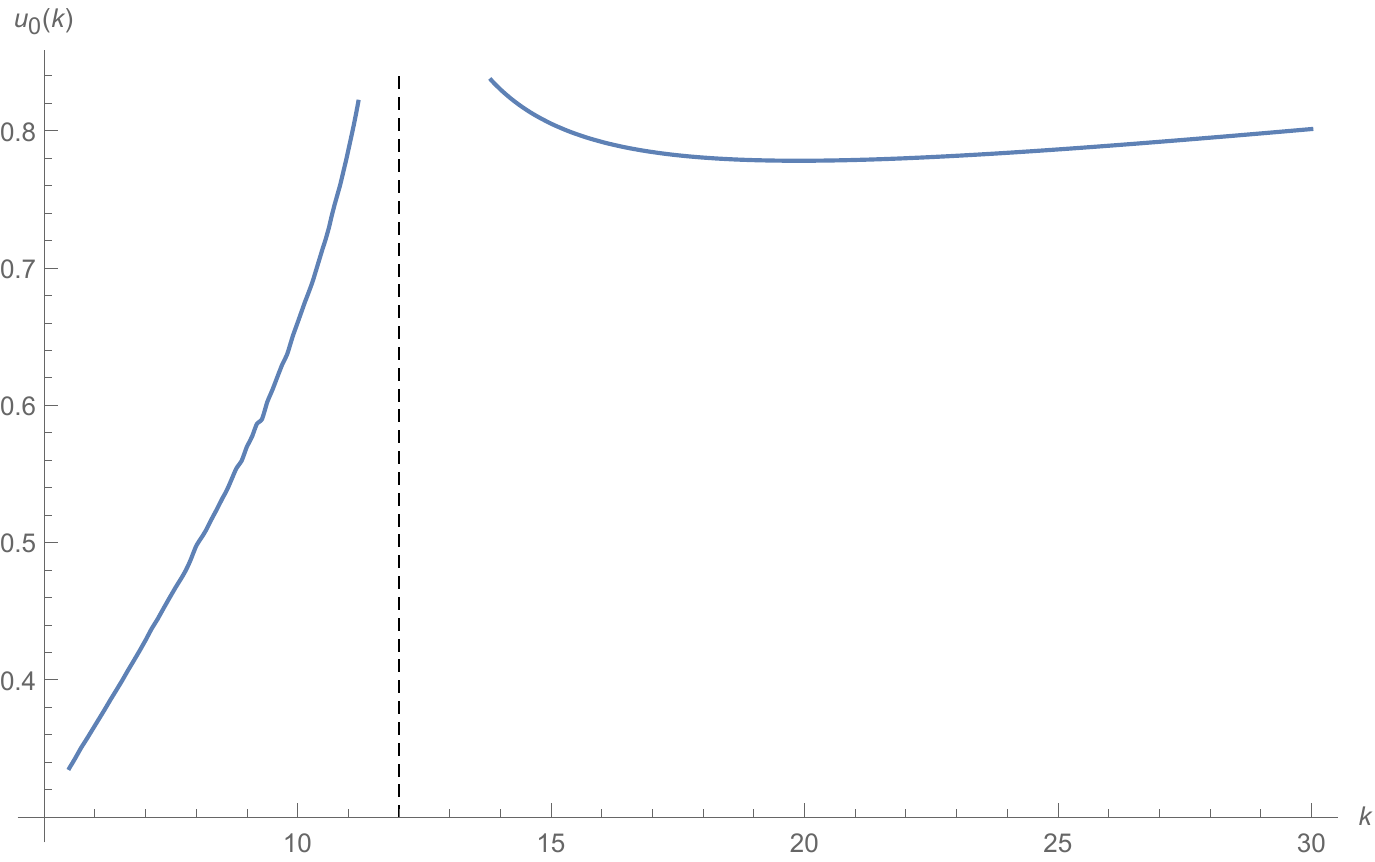},
\end{equation}
where the dashed line is $k=12$. As we can see, the Mahler flow has a change at $k_\textup{iso}=k_0$. This is more clear for $u_0$ where it goes to infinity at $k=12$.
\end{example}

The above discussions would also make sense in terms of the free energy of crystal melting model. As the free energy is an indicator of phase transitions, the Ronkin function would also reflect phase transitions via $F=-\int\text{d}x\text{d}yR(x,y)$. On the other hand, $R(x,y)$ is essentially the generalized Mahler measure at point $(x,y)$. Equivalently, $m(P)$ detects the Ronkin function at $(0,0)$, and we can always shift the amoeba to move any point $(x,y)$ to the origin.

\subsection{Towards a Universal Measure}\label{universal}
For crystal melting and topological strings, the physical interpretations of Mahler measures and Ronkin functions are quite clear as reviewed in \S\ref{crystal}. We shall now discuss their physical meanings from the perspective of quiver gauge theories.
There is a very straightforward implication if we consider the GLSM fields. From \cite{Kenyon:2003uj}, we know that the partition function for perfect matchings (in the thermodynamic limit) can be determined by Ronkin functions. Therefore,
\begin{proposition}
The partition function $Z$ for GLSM fields can be determined via
\begin{equation}
    \log Z=\int\textup{d}x\,\textup{d}y\,R(x,y),
\end{equation}
which also defines the free energy of the dimer/GLSM as $F\equiv-\log Z$.
\end{proposition}
Of course, such expression would diverge, so we need to normalize it by $Z/Z_0$, where $Z_0$ is the partition function whose Ronkin function only has linear (solid) facets.

A more non-trivial interpretation would be the connection to 4d superconformal index $I_{4\text{d}}$ on $S^3\times S^1$. As studied in \cite{Yamazaki:2012cp,Terashima:2012cx}, when the radius of $S^1$ goes to 0, $I_{4d}$ would reduce to the partition function $Z_{3\text{d}}$ on the ellipsoid $S^3_b=\{(z_1,z_2)\in\mathbb{C}^2|b^2|z_1|^2+b^{-2}|z_2|=1\}$. By further taking $b\rightarrow0$, the 3d partition function would give the partition function for 2d $\mathcal{N}=(2,2)$ theory:
\begin{equation}
    Z_{2\text{d}}=\int\text{d}\sigma\exp\left(-\frac{1}{\pi b^2}W_{2\text{d}}(\sigma)\right),
\end{equation}
where $\sigma$ is the scalar in the vector multiplet. The effective twisted superpotential $W_{2\text{d}}$ can then be identified with the volume of some hyperbolic 3-fold $\mathcal{M}$. This 3-fold can be determined by (the zig-zag paths on) the dimer model. Moreover, the genus-0 prepotential for topological B-model is
\begin{equation}
    F_0=\int\dd x\,\dd y\,R(x,y)=\int\text{d}x\,\text{d}y\, \mathcal{L}(\text{vol}(\mathcal{M})),
\end{equation}
where $\mathcal{L}$ denotes the Legendre transformation. As we can see \cite{Yamazaki:2012cp},
\begin{proposition}
The $4$d superconformal index under dimensional reduction is related to topological string partition function and Ronkin function via Legendre transformation.
\end{proposition}

The topological string partition functions are invariant under wall crossing. Its counterpart in 4d, i.e., the superconformal indices, are also invariant under Seiberg duality. This would provide a further evidence that the Mahler measure/Ronkin function should be invariant under Seiberg duality.

Based on the above analysis, we wish to establish a universal measure for quivers. Mahler measure can be defined for any quivers arisen from toric geometry. Furthermore, it is monotonically increasing along the Mahler flow when we increase $k$ regardless of the detailed information of the gauge theory.
Physically, the Mahler measure can be used to determine the R-symmetry of the IR theory. Moreover, recall that
\begin{equation}
    Z_\textup{BPS}\sim\exp\left(\int\textup{d}x\,\textup{d}y\,R(x,y)\right).
\end{equation}
In other words, the Ronkin function characterizes the growth rate of BPS states in the theory. As Mahler measure (with possible rescaling of variables) is the minimum of Ronkin function, it provides a measure of the D-brane bound states.

\section{Conclusions and Outlook}\label{outlook}
In this paper, we studied the properties and physical implications of Mahler measure in the context of quiver gauge theories. We found that maximization of Mahler measure at the isoradial point leads to the correct R-symmetry in the infrared. We also discussed how the Mahler measure, its logarithmic derivative, and the Ronkin function behave under Seiberg duality and specular duality.

Moreover, we introduced {\it Mahler flow}, which encodes many interesting features
such as the flow equation and associated phase transition. Geometrically, Mahler flow can be studied using amoebae and tropical geometry. Physically, we proposed that the Mahler measure, which has a monotonic behaviour along Mahler flow, provides a universal measure of the quiver gauge theory.

The study of Mahler measure and Ronkin function is closely related to many aspects in physics including crystal melting, topological strings, black holes, etc. There is still much to explore in various topics, which may reveal deep connections among them. 
We list a few tantalizing ones below.

\paragraph{Multi-dimensional Mahler flow} Here, we only discussed Mahler flows with respect to a single parameter in the Newton polygon, viz., the constant term $k$. In general, it should be possible to consider varying all polynomial coefficients $k_i$, which dictate the complex structure of the geometry. This variation should constitute a multi-dimensional flow. See for example \cite{stienstra1997resonant,deninger1997deligne} for possible connections to various aspects in mathematics.

\paragraph{Non-isoradial embeddings} With the canonical weight choice, we obtained nice results for non-isoradial dimers. In particular, some tools in isoradial embeddings could also be applied to these non-isoradial ones. Therefore, one may wonder if there could be any similar properties for non-isoradial embeddings such as maximality for some Dirac operator in the sense of \cite{Kenyon_2002}.

\paragraph{Modularity} The modularity of Mahler measures \cite{villegas1999modular,Stienstra:2005wy} is a crucial concept in number theory. Along with the Picard-Fuchs equation, this would lead to important connections and applications in physics (e.g. recent emergence of modularity in \cite{Cheng:2018vpl,Harvey:2014cva}).

\paragraph{Non-Harnack curves} Unlike Harnack curves, the amoebae for non-Harnack curves behave differently. Hence, the boundaries of amoebae may not be the exact separations of different phases. Nevertheless, the area and complementary regions of an amoeba might still indicate the phase structures in a more intricate way.

\paragraph{Superconformal indices} As discussed in the previous section, the Mahler measure has an intimate relation with superconformal indices. It also has some connections to $F$-theorem in 3d \cite{Yamazaki:2012cp}. A more detailed study would give us a better understanding of the physics for Mahler measure in gauge theories.

\paragraph{Holes of the amoeba} We considered how the hole area $A_h$ can be computed in the large $k$ limit. However, how to obtain $A_h$ for small $k$ remain unsolved. We hope the discussions on liquid and gas phase contributions to the Mahler measure could shed some light on this topic.

\paragraph{Quiver entropy and black holes} Since the Mahler measure and Ronkin function provide a measure of degeneracy of D-brane bound states, it would be possible to define a quiver entropy from this. It would also be important to study if relation with the surface tension of the crystal model, which is the Legendre dual of the Ronkin function. In \cite[eqn(2.25)]{Feng:2007ur}, another quiver entropy was defined in terms of the plethystic exponential of the Hilbert series. How this would be connected to the quiver entropy from the Mahler measure is still an interesting open question. On the other hand, the famous OSV conjecture \cite{Ooguri:2004zv} says that $Z_\text{BH}=|Z_\text{topo}|^2$ when the D-brane bound states become black holes with smooth event horizon. The black hole entropy in the supergravity approximation is also the Legendre transformation of the free energy of topological A-model at genus 0. We will study how Mahler measures and quiver entropy are related to black holes in our future work.

\section*{Acknowledgments}
We are grateful to Sebastian Franco, Amihay Hanany and Masahito Yamazaki for fruitful discussions. JB is supported by a CSC scholarship. YHH would like to thank STFC for grant ST/J00037X/1. The research of AZ has been supported by
the French “Investissements d’Avenir” program, project ISITE-BFC (No. ANR-15-IDEX-0003), and EIPHI Graduate School (No. ANR-17-EURE-0002).

\appendix
\section{Example: Mahler Measure and Amoeba for $\mathbb{F}_0$}\label{exF0}
We can compute the Mahler measure for any Newton polygon/polynomial directly from the definition. Let us use \eqref{mahlerre} to perform the integration. As an example, let us consider the familiar $\mathbb{F}_0$ whose Newton polynomial is $P(z,w)=-z-z^{-1}-w-w^{-1}+k$. Then its Mahler measure is
\begin{equation}
\begin{split}
    m(P)&=\frac{1}{(2\pi i)^2}\int_{|z|=|w|=1}\log|k-(z+z^{-1}+w+w^{-1})|\frac{\text{d}z}{z}\frac{\text{d}w}{w}\\
    &=\text{Re}\left(\frac{1}{(2\pi i)^2}\int_{|z|=|w|=1}\log(k-(z+z^{-1}+w+w^{-1}))\frac{\text{d}z}{z}\frac{\text{d}w}{w}\right).
\end{split}
\end{equation}
The log part can be expanded as
\begin{equation}
    \log(k-(z+z^{-1}+w+w^{-1}))=\log k-\sum_{n=1}^\infty\frac{(z+z^{-1}+w+w^{-1})^n}{n}k^{-n},
\end{equation}
where $|k|>4$ or $k=4$ as $\max\limits_{|z|=|w|=1}|p(z,w)|=4$.

For each summand in the above sum, they can be further expanded as
\begin{equation}
    \begin{split}
        (z+z^{-1}+w+w^{-1})^n&=\sum_{i=0}^n\binom{n}{i}(z+z^{-1})^i(w+w^{-1})^{n-i}\\
        &=\sum_{i=0}^n\binom{n}{i}\left(\sum_{j=0}^i\binom{i}{j}z^{2j-i}\right)\left(\sum_{l=0}^{n-i}\binom{n-i}{l}w^{2l-n+i}\right).
    \end{split}
\end{equation}
As the only contribution to the integral, its constant term satisfies $i=2j,n-i=2l$. Therefore, we can write the constant term as
\begin{equation}
    [(z+z^{-1}+w+w^{-1})^n]_0=\sum_{\mathclap{\substack{i=0\\i\text{ even}}}}^n\binom{n}{i}\binom{i}{i/2}\binom{n-i}{(n-i)/2},
\end{equation}
where $[Q]_0$ denotes the constant term of $Q$. Equivalently, this can written as
\begin{equation}
    \begin{split}
        [(z+z^{-1}+w+w^{-1})^{2n}]_0&=\sum_{i=0}^n\binom{2n}{2i}\binom{2i}{i}\binom{2n-2i}{n-i}\\
        &=\frac{4^{2n}\left(\frac{2n-1}{2}\right)!^2}{\pi n!^2}=\frac{16^n\Gamma^2\left(\frac{2n+1}{2}\right)}{\pi n!^2}
    \end{split}\label{constterm2n}
\end{equation}
while $[(z+z^{-1}+w+w^{-1})^{2n-1}]_0$ vanishes. Hence,
\begin{equation}
    \begin{split}
        \left[\sum_{n=1}^\infty\frac{(z+z^{-1}+w+w^{-1})^n}{n}k^{-n}\right]_0&=\left[\sum_{n=1}^\infty\frac{(z+z^{-1}+w+w^{-1})^{2n}}{2n}k^{-2n}\right]_0\\
        &=\sum_{n=1}^\infty\frac{16^n\Gamma^2\left(\frac{2n+1}{2}\right)}{\pi n!^2}\frac{k^{-2n}}{2n}\\
        &=2k^{-2}{}_4F_3\left(1,1,\frac{3}{2},\frac{3}{2};2,2,2;16k^{-2}\right),
    \end{split}
\end{equation}
where ${}_pF_q$ is the (generalized) hypergeometric function. Therefore, we get
\begin{equation}
    m(P)=\text{Re}\left(\log k-2k^{-2}{}_4F_3\left(1,1,\frac{3}{2},\frac{3}{2};2,2,2;16k^{-2}\right)\right),~~~|k|>4\text{ or }k=4.
\end{equation}
In this paper, we are mainly interested in the case $k\geq4$. Therefore,
\begin{equation}
    m(P)=\log k-2k^{-2}{}_4F_3\left(1,1,\frac{3}{2},\frac{3}{2};2,2,2;16k^{-2}\right),~~~k\geq4.
\end{equation}

We may as well calculate the period $u_0(k)$ following the same manner. Recall that $u_0(k)$ can be written as the integral
\begin{equation}
    u_0(k)=\frac{1}{(2\pi i)^2}\int_{|z|=|w|=1}\frac{1}{1-k^{-1}(z+z^{-1}+w+w^{-1})}\frac{\text{d}z}{z}\frac{\text{d}w}{w}.
\end{equation}
Then we have the series expansion
\begin{equation}
    \frac{1}{1-k^{-1}(z+z^{-1}+w+w^{-1})}=\sum_{n=0}^\infty k^{-n}(z+z^{-1}+w+w^{-1})^n,~~~|k|>4.
\end{equation}
Since we are still dealing with $(z+z^{-1}+w+w^{-1})^n$, the constant terms are still the same as \eqref{constterm2n}. Hence,
\begin{equation}
    u_0(k)=\sum_{n=0}^\infty\frac{16^n\Gamma^2\left(\frac{2n+1}{2}\right)}{\pi n!^2}k^{-2n}=\frac{2}{\pi}K(4k^{-1})={}_2F_1\left(\frac{1}{2},\frac{1}{2};1;16k^{-2}\right),~~~|k|>4,
\end{equation}
where $K(x)=\int_0^1\frac{\text{d}t}{\sqrt{(1-t^2)(1-x^2t^2)}}$ is the elliptic integral of the first kind. As we can see, $u_0(k)$ is simply a hypergeometric function and $m(P)$ can also be expressed concisely using some hypergeometric function. In general, there may not be such compact expressions for $m(P)$ and $u_0(k)$. We will then compute them perturbatively in terms of their series expansions as
\begin{equation}
    m(P)=\log k-\sum_{n=2}^\infty\frac{f_n}{nk^n},\quad u_0(k)=1+\sum_{n=2}^\infty\frac{f_n}{k^n}.
\end{equation}

In Figure \ref{flow_amoeba}, we plot $m(P)$ along the Mahler flow. We also plot several amoebae using Monte-Carlo for different $k$ to illustrate how the amoebae would change along the Mahler flow\footnote{In \cite{Bao:2021olg}, this process is called the crawling of amoeba.}.
\begin{figure}[h]
    \centering
    \begin{subfigure}{4cm}
		\centering
		\includegraphics[width=4cm]{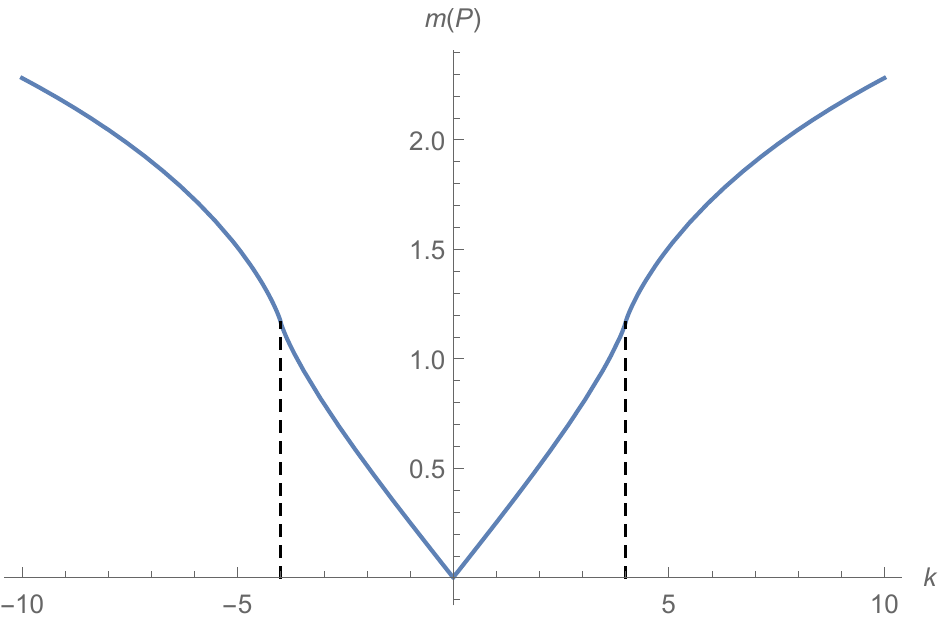}
		\caption{}
	\end{subfigure}
    \begin{subfigure}{4cm}
		\centering
		\includegraphics[width=4cm]{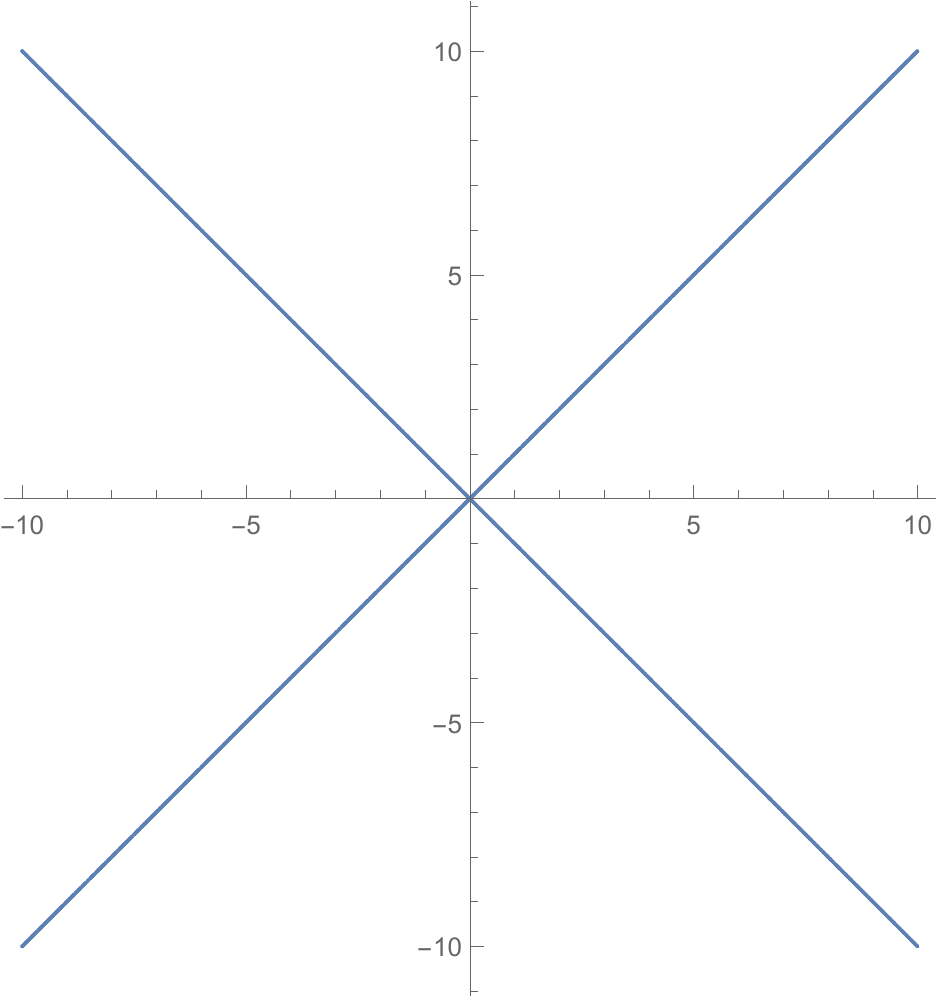}
		\caption{}
	\end{subfigure}
	\begin{subfigure}{4cm}
		\centering
		\includegraphics[width=4cm]{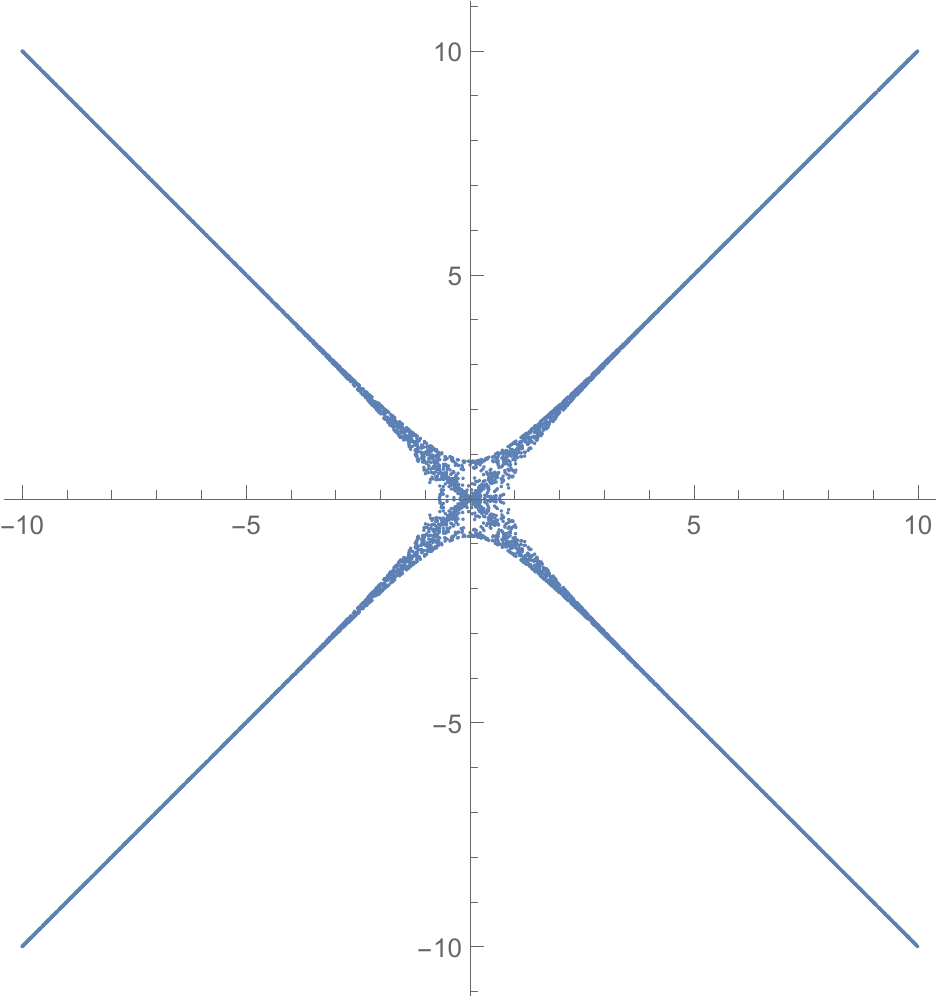}
		\caption{}
	\end{subfigure}
	\begin{subfigure}{4cm}
		\centering
		\includegraphics[width=4cm]{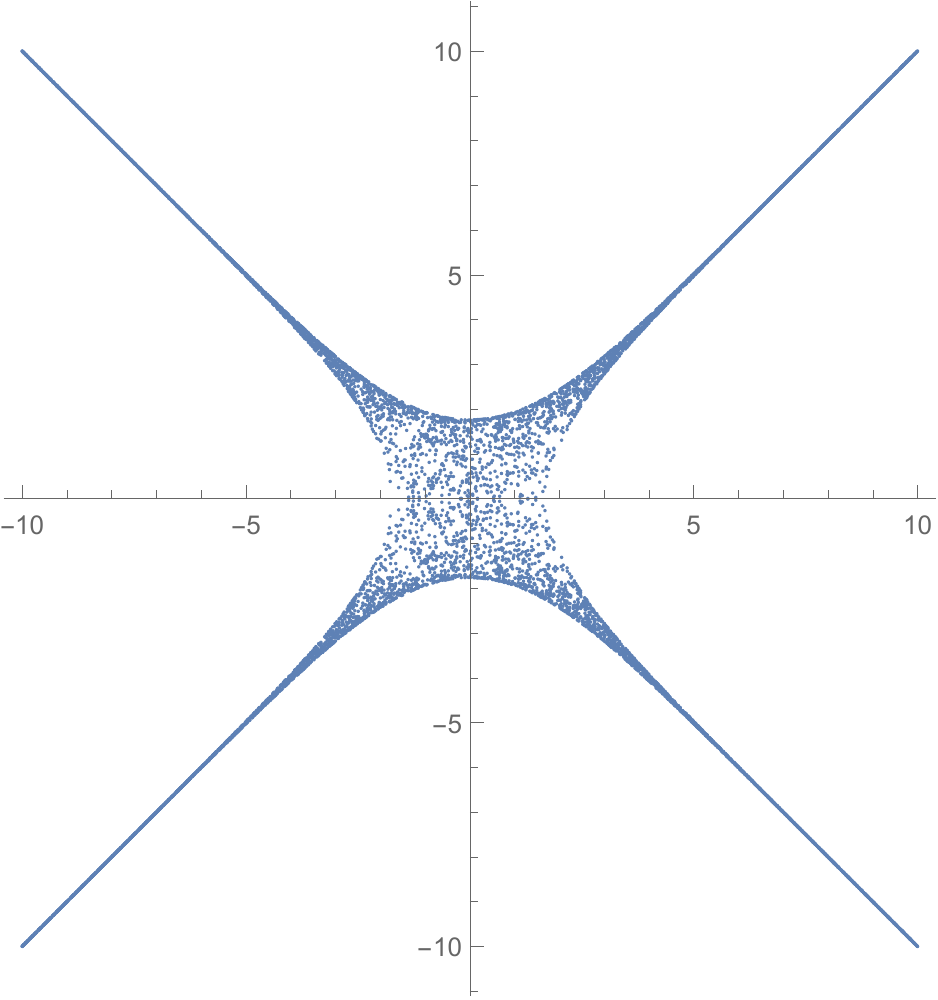}
		\caption{}
	\end{subfigure}
	\begin{subfigure}{4cm}
		\centering
		\includegraphics[width=4cm]{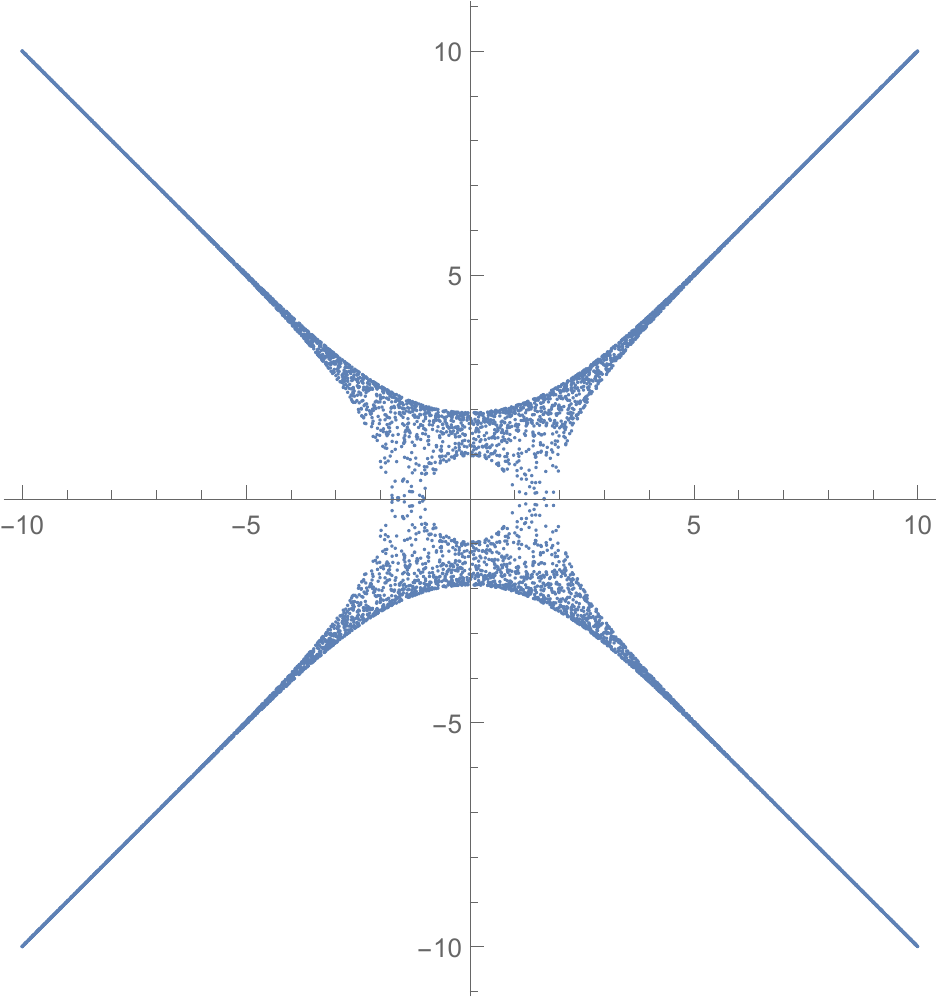}
		\caption{}
	\end{subfigure}
	\begin{subfigure}{4cm}
		\centering
		\includegraphics[width=4cm]{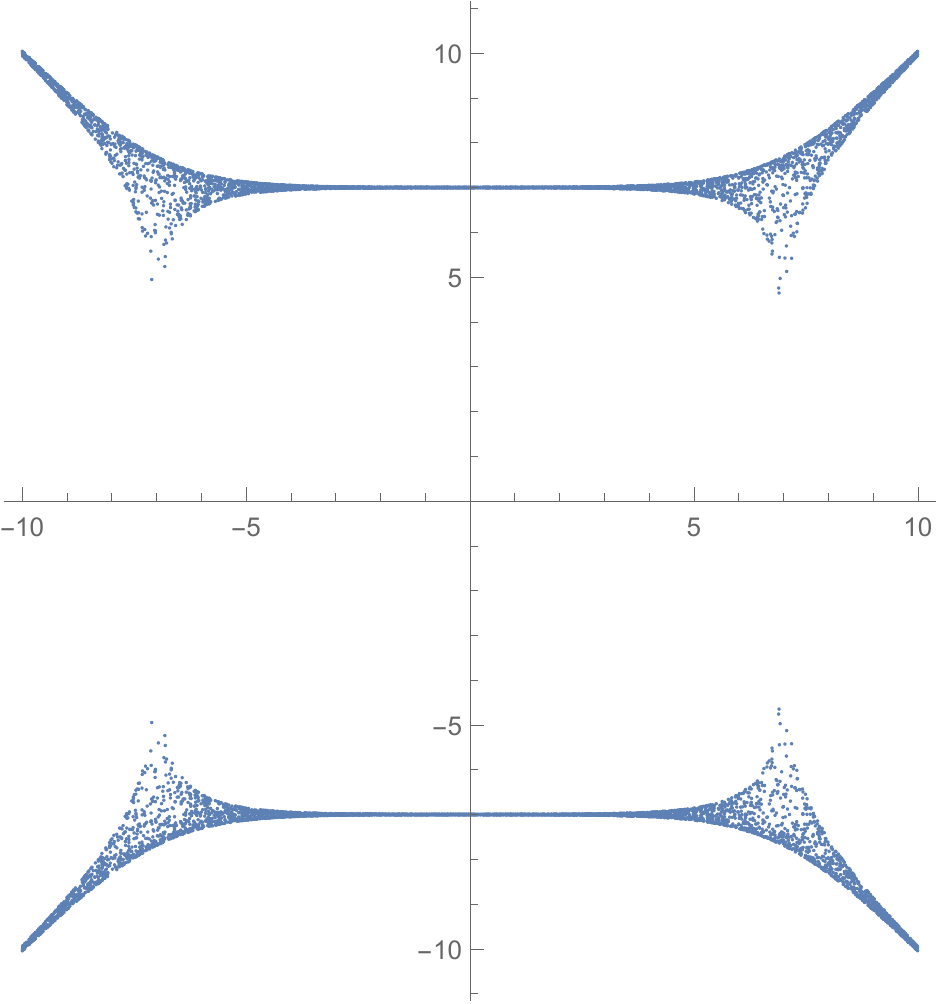}
		\caption{}
	\end{subfigure}
	\caption{The numerical $m(P)$ along the Mahler flow in (a). The amoeba for $P(z,w)$ when (b) $k=0$, (c) $k=3/4$, (d) $k=k_\text{iso}=4$, (e) $k=5$ and (f) $k=\text{e}^7$.}\label{flow_amoeba}
\end{figure}

Incidentally, we also plot the Mahler measures for negative $k$ in Figure \ref{flow_amoeba}(a) which give the mirror shape of the positive Mahler flow. This shows a non-differentiable point at $k=0$. If we plot the amoeba for $k=0$ as in Figure \ref{flow_amoeba}(b), we can see that it degenerates to two lines. In other words, the amoeba retracts to its spine in this case. By definition, $k=0$ gives another tropical limit. Although the details of phase structure in terms of amoeba for a non-Harnack curve is yet unclear, it could still be possible that $k=0$ has a solid-liquid transition.
\newpage

\linespread{0.9}\selectfont
\addcontentsline{toc}{section}{References}
\bibliographystyle{utphys}
\bibliography{ref}

\end{document}